\documentclass[english]{article}
\usepackage[T1]{fontenc}
\usepackage[latin1]{inputenc}
\usepackage{geometry}
\usepackage{float}
\usepackage{mathrsfs}
\usepackage{amsmath}
\usepackage{amssymb}
\usepackage{graphicx}

\usepackage[active]{srcltx}

\makeatletter

\floatstyle{ruled}
\newfloat{algorithm}{tbp}{loa}
\providecommand{\algorithmname}{Algorithm}
\floatname{algorithm}{\protect\algorithmname}

\usepackage{amsthm}

\usepackage{mathrsfs}

\usepackage{amsfonts}

\usepackage{epsfig}

\usepackage{bm}

\usepackage{mathrsfs}

\usepackage{enumerate}

\@ifundefined{definecolor}{\@ifundefined{definecolor}
 {\@ifundefined{definecolor}
 {\usepackage{color}}{}
}{}
}{}
\newcommand{\bbR}{\mathbb R}

\newtheorem{theorem}{Theorem}[section]
\newtheorem{lem}{Lemma}[section]
\newtheorem{rem}{Remark}[section]
\newtheorem{prop}{Proposition}[section]
\newtheorem{cor}{Corollary}[section]
\newcounter{hypA}
\newenvironment{hypA}{\refstepcounter{hypA}\begin{itemize}
  \item[({\bf A\arabic{hypA}})]}{\end{itemize}}
\newcommand{\Exp}{\mathbb{E}}
\newcommand{\Con}{X_{l_d(t_{k-1}(d)),j}^{\prime,i}}

\usepackage{babel}\date{}

\usepackage{babel}

\makeatother

\usepackage{babel}
\begin{document}

%+Title
\begin{center}

{\Large \textbf{On the Stability of Sequential Monte Carlo Methods in High Dimensions}}

\vspace{0.5cm}

BY ALEXANDROS BESKOS$^{1}$, DAN CRISAN$^{2}$ \& AJAY JASRA$^{3}$

{\footnotesize $^{1}$Department of Statistical Science, University College London, London WC1E 6BT, UK.}\\
{\footnotesize E-Mail:\,}\texttt{\emph{\footnotesize alex@stats.ucl.ac.uk}}\\
{\footnotesize $^{2}$Department of Mathematics,
Imperial College London, London, SW7 2AZ, UK.}\\
{\footnotesize E-Mail:\,}\texttt{\emph{\footnotesize d.crisan@ic.ac.uk}}\\
{\footnotesize $^{3}$Department of Statistics \& Applied Probability,
National University of Singapore, Singapore, 117546, SG.}\\
{\footnotesize E-Mail:\,}\texttt{\emph{\footnotesize staja@nus.edu.sg}}
\end{center}

\begin{abstract}
We investigate the stability of a Sequential Monte Carlo (SMC) method applied to the problem of sampling from a target distribution on $\mathbb{R}^d$ for large $d$. It is well known \cite{bengtsson,bickel,snyder} that using a \emph{single}
importance sampling step one produces an approximation for the target that deteriorates as the dimension $d$ increases, unless the number of Monte Carlo samples $N$ increases at an exponential rate in $d$.  We show that this degeneracy can be avoided by introducing a sequence of artificial targets, starting from a `simple' density and moving to the one of interest, using an  SMC method to sample from the sequence (see e.g.~\cite{chopin1,delm:06,jarzynski,neal:01}).
Using this class of SMC methods with a fixed number of samples, one can produce an approximation for which  the effective sample size (ESS)  converges to a random variable $\varepsilon_N$
as  $d\rightarrow\infty$ with $1<\varepsilon_{N}<N$.
The convergence is achieved with a computational cost proportional to $Nd^2$. If  $\varepsilon_N\ll N$, we can  raise  its value by introducing a number of resampling steps, say $m$ (where $m$ is independent of $d$). In this case,  ESS  converges to a random variable $\varepsilon_{N,m}$
as  $d\rightarrow\infty$ and $\lim_{m\to\infty}\varepsilon_{N,m}=N$.
Also, we show that the  Monte Carlo error for estimating a fixed dimensional marginal expectation  is of order $\frac{1}{\sqrt{N}}$
\emph{uniformly} in $d$. The results imply that, in high dimensions,  SMC algorithms can efficiently control the variability of the importance sampling weights and estimate fixed dimensional marginals at a cost which is less than exponential in $d$
and indicate that, in high dimensions, resampling leads to a reduction in the Monte Carlo error and increase in the ESS. 
\end{abstract}
%-Title
%

\section{Introduction}

Sequential Monte Carlo (SMC) methods can be described as a collection of techniques that approximate a sequence of distributions, known up-to a normalizing constant,  of increasing dimension. Typically, the complexity of these distributions is such that one cannot rely upon standard simulation approaches. SMC methods are applied in a wide variety of applications, including engineering, economics and biology, see \cite{doucet} and Chapter VIII in \cite{crisroz} for an overview.
They combine  importance sampling and resampling to approximate distributions. The idea is to introduce a sequence of proposal densities and sequentially simulate a collection of $N>1$ samples, termed particles, in parallel from these proposals. In most scenarios
it is not possible to use the distribution of interest as a proposal. Therefore, one must correct for the discrepancy between proposal
and target  via importance weights. In almost all cases of practical interest, the variance of these importance weights increases with algorithmic time (e.g.~\cite{kong}); this can, to some extent, be dealt with via resampling. This consists of sampling with replacement from the current samples using the weights and resetting them to $1/N$.
The variability of the weights is often measured by the effective sample size (\cite{liu})
and one often resamples when this drops below a threshold (dynamic-resampling).

There are a wide variety of convergence results for SMC methods, most of them concerned with the accuracy of the particle approximation of the distribution of interest as a function of $N$.
A less familiar context, related with this paper, arises in the case when the difference in the dimension of the consecutive densities becomes large. Whilst in filtering there are several studies on the stability of SMC as the time step grows (see e.g.\@ \cite{chopin,delmoral,delmoral_coal,douc,heine,kunsch}) they do not consider this latter scenario. In addition, there is a vast literature on the performance of high-dimensional Markov chain Monte Carlo (MCMC) algorithms e.g.~\cite{beskos,roberts1,roberts2}; our aim is to obtain a similar analytical understanding about the effect of dimension on SMC methods. The articles
 \cite{baer,bengtsson,bickel,snyder} have considered some  problems in this direction.
In \cite{bengtsson,bickel,snyder} the authors show that, for an i.i.d.~target, as the dimension of the state grows to infinity then one requires, for some stability properties, a number of particles which grows exponentially in dimension (or `effective dimension' in \cite{snyder}); the algorithm considered is standard importance sampling.
%%% ALEX: Empirically, \cite{baer} shows that when one introduces an accept/reject step (see \cite{delmoral}) then particle filtering methods can perform well.
We discuss these results below.

\subsection{Contribution of the Article}

We investigate the stability of an  SMC algorithm in high dimensions used to produce a sample from a sequence of probabilities on a common state-space. This problem arises in a wide variety of applications including many encountered in Bayesian statistics.
For some Bayesian problems  the posterior density can be very `complex', that is, multi-modal and/or with high correlations between certain variables in the target (`static' inference, see e.g.~\cite{jasra_static}). A commonly used idea is to introduce a simple distribution, which is more straightforward to sample from, and to interpolate between this distribution and the actual posterior by introducing intermediate distributions from which one samples sequentially. Whilst this problem departs from the standard ones in the SMC literature, it is possible to construct SMC methods to approximate this sequence; see \cite{chopin1,delm:06,jarzynski,neal:01}.
The methodology investigated here is applied in many practical contexts: financial modelling \cite{jasra_adaptive}, regression \cite{shaefer} and approximate Bayesian inference \cite{delmoralabc}.
In addition, high-dimensional problems are of practical importance and normally more challenging than their low dimensional counterparts.
The question we look at is whether such algorithms, as the dimension $d$ of the distributions increases, are stable in any sense. That is, whilst $d$ is \emph{fixed} in practice, we would like
identify the computational cost of the algorithm for large $d$, to ensure that the algorithm is stable. Within the SMC context described here, we quote the following statement made in \cite{bickel}:
\begin{quote}
`Unfortunately, for truly high dimensional systems, we conjecture that the number of intermediate steps would be prohibitively large and render it practically infeasible.'
\end{quote}
One of the objectives of this article is to investigate the above statement from a theoretical perspective. In the sequel we show that for a certain class of target densities:
\begin{itemize}
\item{The SMC algorithm analyzed, with computational cost
$\mathcal{O}(Nd^2)$ is stable. Analytically, we prove that ESS converges weakly to a non-trivial  random variable $\varepsilon_N$
as $d$ grows and the number of particles
\emph{is kept fixed}.
In addition, we show that the Monte Carlo error of the estimation of fixed dimensional marginals, for a fixed number of particles $N$ is of order $1/ \sqrt N$ \emph{uniformly} in $d$.
The algorithm can \emph{include} dynamic resampling at some particular deterministic times. In this case,  the algorithm will resample  $\mathcal{O}(1)$ times.
Our results indicate that estimates will improve when one resamples.}
\item{The dynamically resampling SMC algorithm (with stochastic times and some minor modifications) will, with probability greater than or equal to $1-M/\sqrt{N}$, where $M$ is a positive constant independent of $N$,
also exhibit these properties.
%This modification means only convergence of the ESS and error is relevant.
}
\item{Our results are proved for $\mathcal{O}(d)$ steps in the algorithm. If one takes $\mathcal{O}( d^{1+\delta})$
steps with any $\delta >0$, then ESS converges in probability to $N$ and the Monte Carlo error is the same as with i.i.d.~sampling. If $-1<\delta<0$ then ESS will go-to zero (Corollary \ref{cor:ess=N}). That is, $\mathcal{O}(d)$ steps are a critical order for the stability of the algorithm
in our scenario.}
\end{itemize}

%The first result is that ESS conerges weakly to a non-trivial limit as $d\rightarrow\infty$.
%(also
%the expected ESS converges to a constant in $(1,N)$). This is associated to the fact that the underlying, discrete-time, Feynman-Kac formula (\cite{delmoral}) will exist.
%It also allows one to show that a variety of quantities associated to SMC algorithms are well-behaved: for example the Monte Carlo error of fixed dimensional marginals. In particular, using the resampling-time construction of \cite{delmoral_resampling} we are able to establish the second point
%referring to algorithms actually implemented by practitioners.
Our results show that in high-dimensional problems, one is able to control the variability of the weights; this is a minimum requirement for applying the algorithm.
They also establish that one can estimate \emph{fixed} dimensional marginals even as the dimension $d$ increases.
The results help to answer the point of \cite{bickel} quoted above. In the presence of a quadratic cost and increasingly sophisticated hardware (e.g.~\cite{lee_gpu}) SMC methods are in fact applicable, in the static context, in high-dimensions.
To support this, \cite{jasra_adaptive} presents further empirical evidence of the results presented here. In particular, it is shown there that SMC techniques are algorithmically stable  for models of dimension over 1000 with computer simulations that run in just over 1 hour. Hence the SMC techniques analyzed here can certainly be used for high-dimensional \emph{static} problems.
The analysis of such methods for time-dependent applications (e.g.\@ filtering) is subject to further research.

When there is no resampling, the proofs of our results rely on martingale array techniques. To show that the algorithm is stable we establish a functional central limit  theorem (fCLT) under easily verifiable conditions, for a triangular array of non-homogeneous Markov chains. This allows one to establish  the convergence in distribution of ESS (as~$d$ increases). The result also demonstrates the dependence of the algorithm on a mixture of asymptotic variances (in the Markov chain CLT) of the non-homogeneous kernels.
%The case of resampling is dealt with by using the construction of \cite{delmoral_resampling}.
%Like the proof in that article, our results depend, crucially, on the fact that the effective sample size cannot coincide with the threshold when one chooses to resample; this requires some minor modifications, which do not overly change the algorithm.

%We conclude with the following two caveats: whilst the results presented here do not fully guarantee that the class of SMC algorithms we consider here will have a fast rate of convergence (certainly not for the estimation of functionals that grow with $d$; see also e.g.\@ \cite{chopin1} for some discussion on ESS), they do establish that the algorithms will not collapse with regards to some criteria. Similarly,
%the convergence of the ESS without resampling does not imply that one should avoid resampling.
%Indeed, it can only be determined whether one should resample given a key quantity which is not explicit. This quantity can be numerically approximated, but its approximation is rather difficult (see e.g.~\cite{andrieu_adap_rev} and the references therein). Nonetheless, we believe our results are an important first step in the analysis of high-dimensional SMC algorithms.

\subsection{Structure of the Article}

%This article is structured as follows.
In Section \ref{sec:SMC} we discuss the SMC algorithm of interest and the class of target distributions we consider. In
Section~\ref{sec:main_result} we
show that ESS converges in distribution to a non-trivial random variable as $d\rightarrow\infty$
when the algorithm does not resample.
We also show that the Monte Carlo error of the estimation of fixed dimensional marginals, for a fixed number of particles $N$,
has an upper bound of the form  $M/\sqrt{N}$, where $M$ is independent of $d$.
We address the issue of resampling in Section \ref{sec:resampling_stability}, where it is shown that
as $d\rightarrow\infty$ any dynamically resampling SMC algorithm, using the deterministic ESS
(the expected ESS with one particle)
will resample $\mathcal{O}(1)$ times
and also exhibit convergence of the ESS and Monte Carlo error. In addition,
any dynamically resampling SMC algorithm, using the empirical ESS
(with some modification) will, with high probability, display the
same convergence of the ESS and Monte Carlo error.
In Section \ref{sec:examples} we verify the involved assumptions for a particular example. Finally, we conclude in Section \ref{sec:summary}
with some remarks on  $\mathcal{O}(d)$ steps being a critical order and ideas for future work.
Proofs are collected in the Appendix.

\subsection{Notation}\label{sec:notations}

Let $(E,\mathscr{E})$ be a measurable space and $\mathscr{P}(E)$
be the set of probability measures on $(E,\mathscr{E})$.
For a given function $V:E\mapsto[1,\infty)$ we denote by $\mathscr{L}_V$
the class of functions $f:E\mapsto\bbR$ for which
$$
|f|_V := \sup_{x\in E}\frac{|f(x)|}{V(x)} < +\infty\ .
$$
For two Markov kernels, $P$ and $Q$ on  $(E,\mathscr{E})$, we define the $V$-norm:
$$
|||P-Q|||_{V} := \sup_{x\in E}\frac{\sup_{|f|\leq V}|P(f)(x)-Q(f)(x)|}{V(x)}\ ,
$$
with $P(f)(x):=\int_E P(x,dy)f(y)$.
The notation
$$
\|P(x,\cdot)-Q(x,\cdot)\|_{V} := \sup_{|f|\leq V}|P(f)(x)-Q(f)(x)|
$$
is also used. For $\mu\in\mathscr{P}(E)$ and $P$ a Markov kernel on
$(E,\mathscr{E})$, we adopt the notation $\mu P(f) :=$ $\ \int_{E} \mu(dx) P(f)(x)$. In addition, $P^n(f)(x) := \int_{E^{n-1}}P(x,dx_1)P(x_1,dx_2)\times\cdots\times P(f)(x_{n-1})$.
$\mathscr{B}(\mathbb{R})$ is used to denote the class of Borel sets and $\mathcal{C}_b(\mathbb{R})$ the class of bounded continuous
$\mathscr{B}(\mathbb{R})$-measurable functions. Denote $\|f\|_{\infty}=\sup_{x\in \mathbb{R}}|f(x)|$. We will also define the $\mathbb{L}_{\varrho}$-norm, $\|X\|_{\varrho} =
\Exp^{1/\varrho}\,|X|^{\varrho}$, for
$\varrho\ge 1$ and denote by $\mathbb{L}_{\varrho}$ the space of random variables such that $\|X\|_{\varrho}<\infty$. For $d\geq 1$,
$\mathcal{N}_{d}(\mu,\Sigma)$ denotes the $d$-dimensional normal distribution with mean $\mu$
and covariance $\Sigma$; when $d=1$ the subscript is dropped. For any vector $(x_1,\dots,x_p)$, we denote by $x_{q:s}$
the vector $(x_q,\dots,x_s)$ for any $1\leq q\leq s\leq p$.
Throughout $M$ is used to denote a constant whose meaning may change, depending upon the context; any (important) dependencies are written as $M(\cdot)$.

\section{Sequential Monte Carlo}
\label{sec:SMC}

We wish to sample from a target distribution with density $\Pi$ on $\bbR^{d}$
with respect to Lebesgue measure, known up to a normalizing constant.
We introduce a sequence of `bridging' densities which start from an easy to sample target and  evolve toward $\Pi$. In particular, we will consider (e.g.\@ \cite{delm:06}):
\begin{equation}
\label{eq:aux}
\Pi_n(x) \propto \Pi(x)^{\phi_n}\ , \quad x\in \mathbb{R}^d\ ,
\end{equation}
for
$
0<\phi_0<\cdots < \phi_{n-1}< \phi_n < \cdots<\phi_p=1.
$
The effect of exponentiating with the small constant ${\phi_0}$ is that $\Pi(x)^{\phi_0}$ is much `flatter'
than~$\Pi$.
%Below, we use the short-hand $\Gamma_n$ to denote un-normalized densities associated to $\Pi_n$.
Other
choices of bridging  densities are possible and are discussed in the sequel.

One can sample from the sequence of densities using an SMC sampler,
which is, essentially, a Sequential Importance Resampling (SIR) algorithm or particle filter that targets the sequence of densities:
\begin{equation*}
\widetilde{\Pi}_n(x_{1:n}) = \Pi_n(x_n) \prod_{j=1}^{n-1} L_j(x_{j+1},x_j)
\end{equation*}
% %
with domain $(\mathbb{R}^d)^n$ of dimension that increases with $n=1,\ldots,p$;
here, $\{L_n\}$ is a sequence of artificial backward Markov kernels that can, in principle, be arbitrarily selected.
The work in \cite{delm:06} motivates the selection of $\{L_n\}$ and characterizes the optimal kernel, in terms of minimizing the variance of the importance weights for SMC.
Let $\{K_n\}$ be a sequence of Markov kernels of invariant density $\{\Pi_n\}$ and $\Upsilon$ a distribution; assuming the weights appearing below are well-defined Radon Nikodym derivatives, the SMC algorithm
we will ultimately explore is the one defined in Figure \ref{tab:SMC}. It arises
when the backward Markov kernels $L_n$ are chosen as follows:
\begin{equation*}
L_n(x,x') = \frac{\Pi_{n+1}(x') K_{n+1}(x',x)}{\Pi_{n+1}(x)}\ .
\end{equation*}
With no resampling, the algorithm coincides with the annealed importance sampling in \cite{neal:01}.
%and arises when the backward Markov kernel $L_n$ is chosen as follows:
%%
%\begin{equation*}
%L_n(x,x') = \frac{\Pi_{n+1}(x') K_{n+1}(x',x)}{\Pi_{n+1}(x)}\ .
%\end{equation*}
%%
%
\begin{figure}[!h]
\begin{flushleft}
\medskip
\hrule
\medskip
%\textit {SMC:}
{\itshape
\begin{enumerate}
\item[\textit{0.}] Sample $X_0^1,\dots X_0^N$ i.i.d.\@ from $\Upsilon$ and compute the weights for each particle $i\in\{1,\dots,N\}$:
\begin{equation*}
w_0(x_0^i) = \frac{\Pi_0(x_0^i)}{\Upsilon(x_0^i)}\ .
\end{equation*}
Set $n=1$ and $l=0$.
\vspace{0.18cm}
\item[\textit{1.}]  If $n\leq p$, for each $i$
sample $X_n^i\mid x_{n-1}^i$ from $K_n$ and calculate the weights
\begin{equation*}
w_n(x_{l:n-1}^i)  =
%\frac{\Gamma_n(x_{n}^i)L_{n-1}(x_{n}^i,x_{n-1}^i)}{\Gamma_{n-1}(x_{n-1}^i)K_n(x_{n-1}^i,x_n^i)}\,
%w_{n-1}(x_{0:n-2}^i)
\frac{\Pi_n(x_{n-1}^i)}{\Pi_{n-1}(x_{n-1}^i)}\,w_{n-1}(x_{l:n-2}^i)\
\end{equation*}
with the convention $x_{0:-1}^i\equiv x_{0}^{i}$.
Calculate the Effective Sample Size (ESS):
\begin{equation}
\label{eq:ess_def}
\textrm{ESS}_{(l,n)}(N) := \frac{\left(\sum_{i=1}^N w_n(x_{l:n-1}^i)\right)^2}{\sum_{i=1}^N w_n(x_{l:n-1}^i)^2} \ .
\end{equation}
If $ESS_{(l,n)}(N)<a$: \\
\hspace{0.3cm} resample $x_{n}^{1},\ldots x_n^{N}$ according to their normalised
weights
\begin{equation}
\label{eq:normed}
w_n(x_{l:n-1}^i)/\sum_{j=1}^{N}w_n(x_{l:n-1}^j)\ ;
\end{equation}
\hspace{0.3cm} set $l=n$; \\
\hspace{0.3cm} re-initialise the weights by setting $w_{n}(x_{l:n-1}^i)\equiv 1$, $1\le i \le N$;\\
\hspace{0.3cm} let $x_{n}^{1},\ldots x_n^{N}$ now denote the resampled particles.\\
Set $n=n+1$. \\
Return to the start of Step 1.

%\vspace{0.18cm}
%
\end{enumerate} }
\medskip
\hrule
\medskip
\end{flushleft}
\vspace{-0.5cm}
\caption{The SMC algorithm analyzed in this article.}
\label{tab:SMC}
\end{figure}
For simplicity, we will henceforth  assume  that $\Upsilon\equiv \Pi_0$. % see \cite{whiteley} for some discussion.
It is remarked that, due to the results of \cite{bengtsson,bickel,snyder}, it appears that the cost of the population
Monte Carlo method of \cite{cappe} would increase exponentially with the dimension;
instead we will show that the `bridging' SMC sampler framework above will be of smaller cost.

ESS defined in (\ref{eq:ess_def}) is typically used to quantify the quality of SMC approximations associated to systems of weighted particles. It is a number between $1$ and $N$, and in general the larger the value, the better the approximation. Resampling is often performed when ESS falls below some pre-specified threshold  such as $a=N/2$. The  operation of resampling consists of sampling with replacement from the current set of particles via the normalized weights in (\ref{eq:normed}) and resetting the (unnormalized) weights to $1$.
There is a wide variety of resampling techniques and we refer the reader to \cite{doucet} for details; in this article we only consider the multinomial method just described above.

\subsection{Framework}
\label{sec:frame}
We will investigate the stability of the SMC algorithm in Figure \ref{tab:SMC} as $d\rightarrow\infty$.
To obtain analytical results we will need to simplify the structure of the algorithm (similarly
to MCMC results in high dimensions in e.g.\@ \cite{bedard, beskos, roberts1,roberts2}). In particular, we will consider
an i.i.d.~target:
\begin{equation}
 \label{eq:target}
\Pi(x) = \prod_{j=1}^{d}\pi(x_j)\ ; \quad \pi(x_j) =  \exp\{g(x_j)\}\ , \quad x_j\in \bbR\ ,
\end{equation}
for some $g:\bbR\mapsto \bbR$. In such a case all bridging densities are also i.i.d.:
\begin{equation*}
\Pi_{n}(x) \propto \prod_{j=1}^{d}\pi_{n}(x_j)\ ;\quad \pi_n(x_j) \propto \exp\{\phi_n\,g(x_j)\} \ .
\end{equation*}
It is remarked that this assumption is made for mathematical convenience (clearly, in an i.i.d.\@ context one could use standard  sampling schemes). Still, such a context allows for a rigorous mathematical treatment; at the same time
(and similarly to corresponding extensions of results for MCMC algorithms in high dimensions) one would expect that
the analysis we develop in this paper for i.i.d.\@ targets will also be relevant in practice for more general scenarios; see Section \ref{sec:summary} for some discussion.
A further assumption that will facilitate the mathematical analysis is to apply independent kernels along the different co-ordinates. That is, we will assume:
\begin{equation}
\label{eq:kernel}
K_{n}(x,dx') =  \prod_{j=1}^d k_{n}(x_j,dx_j') \ ,
\end{equation}
where each transition kernel $k_n(\cdot, \cdot)$ preserves
$\pi_n(x)$;
that is, $\pi_{n}k_{n}=\pi_{n}$.
Clearly, this also implies that $\Pi_n K_n = \Pi_n$.

 The stability of ESS will be investigated as $d\rightarrow \infty$: first
without resampling and then with resampling.
We study the case when one selects cooling constants $\phi_n$ and $p$ as below:
\begin{equation}
\label{eq:tune}
p = d\  ; \quad \phi_{n} (=\phi_{n,d})= \phi_0 + \frac{n(1-\phi_0)}{d}\ ,\quad  0\le n\le  d \ ,
\end{equation}
with $0<\phi_0<1$ given and fixed with respect to $d$. It will be shown that such a selection will indeed
provide a `stable' SMC algorithm as $d\rightarrow \infty$.
Note that $\phi_0>0$ as we will be concerned with probability densities on non-compact spaces.
\begin{rem}
Since $\{\phi_n\}$ will change with $d$, all elements of our SMC algorithm will also
depend on~$d$. We  use the double-subscripted notation $k_{n,d}$, $\pi_{n,d}$ when needed to emphasize the dependence of $k_n$ and $\pi_n$ on $d$, which ultimately, depend on $n,d$ through $\phi_{n,d}$.
Similarly, we will sometimes write $X_n(d)$, or $x_n(d)$, for the Markov chain involved in the specification of the SMC algorithm.
\end{rem}

\begin{rem}
\label{rem:switch}
Although the algorithm runs in discrete time, it will be convenient for the presentation of our results that we consider the successive steps
of the algorithm as placed on the continuous time interval $[\phi_0,1]$, incremented by
the annealing discrepancy $(1-\phi_0)/d$. We will use the mapping
\begin{equation}
\label{eq:switch}
l_d(t)=\bigl\lfloor d\bigl(\tfrac{t-\phi_0}{1-\phi_0}\bigr)\bigr\rfloor
\end{equation}
to switch between continuous time and discrete time.
Related to the above, it will be convenient to
consider the continuum of invariant densities and kernels on the whole of the time interval
$[\phi_0,1]$. So, we will set:
\begin{equation*}
\pi_s(x) \propto \pi(x)^{s} = \exp\{s\, g(x)\}\ , \quad s \in [\phi_0,1] \ .
\end{equation*}
That is, we will use the convention $\pi_n\equiv \pi_{\phi_n}$ with the subscript on the left running
on the set $\{1,2,\ldots,d\}$.
% and the one on the right on the grid in $[\phi_0,1]$:
% %
% \begin{equation}
% \label{eq:grid}
% G_d=\{\phi_0+n\,(1-\phi_0)/d\,;\,n=1,\ldots, d\} \ .
% \end{equation}
%
Accordingly, $k_s(\cdot,\cdot)$, with $s\in (\phi_0,1]$, will denote the transition kernel preserving $\pi_s$.
%This is an analogue to the connection between bridge sampling and path sampling in \cite{gelman}.
\end{rem}

\subsection{Conditions}
We state the conditions under which we derive our results.
 Throughout, we set $k_{\phi_0}\equiv\pi_{\phi_0}$ and $(E,\mathscr{E})=(\mathbb{R},\mathscr{B}(\mathbb{R}))$.
We assume that   $g(\cdot)$ is an upper bounded function.
In addition, we make the following assumptions for the continuum of kernels/densities:

\begin{hypA}
\label{hyp:A}
{\em Stability of $\{k_{s}\}$.}
\renewcommand{\labelitemii}{}
\begin{enumerate}[(i)]
\item {\em (One-Step Minorization).}\label{hyp:constraint_lambda_b_d}  We assume that there exists a
set $C\in \mathscr{E}$, a constant $\theta\in(0,1)$ and some $\nu\in \mathscr{P}(E)$ such that for each
$s\in(\phi_0,1]$ the set  $C$ is $(1,\theta,\nu)-$small with respect to $k_s$.
\item \label{hyp:K_driftcondition} (\emph{One-step Drift
Condition}). There exists $V:E\mapsto[1,\infty)$ with $\lim_{|x|\rightarrow \infty}V(x) = \infty$,
constants $\lambda<1$, $b<\infty$, and $C \in \mathscr{E}$ as specified in (i) such that for any $x \in E$
and $s\in(\phi_0,1]$:
\begin{equation*}
k_s\,V(x) \leq \lambda\,V(x) + b\,\mathbb{I}_{\,C}(x) \, .
\end{equation*}
In addition $\pi_{\phi_0}(V)<\infty$.
\item \label{hyp:level} (\emph{Level Sets}). Define $C_c=\{x:V(x)\leq c\}$ with $V$ as in (\ref{hyp:K_driftcondition}).
Then there exists a $c\in(1,\infty)$ such that for every $s\in(\phi_0,1)$, $C_c$ is a $(1,\theta,\nu)-$small set with respect to~$k_s$. In addition, condition (ii) holds for  $C=C_c$, and $\lambda$, $b$ (possibly depending on $c$) such that $\lambda + b/(1+c)<1$.
\end{enumerate}
\end{hypA}

\begin{hypA}
\label{hyp:B}
{\em Perturbations of $\{k_{s}\}$.} \vspace{0.3cm} \\
There exists an $M<\infty$ such that
for any $s,t\in(\phi_0,1]$
$$
|||k_s-k_t|||_V  \leq M\,|s-t|\ .
$$
\end{hypA}
The statement that $C$ is $(1,\theta,\nu)-$small w.r.t.\@ to $k_s$ means that $C$ is an one-step small set for the Markov kernel, with minorizing distribution $\nu$ and parameter $\theta\in(0,1)$ (see e.g.~\cite{meyn}).

Assumptions like (A\ref{hyp:A}) are fairly standard in the literature on adaptive MCMC (e.g.~\cite{andrieu1}).
Note though that the context in this paper is different.
For adaptive MCMC one typically has that the kernels will eventually converge to some limiting kernel.
Conversely, in our set-up, the $d$ bridges (resp.~kernels) in between $\pi_0$ (resp.~$k_{0}$) and $\pi_d$
(resp.~$k_{d}$) will effectively make up
a continuum of densities $\pi_s$ (resp.~kernels $k_{s}$), with $s\in[\phi_0,1]$, as $d$ grows to infinity.
The second assumption above differs from standard  adaptive MCMC  but will be verifiable in real contexts. Note that one could maybe relax our assumptions to, e.g.~sub-geometric ergodicity versus geometric ergodicity, at the cost of an increased level of complexity in the proofs.
It is also remarked that the assumption that $g$ is upper bounded is only
used in Section \ref{sec:resampling_stability}, when controlling the resampling times.
The assumptions  adopted in this article are certainly not weak, but still are very close to the weakest assumptions adopted in state-of-the-art
research on stability of SMC, see \cite{whiteley,whiteley2,whiteley1}.

\section{The Algorithm Without Resampling}
\label{sec:main_result}

We will now consider the case when we omit the resampling steps in the specification of our SMC algorithm in Figure \ref{tab:SMC}.
Critically, due to the i.i.d.~structure of the bridging densities $\Pi_n$ and the kernels $K_n$
each particle will evolve according to a $d$-dimensional
Markov chain $X_n$ made up of $d$ i.i.d.~one-dimensional Markov chains $\{X_{n,j}\}_{n=0}^{d}$, with
$j$  the co-ordinate index, evolving under the kernel $k_n$. Also, all particles move independently.

We consider first the stability of the terminal ESS, i.e.,
\begin{equation}
\label{eq:term_ess}
\textrm{ESS}_{(0,d)}(N)= \frac{\left(\sum_{i=1}^N w_d(x_{0:d-1}^i)\right)^2}{\sum_{i=1}^N w_d(x_{0:d-1}^i)^2}
\end{equation}
where, due to the i.i.d.~structure and our selection of $\phi_n$'s in (\ref{eq:tune}), we can rewrite:
\begin{equation}
\label{eq:weight}
w_d(x_{0:d-1})=  \exp\bigg\{\frac{(1-\phi_0)}{d}\sum_{j=1}^d \sum_{n=1}^{d}
g(x_{n-1,j})\bigg\}\ .
\end{equation}
It will be shown that under our set-up $\textrm{ESS}_{(0,d)}(N)$ converges in distribution to a non-trivial
variable and analytically characterise the limit; in particular we will have
%
%\begin{equation*}
$\lim_{d\rightarrow\infty}  \Exp\,[\,\textrm{ESS}_{(0,d)}(N)\,]  \in (1,N)$.
%\end{equation*}
%

% A key question of interest is the stability of the expected ESS of the SMC algorithm specified in the previous section as $d$ grows. More precisely, when there is no resampling, we are interested in the behaviour of
% %
% \begin{align*}
% I(d):&=\mathbb{E}\,\big[\,\textrm{ESS}_{(0,p)}(N)\,\big]\\
% &=\int_{\bbR^{d N p}} \prod_{i=1}^N\bigg\{  \Pi_0(dx_{0}^i)
% \prod_{n=1}^{p-1} K_{n}(x_{n-1}^i,dx_n^i)\bigg\}
% \bigg[\frac{[\sum_{i=1}^Nw_p(x_{0:p-1}^i)]^2}{\sum_{i=1}^Nw_p(x_{0:p-1}^i)^2}\bigg]\ .% dx_0^{1:N} \
% \end{align*}
% %
% Due to the i.i.d.~structure of the target in (\ref{eq:target}) and the specification of the bridging densities in (\ref{eq:aux}),
% the weights can be expressed as:
% %
% \begin{equation*}
% w_p(x_{0:p-1})  =
% \prod_{n=1}^{p} \frac{\Gamma_{n}(x_{n-1})}{\Gamma_{n-1}(x_{n-1})}
% = \exp\bigg\{\sum_{j=1}^d \sum_{n=1}^{p}(\phi_{n}-\phi_{n-1})\,
% g(x_{n-1,j})\bigg\}\ .
% \end{equation*}
% %
%
%
%
%
% Under the construction above, we can rewrite:
% \begin{equation}
% \label{eq:weight}
%  w_p(x_{0:p-1})\equiv  w_d(x_{0:d-1})=  \exp\bigg\{\frac{(1-\phi_0)}{d}\sum_{j=1}^d \sum_{n=1}^{d}
% g(x_{n-1,j})\bigg\}\ .
% \end{equation}
% We will show that under this selection of $\phi_n$'s:
% \begin{equation*}
% \lim_{d\rightarrow\infty}  I(d)  \in (1,N]
% \end{equation*}
% and analytically characterise the non-trivial limit.

\subsection{Strategy of the Proof}

%\section{Statement of Main Result}

To demonstrate that the selection of the cooling sequence $\phi_n$ in (\ref{eq:tune}) will control ESS we  look at the behaviour of the sum:
\begin{equation}
\label{eq:or}
\frac{1-\phi_0}{d}\,\sum_{j=1}^d\sum_{n=1}^d
g(x_{n-1,j})
\end{equation}
appearing in the expression for the weights, $w_{d}(x_{0:d-1})$, in (\ref{eq:weight}).
Due to the nature of the expression for ESS one can re-center, so we can consider the limiting properties of:
\begin{equation}
\label{eq:key_quant}
\alpha(d) = \frac{1}{\sqrt{d}}\,\sum_{j=1}^d
\overline{W}_{j}(d)
\end{equation}
differing from (\ref{eq:or}) only in terms of a constant (the same for all particles),
where we have defined:
\begin{equation}
\label{eq:Wbar}
\overline{W}_{j}(d) = W_{j}(d) -\mathbb{E}\,[\,W_{j}(d)\,]
\end{equation}
and
\begin{equation}
\label{eq:W}
\quad W_{j}(d)= \frac{1-\phi_0}{\sqrt{d}}\,\sum_{n=1}^d \big\{ g(x_{n-1,j})-\pi_{n-1}(g)\big\} \ .
\end{equation}
As mentioned above,  the dynamics of the involved random variables correspond to those of $d$ independent scalar non-homogeneous Markov chains $\{X_{n,j}\}_{n=0}^{d}\equiv \{X_{n,j}(d)\}_{n=0}^{d}$ of initial position  $X_{0,j}\sim\pi_0$ and evolution according to the transition kernels $\{k_{n}\}_{1\leq n\leq d}$.
We will proceed as follows. For any fixed $d$ and co-ordinate $j$, $\{X_{n,j}\}_{n=0}^{d}$ is a non-homogeneous Markov chain of total length $d+1$. Hence, for fixed $j$, $\{X_{n,j}\}_{d,n}$
constitutes an array of non-homogeneous Markov chains.
We will thus be using the relevant theory to prove a central limit theorem (via a fCLT) for $\overline{W}_j(d)$ as $d\rightarrow\infty$.
Then, the independency of the $\overline{W}_{j}(d)$'s over $j$  will essentially provide a central limit theorem for $\alpha(d)$ as $d\rightarrow\infty$.
%It is remarked that one can treat the double-sum in (\ref{eq:or}) in one go, as in \cite[pp.~295-297]{delmoral}, but the end result will be the same.

\subsection{Results and Remarks for ESS}\label{sec:rem_res}

Let $t\in[\phi_0,1]$ and recall the definition of $l_d(t)$ in (\ref{eq:switch}).
We define:
\begin{equation*}
S_{t} = \frac{1-\phi_0}{\sqrt{d}}\sum_{n=1}^{l_{d}(t)}\{ {g}(X_{n-1,j}) -\pi_{n-1}({g})\}\ .
\end{equation*}
Note that $S_1\equiv W_{j}(d)$.
Our fCLT considers the continuous linear interpolation:
\begin{equation*}
s_{d}(t) = S_{t} + \bigg(d\,\frac{t-\phi_0}{1-\phi_0}-l_d(t)
\bigg)[S_{t^{+}}-S_{t}] \ ,
\end{equation*}
where we have denoted
\begin{equation*}
S_{t^{+}} = \frac{1-\phi_0}{\sqrt{d}}\sum_{n=1}^{l_{d}(t)+1}\{ {g}(X_{n-1,j}) -\pi_{n-1}({g})\}\ .
\end{equation*}

\begin{theorem}[fCLT]
\label{th:fCLT}
Assume (A\ref{hyp:A}(i)(ii), A\ref{hyp:B}) and that $g\in\mathscr{L}_{V^{r}}$ for some $r\in[0,\tfrac{1}{2})$.
%Let $k_d(t)=\lfloor d\big(\frac{t-\phi_0}{1-\phi_0}\big)\rfloor$,
%$t\in[\phi_0,1]$,
Then:
\begin{equation*}
s_{d}(t)
\Rightarrow \mathcal{W}_{\sigma_{\phi_0:t}^2} \ ,
\end{equation*}
where $\{\mathcal{W}_t\}$ is a Brownian motion and
\begin{equation}
\label{eq:sigma}
\sigma_{\phi_0:t}^2 = (1-\phi_0)\int_{\phi_0}^{t} \pi_u\,\big(\,\widehat{g}_u^2 - k_u(\widehat{g}_u)^2\,\big)du\ ,
\end{equation}
with $\widehat{g}_u(\cdot)$ the unique solution of the Poisson equation:
\begin{equation}
\label{eq:poisson}
g(x) -  \pi_u(g) =  \widehat{g}_u(x) - k_u(\widehat{g}_u)(x) \ .
\end{equation}
In particular,
%\begin{equation*}
$W_{j}(d)  \Rightarrow \mathcal{N}(0,\sigma_{\star}^2)$
%\end{equation*}
with $\sigma^2_{\star}=\sigma_{\phi_0:1}^2$.
\label{th:clt}
\end{theorem}

%\begin{rem}
%The quantity $\sigma^2_{\star}$ above when divided by $(1-\phi_0)^2$ is a continuous
%average of the asymptotic variances of the $\{k_s\}$ in the Markov chain CLT (e.g.~\cite{meyn}).
%It will be key in understanding the dependence of the algorithm on the underlying kernels, as well as understanding when to resample (if at all).
%\end{rem}
We will now need the following result on the growth of $W_{j}(d)$.
\begin{lem}
\label{lem:growth_main}
Assume (A\ref{hyp:A}(i)(ii), A\ref{hyp:B}) and that $g\in\mathscr{L}_{V^{r}}$ for some $r\in[0,\tfrac{1}{2})$.
Then, there exists $\delta>0$ such that:
\begin{equation*}
\sup_d\,\mathbb{E}\,[\,|W_{j}(d)\,|^{2+\delta}\,] < \infty \ .
\end{equation*}
\end{lem}
\begin{proof}
This follows from the decomposition in Theorem \ref{th:decompose} and the following inequality:
$$
\mathbb{E}\,[\,|W_{j}(d)|^{2+\delta}\,] \leq \bigg(\frac{1}{\sqrt{d}}\bigg)^{2+\delta}M(\delta)
\big(\, \mathbb{E}\,[\,|M_{0:d-1}|^{2+\delta}\,]+
\mathbb{E}\,[\,|R_{0:d-1}|^{2+\delta}\,]
\,\big)\ .
$$
Applying the growth bounds in Theorem \ref{th:decompose} we get that
the remainder term $\mathbb{E}\,[\,|R_{0:d-1}|^{2+\delta}\,]$ is controlled as $\pi_{\phi_0}(V^r)<\infty$ (due to $r\in[0,\tfrac{1}{2})$).
The martingale array term  $\mathbb{E}\,[\,|M_{0:d-1}|^{2+\delta}\,]$ is upper bounded by $Md^{(2+\delta)/2}$, which allows us to conclude.
\end{proof}

One can now obtain the general result.
\begin{theorem}\label{theo:main_result}
Assume (A\ref{hyp:A}(i)(ii), A\ref{hyp:B}). Suppose also
that $g\in\mathscr{L}_{V^r}$ for some $r\in[0,\tfrac{1}{2})$.
Then, for any fixed $N>1$, $\textrm{\emph{ESS}}_{(0,d)}(N)$ converges in distribution to
$$
\varepsilon_N := \frac{[\sum_{i=1}^N e^{X_i}]^2}{\sum_{i=1}^N e^{2 X_i}}
$$
where $X_i \stackrel{i.i.d.}{\sim}\mathcal{N}(0,\sigma^2_{\star})$ for $\sigma^2_{\star}$ specified in Theorem \ref{th:fCLT}. In particular,
\begin{equation}
\lim_{d\rightarrow\infty}
\mathbb{E}\,\big[\,\textrm{\emph{ESS}}_{(0,d)}(N)\,\big]
=\mathbb{E}\bigg[\frac{[\sum_{i=1}^N e^{X_i}]^2}{\sum_{i=1}^N e^{2 X_i}}\bigg].
\label{eq:limitingess}
\end{equation}
\end{theorem}

\begin{proof}
We will prove that $\alpha(d)$, as defined in \eqref{eq:key_quant}, converges in distribution to $\mathcal{N}(0,\sigma^2_{\star})$.
The argument is standard: it suffices to check that the random variables $\overline{W}_j(d)$, $j=1,...,d$, satisfy the Lindeberg condition and that their second moments converge (see e.g.\@ an adaptation of Theorem~2 of \cite[pp.334]{shiryaev}). To this end, note that $\{\overline{W}_j(d)\}_{d,j}$ form a triangular array of independent variables of zero expectation
across each row. Let
\begin{equation*}
S_{d}^2=\frac{1}{d}\,
\sum_{j=1}^d\mathbb{E}\,[\,\overline{W}_j(d)^2\,]\equiv \mathbb{E}\,[\,\overline{W}_1(d)^2\,]\ ,
\end{equation*}
the last equation following from $\overline{W}_j(d)$ being i.i.d.\@ over $j$. Now, Theorem \ref{th:clt}
gives that $W_1(d)$ converges in distribution to $\mathcal{N}(0,\sigma^{2}_{\star})$ for
$d\rightarrow\infty$. Lemma \ref{lem:growth_main} implies that (e.g.~Theorem 3.5 of \cite{bill:99})
also the first and second moments of ${W}_1(d)$ converge to $0$ and $\sigma^{2}_{\star}$ respectively;
we thus obtain:
\begin{equation}
\label{eq:first_cond}
\lim_{d\rightarrow\infty} S_d^2
=\sigma_{\star}^2 \ .
\end{equation}
We consider also the Lindeberg condition, and for each $\epsilon>0$ we have:
\begin{equation}
\lim_{d\rightarrow\infty}\frac{1}{d}\sum_{j=1}^d \,
\Exp\,[\,\overline{W}_j(d)^2\,\mathbb{I}_{|\overline{W}_j(d)|\geq \epsilon \sqrt{d}}\,]= 0
\label{eq:second_cond}
\end{equation}
a result directly implied again from Lemma \ref{lem:growth_main}. Therefore, by Theorem 2 of \cite[pp.334]{shiryaev},
$\alpha(d)$ converges in distribution to $\mathcal{N}(0,\sigma^2_{\star})$.
In particular we have proved that:
\begin{equation*}
(\alpha_1(d),\dots,\alpha_N(d)) \Rightarrow \mathcal{N}_N(0,\sigma^2_{\star} I_N)\ ,
\end{equation*}
where the subscripts denote the indices of the particles.
The result now follows directly after noticing that
\begin{equation*}
\mathrm{ESS}_{(0,d)}=\frac{[\sum_{i=1}^N e^{\alpha_i(d)}]^2}{\sum_{i=1}^N e^{2\alpha_i(d)}}
\end{equation*}
and the mapping $(\alpha_1,\alpha_2,\ldots,\alpha_N)\mapsto
 \frac{[\sum_{i=1}^N e^{\alpha_i}]^2}{\sum_{i=1}^N e^{2\alpha_i}}$ is bounded and continuous.
\end{proof}
\subsection{Monte Carlo Error}

We have shown that the choice of bridging steps as in (\ref{eq:tune}) stabilises  ESS in
high dimensions.
The error in the estimation of expectations, which can be of even more practical interest than ESS, is now considered.
In particular we look at expectations associated with finite-dimensional marginals of the target distribution.
Recall the definition of the weight of the $i$-th particle $w_d(x_{0:d-1}^{i})$ from (\ref{eq:weight}),
for $1\le i\le N$.
In order to consider the Monte Carlo error, we use the result below, which is of some interest in its own right.
\begin{prop}
\label{prop:marginals}
Assume (A\ref{hyp:A}(i)(ii), A\ref{hyp:B}).
 and let $\varphi\in\mathscr{L}_{V^r}$ for $r\in[0,1]$.
Then we have:
\begin{equation*}
\lim_{d\rightarrow \infty}|\,\mathbb{E}\,[\,\varphi(X_{d,1})\,]-\pi(\varphi)\,| = 0\ .
%X_{d,1} \Rightarrow  Z\ , \quad Z\sim\pi \ .
\end{equation*}
\end{prop}

\begin{proof}
This follows from Proposition \ref{prop:prop} in the Appendix when choosing time sequences $s(d)\equiv\phi_0$
and $t(d)\equiv 1$.
\end{proof}

\begin{rem}
The above result is interesting as it suggests one can run an alternative algorithm that just samples a collection of independent  particles
through a grid of values of the annealing parameter and average the values of the function of interest. However, it is not clear how such an algorithm can be validated in practice (that is how many steps one should take for a finite time algorithm) and is of interest
in the scenario where one fixes $d$ and allows the time-steps to grow; see \cite{whiteley}.
In our context, we are concerned with the performance
of the estimator that one would use for fixed $d$ (and hence a finite number of steps in practice) from the SMC sampler in high-dimensions; it is not at all clear
a-priori that this will stabilize with a computational cost $\mathcal{O}(Nd^2)$ and if it does, how the error behaves.
\end{rem}

The Monte Carlo error result now follows; recall $\|\cdot\|_{\varrho}$ is defined in Section \ref{sec:notations}:

\begin{theorem}\label{theo:mc_error}
Assume (A\ref{hyp:A}(i)(ii), A\ref{hyp:B}) with $g\in\mathscr{L}_{V^r}$ for some $r\in[0,\tfrac{1}{2})$. Then
for any $1\leq \varrho <\infty$
there exists a constant $M=M(\varrho)<\infty$ such that for any $N\geq 1$,  $\varphi\in\mathcal{C}_b(\mathbb{R})$
\begin{equation*}
\lim_{d\rightarrow\infty}
\bigg\|\sum_{i=1}^N
\frac{w_d(X_{0:d-1}^{i})}{\sum_{l=1}^{N}w_d(X_{0:d-1}^{l})}
\varphi(X_{d,1}^i)-\pi(\varphi)\bigg\|_{\varrho}\leq
\frac{M(\varrho)\|\varphi\|_{\infty}}{\sqrt{N}}\big[\,e^{\frac{\sigma_{\star}^2}{2}\varrho(\varrho-1)} + 1\,\big]^{1/\varrho}.
\end{equation*}
\end{theorem}

\begin{proof}

Recall that the $N$ particles remain independent. From the definition of the weights in (\ref{eq:weight}),
we can write $w_d(X_{0:d-1}) = e^{\frac{1}{\sqrt{d}}\sum_{j=1}^{d}\overline{W}_{j}(d)}$
 for $\overline{W}_{j}(d)$ being i.i.d.\@ and given in (\ref{eq:Wbar}). Now, we have shown in the proof
of Theorem~\ref{theo:main_result} that $\frac{1}{\sqrt{d}}\sum_{j=1}^{d}\overline{W}_{j}(d)\Rightarrow \mathcal{N}(0,\sigma_{\star}^2)$, thus:
\begin{equation}
\label{eq:111}
w_d(X_{0:d-1}) \Rightarrow e^{X}\ ,\quad X\sim \mathcal{N}(0,\sigma_{\star}^2)\ .
\end{equation}
Then, from Proposition \ref{prop:marginals}, $X_{d,1}$
converges weakly to a random variable $Z\sim \pi$. A simple argument shows that the variables $Z$, $X$ are independent as $Z$ depends only on the first co-ordinate which will not affect
(via $\overline{W}_{1}(d)$)
the limit of $\frac{1}{\sqrt{d}}\sum_{j=1}^{d}\overline{W}_{j}(d)$.
The above results allow us to conclude (due to the boundedness and continuity of the involved functions) that:
\begin{equation}
\label{eq:12}
\lim_{d\rightarrow\infty}
\bigg\|\sum_{i=1}^N
\frac{w_d(X_{0:d-1}^{i})}{\sum_{l=1}^{N}w_d(X_{0:d-1}^{l})}
\varphi(X_{d,1}^i)-\pi(\varphi)\bigg\|_{\varrho}
=
\bigg\|\sum_{i=1}^N\frac{e^{X_i}}{\sum_{l=1}^N e^{X_l}}\varphi(Z_i)-\pi(\varphi)\bigg\|_{\varrho}\ ,
\end{equation}
where the $X_i$ are i.i.d.~$\mathcal{N}(0,\sigma_{\star}^2)$
and independently $Z_i$ are i.i.d.~$\pi$. Now, the limiting random variable in the $\mathbb{L}_{\varrho}$-norm
on the right-hand-side of (\ref{eq:12}) can be written as:
%
%\begin{equation}
%\label{eq:term}
%\frac{\frac{1}{N}\sum_{i=1}^N e^{X_i}\varphi(Z_i)}{e^{\sigma_{\star}^2/2}\frac{1}{N}\sum_{l=1}^N e^{X_l}}
%\big[\,e^{\sigma_{\star}^2/2}-\frac{1}{N}\sum_{l=1}^N e^{X_l}\,\big] +
%e^{-\sigma_{\star}^2/2}\big[\,\frac{1}{N}\sum_{i=1}^N e^{X_i}\varphi(Z_i)-e^{\sigma_{\star}^2/2}\pi(\varphi)\,\big]\ .
%\end{equation}
%
\begin{equation}
\label{eq:term}
\frac{A_{N,\varphi}}{e^{\sigma_{\star}^2/2}A_{N}}
\big[\,e^{\sigma_{\star}^2/2}-A_N\,\big] +
e^{-\sigma_{\star}^2/2}\big[\,A_{N,\varphi}-e^{\sigma_{\star}^2/2}\pi(\varphi)\,\big]
\end{equation}
for
$A_{N,\varphi} = \frac{1}{N}\sum_{i=1}^N e^{X_i}\varphi(Z_i)$ and  $A_N=\frac{1}{N}\sum_{l=1}^N e^{X_l}$.
%On applying triangular inequality, we have to control the sum of:
%\begin{equation}
%\bigg\|
%\frac{\frac{1}{N}\sum_{i=1}^N e^{X_i}\varphi(Z_i)}{e^{\sigma_{\star}^2/2}\frac{1}{N}\sum_{i=1}^N e^{X_i}}
%\big[\,
%e^{\sigma_{\star}^2/2} - \frac{1}{N}\sum_{i=1}^N e^{X_i}\,
%\big]
%\bigg\|_{\varrho}
%\label{eq:ft_mink}
%\end{equation}
%and
%\begin{equation}
%\bigg\|
%e^{-\sigma_{\star}^2/2}
%\bigg[
%e^{\sigma_{\star}^2/2}\pi(\varphi) - \frac{1}{N}\sum_{i=1}^N e^{X_i}\varphi(Z_i)
%\bigg]
%\bigg\|_{\varrho}.
%\label{eq:st_mink}
%\end{equation}
%
Now, using the Marcinkiewicz-Zygmund inequality (there is a version with $\varrho\in[1,2)$ see e.g.~\cite[Chapter 7]{delmoral}), the $\mathbb{L}_{\varrho}$-norm of the first summand in (\ref{eq:term}) is upper-bounded by:
$$
\frac{\|\varphi\|_{\infty}}{e^{\sigma_{\star}^2/2}}\cdot \frac{M(\varrho)}{\sqrt{N}}\,\|e^{X_1}-e^{\sigma_{\star}^2/2}\|_{\varrho}
$$
where $M(\varrho)$ is a constant that depends upon $\varrho$ only. Then applying the $C_p-$inequality and doing standard calculation, this is upper-bounded by
%$$
%\frac{\|\varphi\|_{\infty}e^{-\sigma_{\star}^2/2}}{\sqrt{N}}[M(\varrho)2^{\varrho-1}[e^{\sigma_{\star}^2\varrho/2}+\mathbb{E}[e^{\varrho X_1}]]]^{1/\varrho}.
%$$
%Then by standard calculations, this is upper-bounded by
$$
\frac{M(\varrho)\|\varphi\|_{\infty}}{\sqrt{N}}\,\big[\,e^{\frac{\sigma_{\star}^2}{2}\varrho(\varrho-1)} + 1\,\big]^{1/\varrho}
$$
for some finite constant $M(\varrho)$ that only depends upon $\varrho$.
For the  $\mathbb{L}_{\varrho}$-norm of the second summand in (\ref{eq:term}), again after applying the Marcinkiewicz-Zygmund inequality we have the upper-bound:
$$
e^{-\sigma_{\star}^2/2}\cdot \frac{M(\varrho)}{\sqrt{N}}\,\|e^{X_1}\varphi(Z_1)-e^{\sigma_{\star}^2/2}\pi(\varphi)\|_{\varrho}\ .
$$
Using the $C_p-$inequality and standard calculations we have the upper bound:
%$$
%\frac{M(\varrho)e^{-\sigma_{\star}^2/2}}{\sqrt{N}}[ [\pi(\varphi)e^{-\sigma_{\star}^2/2}]^{\varrho} + \pi(\varphi^\varrho)e^{\sigma_{\star}^2/2 \varrho^2}]^{1/\varrho}
%$$
%which is upper-bounded by
$$
\frac{M(\varrho)\|\varphi\|_{\infty}}{\sqrt{N}}[e^{\frac{\sigma_{\star}^2}{2}\varrho(\varrho-1)} + 1]^{1/\varrho}
$$
for some finite constant $M(\varrho)$ that only depends upon $\varrho$. Thus, we can easily conclude from here.
%In conclusion, we have established that
%$$
%\bigg\|
%\frac{\frac{1}{N}\sum_{i=1}^N e^{X_i}\varphi(Z_i)}{e^{\sigma^2/2}\frac{1}{N}\sum_{i=1}^N e^{X_i}}
%\bigg[
%e^{\sigma^2/2} - \frac{1}{N}\sum_{i=1}^N e^{X_i}
%\bigg]
%\bigg\|_{p}
%+
%\bigg\|
%e^{-\sigma^2/2}
%\bigg[
%e^{\sigma^2/2}\pi(\varphi) - \frac{1}{N}\sum_{i=1}^N e^{X_i}\varphi(Z_i)
%\bigg]
%\bigg\|_{p}
%\leq \frac{C_p\|\varphi\|_{\infty}}{\sqrt{N}}[e^{\frac{\sigma^2}{2}p(p-1)} + 1]^{1/p}
%$$
%for some finite constant $C_p$ that only depends upon $p$, which was to be established.
%Finally, we have $\big\|\frac{1}{N}\sum_{l=1}^N e^{X_l} - e^{\sigma_{\star}^2/2}\big\|_{\varrho}$ and $\big\|\frac{1}{N}\sum_{i=1}^N e^{X_i}\varphi(Z_i)-e^{\sigma_{\star}^2/2}\pi(\varphi)\big\|_{\varrho}$
%are both bounded by $\frac{M_1(\varrho,\varphi)}{\sqrt{N}}$
%as $\mathbb{L}_{\varrho}$-norms of averages of i.i.d.\@ variables of zero expectation; the fact that the coefficients
%in front of these averages in (\ref{eq:term}) are bounded by some $M_{2}(\varphi)$ completes the proof.
\end{proof}

%\begin{rem}
%The result establishes that the error in the estimation of a fixed dimensional marginal is stable as the overall dimension grows to infinity. It is to be expected that for a fixed number of particles, the error in calculating the expectation of high-dimensional marginals (i.e.~whose dimension grows with $d$), will increase with $d$. In line with the results of \cite{cerou1} for normalizing constants, we conjecture the cost to be $\mathcal{O}(d^3)$, when $\varphi:\mathbb{R}^d\rightarrow\mathbb{R}$.
%\end{rem}
%

%%%%%%%%%%%%%%%%%%%%%%%%%%%%%%%%%%%%%%%%%%%%%%%%%%%%%%%%%%%%%%%%%%%%%%%%%%%%%%%%%%%%%%%%%%%%%%%%%%%%%%%%%%%%%%%%%%%

\section{Incorporating Resampling}
\label{sec:resampling_stability}

We have already shown that, even without resampling, the expected ESS converges as $d\rightarrow\infty$
to a non-trivial limit. In practice, this limiting value could sometimes be prohibitively
close to~$1$ depending on the value of $\sigma^{2}_{\star}$; related to this
notice that the constant at the upper bound for the Monte Carlo error in
Theorem \ref{theo:mc_error} is an exponential function of $\sigma^2_{\star}$ and could be large if $\sigma^2_{\star}$ is big.
As a result, it makes sense
to consider the option of resampling in our analysis in high dimensions. We will see that this will result in
smaller bounds for Monte Carlo estimates.

The algorithm carries out $d$ steps as in the case of the
algorithm without resampling considered  in Section~\ref{sec:main_result},
but now resampling occurs at the instances when ESS goes below a specified threshold. For fixed $d$, the algorithm runs in discrete time. Recalling the analogue between discrete and continuous time we have introduced in Remark \ref{rem:switch}
a statement like `resampling occurred at $t\in[\phi_0,1]$'
will literally mean that resampling took place after $l_d(t)$ steps of the algorithm, for the mapping $l_d(t)$ between
continuous and discrete instances defined in (\ref{eq:switch}); in particular, the resampling times,
when considered on the continuous domain, will lie on
the grid $G_d$:
\begin{equation*}
% \label{eq:grid}
 G_d=\{\phi_0+n\,(1-\phi_0)/d\,;\,n=1,\ldots, d\} \
 \end{equation*}
 for any fixed $d$.

Assume that $s\in[\phi_0,1]$ is a resampling time and $x_{l_d(s)}^{\prime,1},\ldots x_{l_d(s)}^{\prime,N}$ are the (now equally weighted) resampled particles.
Due to the i.i.d.~assumptions in (\ref{eq:target}) and (\ref{eq:kernel}), after resampling
each of these particles will evolve according to the Markov kernels $k_{l_d(s)+1}$, $k_{l_d(s)+2}$, $\ldots$,
independently over the $d$ co-ordinates and different particles. The empirical ESS will also
evolve as:
\begin{equation}
\label{eq:ESS1}
\textrm{ESS}_{(s,u)}(N) =  \frac{\big(\sum_{i=1}^N\exp\{
\frac{1}{\sqrt{d}}\sum_{j=1}^d S^i_{s:u,j}\}\big)^2}
{\sum_{i=1}^N\exp\{\frac{2}{\sqrt{d}}\sum_{j=1}^d S^i_{s:u,j}\}}\\
\end{equation}
for $u\in[s,1]$, where we have defined:
%until the next resampling instance $t>s$, whence the $N$ particles:
%
\begin{equation*}
S^i_{s:u,j} = \frac{1-\phi_0}{\sqrt{d}}\sum_{n=l_d(s)+1}^{l_d(u)}
\{g(x_{n-1,j}^{i})-\pi_{n-1}(g)\}\ ,
\end{equation*}
until the next resampling instance $t>s$, whence the $N$ particles,
$
x_{l_d(t)}^{i} =( x_{l_d(t),1}^i,\dots,x_{l_d(t),d}^i )
$
will be resampled according to their weights:
\begin{equation*}
w_{l_d(t)}(x_{l_d(s):(l_d(t)-1)}^{i}) =\exp\{\tfrac{1}{\sqrt{d}}\sum_{j=1}^d S^i_{s:t,j}\}\ .
\end{equation*}
Note that we have modified the subscripts of ESS in (\ref{eq:ESS1}), compared to the original definition in
(\ref{eq:ess_def}), to now run in continuous time.
It should be noted that the dynamics differ from the previous section due to the resampling steps.
For instance $S^{i}_{s:u,j}$ are no longer independent over $i$ or $j$, unless one
conditions on the resampled particles $x_{l_d(s)}^{\prime,i}$, $1\le i\le N$.
% We will exploit this conditional independence structure in some of the %results that follow.
%In addition, if a resampling step occurs, we denote
%by $X_{l_d(t_{n-1}(d))}^{\prime}$ the resampled particle.
% The short-hand notation:
% %
% \begin{equation*}
% x_{l_d(t_{n-1}(d)),1:d}^{',1:N}=(x_{l_d(t_{n-1}(d)),1}^{',1},\dots,x_{l_d(t_{n-1}(d)),d}^{',1},
% \dots,x_{l_d(t_{n-1}(d)),1}^{',N},\dots,x_{l_d(t_{n-1}(d)),d}^{',N})
% \end{equation*}
% is used for the resampled particle at time $t_{n-1}(d)$.
% Finally, in an abuse of notation, we set:
% %
% \begin{equation*}
% \textrm{ESS}_{(t_{n-1}(d),s)}(N) =  \frac{\big(\sum_{l=1}^N\exp\{
% \frac{1}{\sqrt{d}}\sum_{j=1}^d S_{t_{n-1}(d):s,j}\}\big)^2}
% {\sum_{l=1}^N\exp\{\frac{2}{\sqrt{d}}\sum_{j=1}^d S_{t_{n-1}(d):s,j}\}}\\
% \end{equation*}
% %
% for $t_{n-1}(d)\leq s\leq t_{n}(d)$.

\subsection{Theoretical Resampling Times}

We start by showing that the dynamically resampling SMC algorithm, using a deterministic version of ESS
(namely, the expected ESS with one particle)
will resample a finite number of times (again as $d\rightarrow\infty$)
and also exhibit convergence of ESS and of the Monte Carlo error. Subsequently, we show that a dynamically resampling SMC algorithm, using the empirical ESS
(with some modification) will, with high probability, display the
same convergence properties.

We use the resampling-times construction of \cite{delmoral_resampling}:
this involves considering the expected value of the importance weight, and its square, over a system with a  \emph{single} particle.
% We now define for $s\le t$, with both $s,t$, allowed to depend on $d$:
% %
% \begin{equation*}
% S_{s:t,j} :=
% \frac{(1-\phi_0)}{\sqrt{d}}\sum_{n=l_{d}(s)+1}^{l_d(t)}\{ g(X_{n-1,j})-\pi_{n-1}(g)\}\ .
% \end{equation*}
%
The theoretical resampling times are defined as:
\begin{align}
t_1(d) & =
\inf\bigg\{t\in[\phi_0,1]: \frac{\mathbb{E}\,
\big[\,
\exp\big\{\frac{1}{\sqrt{d}}\sum_{j=1}^d S_{\phi_0:t,j}\big\}
\,\big]^2}
{\mathbb{E}\,
\big[\,
\exp\big\{\frac{2}{\sqrt{d}}\sum_{j=1}^d S_{\phi_0:t,j}\big\}
\,\big]} <a\bigg\}\, ; \label{eq:res_time_1}\\
t_k(d) & =  \inf\bigg\{t\in [t_{k-1}(d),1]:
\frac{\mathbb{E}\,
\big[\,
\exp\big\{\frac{1}{\sqrt{d}}\sum_{j=1}^d S_{t_{k-1}(d):t,j}\big\}\,\big]^2}
{\mathbb{E}\,\big[\,\exp\big\{\frac{2}{\sqrt{d}}\sum_{j=1}^d
S_{t_{k-1}(d):t,j} \big\}\,\big]}<a\bigg\}\ , \quad k\geq 2\ ,
\label{eq:res_time_2}
\end{align}
for a constant $a\in (0,1)$,  under the convention that $\inf\varnothing=\infty$.
Note that, for most applications in practice, these times cannot be found analytically.
We emphasize here that the dynamics of $S_{s:t}$ appearing above do not involve resampling
but simply follow the evolution of a single particle with $d$ i.i.d.~co-ordinates, each of which
 starts at $x_{0,j}\sim \pi_{0}$ and then evolves according to the kernels $k_n$.
Intuitively, following the ideas in \cite{delmoral_resampling}, one could think of the deterministic times in (\ref{eq:res_time_1})-(\ref{eq:res_time_2}) as the limit of the resampling times of the practical SMC algorithm in
Figure \ref{tab:SMC} as the number of particles $N$ increases to infinity.

We will for the moment consider the behaviour of the above times in high dimensions. Consider the following instances:
\begin{align}
t_1 &=  \inf\{t\in [\phi_0,1]: e^{-\sigma^2_{\phi_0:t}} <a\}\label{eq:limiting_time1}\ ; \\
t_k &=  \inf\{t\in[t_{k-1},1]: e^{-\sigma^2_{t_{k-1}:t}} <a\}\label{eq:limiting_time2} \ , \quad k\ge 2\ ,
\end{align}
 where for any $s<t$ in $[\phi_0,1]$:
\begin{equation}
\sigma^2_{s:t}= \sigma^{2}_{\phi_0:t} -\sigma^{2}_{\phi_0:s} \equiv (1-\phi_0)\int_{s}^{t}
\pi_u(\widehat{g}_u^2-k_u(\widehat{g}_u)^2)du \ .
\label{eq:varst}
\end{equation}
Under our standard assumptions (A\ref{hyp:A}-\ref{hyp:B}), and the requirement that
 $g\in\mathscr{L}_{V^{r}}$ for some $r\in[0,\tfrac{1}{2})$, we have that (using Lemma \ref{lem:growth} in
the Appendix):
\begin{equation*}
\pi_u(\widehat{g}_u^2-k_s(\widehat{g}_u)^2) \leq M \pi_{u}(V^{2r})\le M'\pi_{\phi_0}(V) < \infty \ .
\end{equation*}
Thus, we can find a \emph{finite} collection of times
that dominate the $t_k$'s (in the sense that there will be more than them),
so also the number of the latter is finite and we can define:
\begin{equation}
\label{eq:no}
m^{\ast} = \#\{\,t_k:k\ge 1\,,t_k\in[\phi_0,1]\,\}<\infty \ .
\end{equation}

We have the following result.
\begin{prop}
\label{prop:limiting_times}
As $d\rightarrow \infty$ we have that $t_k(d)\rightarrow t_k$ for any $k\ge 1$.
\end{prop}
\begin{rem}
Note that the time instances $\{t_k\}$ are derived only through the asymptotic
variance function $t\mapsto\sigma^{2}_{\phi_0:t}$;
our main objective in the current resampling part of this paper will be to illustrate that
investigation of these deterministic
times provides essential information about the resampling times of the practical SMC algorithm  in Figure \ref{tab:SMC}. These latter stochastic times will coincide with the former (or, rather, a slightly modified version of it) as $d\rightarrow \infty$ with a probability
that converges to $1$ with a rate $\mathcal{O}(N^{-1/2})$.
\end{rem}

\subsection{Stability  under  Theoretical Resampling Times}

Consider an SMC algorithm similar to the one in Figure \ref{tab:SMC}, but with the difference
that resampling occurs at the times $\{t_k(d)\}$ in
\eqref{eq:res_time_1}-\eqref{eq:res_time_2}; it is assumed that $t_0(d)=\phi_0$.
Note that due to Proposition \ref{prop:limiting_times}, the number of these resampling times:
\begin{equation*}
m_d^{\ast} = \#\{\,t_k(d):n\ge 1\,,t_k(d)\in[\phi_0,1]\,\}
\end{equation*}
will eventually, for big enough $d$, coincide with $m^{\ast}$ in (\ref{eq:no}). We will henceforth assume
that $d$ is big enough so that $m^{\ast}_d\equiv m^{\ast}<\infty$.

We state our result in Theorem \ref{theorem:limit_adaptive} below, under the convention that $t_{m^{\ast}+1}(d)\equiv 1$. The proof can be found in Appendix \ref{appendix:adaptive}. It relies on a novel construction of a filtration, which starts with all the information of all particles and co-ordinates up-to and including
 the last resampling time.
Subsequent $\sigma-$algebras are generated, for a given particle, by adding each dimension for a given trajectory. This allows one to a use a Martingale CLT approach by taking advantage of the independence of particles and
co-ordinates once we condition on their positions at the resampling times.

\begin{theorem}\label{theorem:limit_adaptive}
Assume (A\ref{hyp:A}-\ref{hyp:B}) and  $g\in\mathscr{L}_{V^r}$ with $r\in[0,\tfrac{1}{2})$. Then, for any fixed $N>1$, any  $k\in\{1,\dots,m^*+1\}$, times $t_{k-1}<t_k$, and
$s_k(d)\in (t_{k-1}(d),t_k(d))$ any sequence converging to a point $s_k\in(t_{k-1},t_k)$, we have that $\textrm{\emph{ESS}}_{(t_{k-1}(d),s_k(d))}(N)$ converges in distribution to a random variable
$$
\frac{[\,\sum_{i=1}^N e^{X_i^k}\,]^2}{\sum_{i=1}^N e^{2X_i^k}}$$
where $X_i^k \stackrel{i.i.d.}{\sim}\mathcal{N}(0,\sigma_{t_{k-1}:s_k}^2)$ and $\sigma_{t_{k-1}:s_k}^2$ as in \eqref{eq:varst}. In particular,
\begin{equation*}
\lim_{d\rightarrow\infty}
\mathbb{E}\,\big[\,\textrm{\emph{ESS}}_{(t_{k-1}(d),s_k(d))}(N)\,\big] = \mathbb{E}\bigg[
\frac{[\,\sum_{i=1}^N e^{X_i^k}\,]^2}{\sum_{i=1}^N e^{2X_i^k}}\bigg].
\end{equation*}
\end{theorem}
%\begin{rem} Observe that the random variables $X_i^k$ converge to $0$ as $a\to 1$  and, therefore, the asymptotic \emph{ESS} converges to the upper limit $N$ as the number of resampling steps increases.
%\end{rem}
%
 Note that, had the $t_{k}(d)$'s been analytically available,
resampling at these instances would deliver an algorithm of $d$ bridging steps for which the expected ESS
would be regularly regenerated. In addition, this latter quantity depends, asymptotically, on the `incremental' variances $\sigma^2_{\phi_0:t_1}$,
$\sigma^2_{t_1:t_2}$, $\ldots$, $\sigma^2_{t_{m^*}:1}$; in contrast, in the context
of Theorem \ref{theo:main_result}, the limiting expectation depends on $\sigma^{2}_{\phi_0:1}\equiv \sigma^2_{\star}$.
We can also consider the Monte-Carlo error when estimating  expectations w.r.t.\@ a single marginal
co-ordinate of our target. Again,  the proof is in Appendix \ref{appendix:adaptive}.

\begin{theorem}\label{theorem:error_adaptive}
Assume (A\ref{hyp:A}-\ref{hyp:B}) with $g\in\mathscr{L}_{V^r}$ for some $r\in[0,\tfrac{1}{2})$. Then
for any $1\leq \varrho <\infty$
there exists a constant $M=M(\varrho)<\infty$ such that for any fixed $ N\ge 1$, $\varphi\in\mathcal{C}_b(\mathbb{R})$
\begin{equation*}
\lim_{d\rightarrow\infty}
\bigg\|\sum_{i=1}^N\frac{
w_{d}(X_{l_d(t_{m^*}(d)):(d-1)}^{i})
}{
\sum_{l=1}^N
w_{d}(X_{l_d(t_{m^*}(d)):(d-1)}^{l})
}\varphi(X_{d,1}^i)-\pi(\varphi)\bigg\|_{\varrho} \leq
\frac{M(\varrho)\|\varphi\|_{\infty}}{\sqrt{N}}\,
\big[e^{\frac{\sigma_{t_{m^*}:1}^2}{2}\varrho(\varrho-1)} + 1\,\big]^{1/\varrho}.
\end{equation*}
\end{theorem}

\begin{rem}
In comparison to the bound in Theorem \ref{theo:mc_error}, the bound is smaller with resampling: as $\phi_0\leq t_{m^*}$ the bound
in Theorem \ref{theorem:error_adaptive}
is clearly less than in Theorem \ref{theo:mc_error}. Whilst these are both upper-bounds on the error they are based on the same calculations - that is a CLT and using the Marcinkiewicz-Zygmund
inequality.
%The key $\sigma^2_{\cdot}$ expression is what allows one to reduce the Monte Carlo error in the case of resampling.
\end{rem}

\begin{rem}
On inspection, the bound in the above result can be seen as counter-intuitive. Essentially, the bound gets \emph{smaller} as $t_{m^*}$ increases, i.e.~the closer to the end one resamples. However, this can be explained as follows. As shown in Proposition \ref{prop:marginals}, the terminal point, thanks to the ergodicity of the system, is asymptotically
drawn from the correct distribution $\pi$. Thus, in the limit $d\rightarrow \infty$
the particles do not require weighting. Clearly, in finite dimensions, one needs to assign weights to compensate for the finite run time of the algorithm.
\end{rem}

We remark that our analysis, in the context of resampling, relies on the fact that $N$ is fixed
and $d\rightarrow \infty$. If $N$ is allowed to grow as well our analysis must be modified when one resamples.
Following closely the proofs in the Appendix,
it should be possible by considering bounds (which do not increase with $N$ and $d$) on quantities of the form
$$
\mathbb{E}\bigg[\sum_{i=1}^N\frac{
w_{l_d(t_k(d))}(X_{l_d(t_{k-1}(d)):(l_d(t_k(d))-1)}^{i})
}{
\sum_{l=1}^N
w_{l_d(t_k(d))}(X_{l_d(t_{k-1}(d)):(l_d(t_k(d))-1)}^{l})
}V(X_{l_d(t_k(d))}^{i})\bigg]
$$
to establish results also for large $N$;
we are currently investigating this.
However, at least following our arguments, the asymptotics under resampling will only be apparent for $N$ much smaller than $d$; we believe that is only due to mathematical complexity and does not need to be the case.

\subsection{Practical Resampling Times}

We now consider the scenario
when one resamples
at the empirical versions of the times \eqref{eq:res_time_1}-\eqref{eq:res_time_2}.
To this end, we will follow closely the proof of \cite{delmoral_resampling} and this will require the
consideration of a finite mesh at the definition of the resampling times. 
Consider some positive integer $\delta$, and the grid:
$$
G_\delta =  \{\phi_0,\phi_0 + (1-\phi_0)/\delta, \phi_0 + 2(1-\phi_0)/\delta,\dots,1\}\ .
$$
We consider the SMC algorithm that attempts to resample only when crossing the instances of the grid $G_\delta$,
using the practically relevant empirical ESS. That is, we are interested in the times $\{T_k = T_{k}^{N}(d)\}$ 
defined as:
\begin{align*}
T_1 &=  \inf\{t\in G_{\delta}\cap [\,\phi_0,1\,]: \tfrac{1}{N}\,\mathrm{ESS}_{\phi_0:t}(N) <a_1\} \ ; \\
T_k &=  \inf\{t\in G_{\delta}\cap  [\,T_{k-1},1\,]:\tfrac{1}{N}\,\mathrm{ESS}_{T_{k-1}:t}(N) <a_k\} \ , \quad k\ge 2\ ,
\end{align*}
for a collection of thresholds $(a_k)$ in $(0,1)$.

Following the development in \cite{delmoral_resampling}, we will need the following theoretical times: 
%defined on a discrete-mesh
%$$
%\Delta_{[s,t]}^{\delta} := \{s,s+(1-\phi_0)/\delta,s+2(1-\phi_0)/\delta,\dots,t\}\ ,
%$$
%determined as follows:
%
\begin{align*}
t_1^{\delta}(d) & =
\inf\bigg\{t\in G_\delta \cap [\phi_0,1]: \frac{\mathbb{E}\,
\big[\,
\exp\big\{\frac{1}{\sqrt{d}}\sum_{j=1}^d S_{\phi_0:t,j}\big\}
\,\big]^2}
{\mathbb{E}\,
\big[\,
\exp\big\{\frac{2}{\sqrt{d}}\sum_{j=1}^d S_{\phi_0:t,j}\big\}
\,\big]} <a_1\bigg\}\, ; \\
t_k^{\delta}(d) & =  \inf\bigg\{t\in G_{\delta}\cap [t_{k-1}^{\delta}(d),1]:
\frac{\mathbb{E}\,
\big[\,
\exp\big\{\frac{1}{\sqrt{d}}\sum_{j=1}^d S_{t_{k-1}^{\delta}(d):t,j}\big\}\,\big]^2}
{\mathbb{E}\,\big[\,\exp\big\{\frac{2}{\sqrt{d}}\sum_{j=1}^d
S_{t_{k-1}^{\delta}(d):t,j} \big\}\,\big]}<a_k\bigg\}\ , \quad k\geq 2\ .
\end{align*}
%
%for a collection of thresholds $(a_k)$ in $(0,1)$.
We can, for a moment, obtain an understanding of the behavior 
of these times as $d\rightarrow \infty$. Define the time instances: 
\begin{align*}
t_1^{\delta} &=  \inf\{t\in G_{\delta}\cap [\phi_0,1]: e^{-\sigma^2_{\phi_0:t}} <a_1\} \ ; \\
t_k^{\delta} &=  \inf\{t\in G_{\delta}\cap  [t_{k-1}^{\delta},1]: e^{-\sigma^2_{t_{k-1}^{\delta}:t}} <a_k\} \ , \quad k\ge 2\ .
\end{align*}
If $m^*(\delta)$ denotes the number of these times, we have that $m^*(\delta)\leq m^*$
(with $m^{*}$ now taking into account the choices of different thresholds $a_k$), but for $\delta$ large enough these values will be very close.
\begin{prop}
\label{prop:limiting_times_2}
As $d\rightarrow \infty$ we have that $t_k^{\delta}(d)\rightarrow t_k^{\delta}$ for any $k\ge 1$.
\end{prop}
\begin{proof}
The proof of $t_{1}(d)\rightarrow t_1$ in Proposition \ref{prop:limiting_times} is based on showing uniform convergence
of $$t\mapsto \frac{\mathbb{E}\,
\big[\,
\exp\big\{\frac{1}{\sqrt{d}}\sum_{j=1}^d S_{\phi_0:t,j}\big\}
\,\big]^2}
{\mathbb{E}\,
\big[\,
\exp\big\{\frac{2}{\sqrt{d}}\sum_{j=1}^d S_{\phi_0:t,j}\big\}
\,\big]}$$  to
 $t \mapsto e^{-\sigma^{2}_{\phi_0:t}}$. Repeating this argument also for subsequent time instances gave that 
$t_{k}(d)\rightarrow t_k$ for all relevant $k\ge 1$. This uniform convergence result can now be called upon to provide the proof 
of the current proposition.
\end{proof}
\noindent Also, Theorems \ref{theorem:limit_adaptive}
and \ref{theorem:error_adaptive} hold under these modified times on $G_{\delta}$.

\subsubsection*{Main Result and Interpretation}
We will use the construction in \cite{delmoral_resampling}. The results therein 
determine the behavior of the SMC method for $d$ \emph{fixed} and increasing number of particles $N$, as 
described in the sequel. Define, for a given $\upsilon\in(0,1)$, the following event:
\begin{align*}
\Omega_{d}^N = \Omega_{d}^N (\upsilon,\{&a_k\}_{1\leq k \leq m^*(\delta)}) := \Big\{\,
\textrm{for all }\, 1\leq k \leq m^*(\delta),\, s\in G_{d}\cap [\,t_{k-1}^\delta(d),t_{k}^\delta(d)\,]:
\\
&\big|\,\tfrac{1}{N}\,\textrm{ESS}_{(t_{k-1}^{\delta}(d),s)}(N)-
\textrm{ESS}_{(t_{k-1}^{\delta}(d),s)}\,\big| < \upsilon\,\big|\,\textrm{ESS}_{(t_{k-1}^{\delta}(d),s)}-
a_{k}\big|
\,\Big\}
\end{align*}
where   
\begin{equation*}
\textrm{ESS}_{(t_{k-1}^\delta(d),s)} =  \frac{\mathbb{E}\,\big[\,\exp\big\{
\frac{1}{\sqrt{d}}\sum_{j=1}^d S_{t_{k-1}^{\delta}(d):s,j}\big\}
\,\big]^2}{\mathbb{E}\,\big[\,\exp\big\{
\frac{2}{\sqrt{d}}\sum_{j=1}^d S_{t_{k-1}^{\delta}(d):s,j}\big\}
\,\big]}\
\end{equation*}
corresponds to the expected ESS over a single particle involved in the definition of $\{t_k^{\delta}(d)\}$.
Here $(a_k)_{1\leq k \leq m^*}$ are a collection of thresholds
which are sampled from some absolutely continuous distribution; they are determined in such a way to avoid the degenerate situation when the thresholds $a_k$ coincide with ESS; see \cite{delmoral_resampling} for details.
Now, the definition of $\Omega_{d}^N$ implies the following:
%Now, the work in  \cite{delmoral_resampling}, for fixed $d$, establishes the following:
%
\begin{itemize}
\item[1.] Within $\Omega_{d}^{N}$, if the deterministic resampling criteria
tell us to resample, so do the empirical ones. That is:
\begin{equation*}
\textrm{ESS}_{(t_{k-1}^{\delta}(d),s)} > a_{k}\,\,\, \Rightarrow \,\,\,
\tfrac{1}{N}\,\textrm{ESS}_{(t_{k-1}^{\delta}(d),s)}(N) > a_{k} \ ,
\quad s\in G_\delta \cap [\,t_{k-1}^\delta(d),t_{k}^\delta(d)\,]\ ,
\end{equation*}
and
\begin{equation*}
\textrm{ESS}_{(t_{k-1}^{\delta}(d),s)} < a_{k}\,\,\, \Rightarrow \,\,\,
\tfrac{1}{N}\,\textrm{ESS}_{(t_{k-1}^{\delta}(d),s)}(N) < a_{k}\ ,
\quad s\in G_{\delta}\cap[\,t_{k-1}^\delta(d),t_{k}^\delta(d)\,]\ .
\end{equation*}

\item[2.] A consequence of the above is that (this is Proposition~5.3 of \cite{delmoral_resampling}):
\begin{equation*}
\bigcap_{1\leq k \leq m^*(\delta)}\{\,T_k=t_{k}^{\delta}(d)\,\}\, \supset\, 
\Omega_{d}^N\ .
\end{equation*}
\item[3.] Conditionally on $\{a_k\}_{1\leq k \leq m^*(\delta)}$, we have that
$\mathbb{P}\,[\,\Omega\setminus\Omega_{d}^N\,]\rightarrow 0$ as $N$ grows \cite[Theorem 5.4]{delmoral_resampling} ($d$ is fixed).
\end{itemize}
The above results provide the interpretation that, with a probability that
increases to $1$ with $N$, the theoretical resampling times $\{t_k^{\delta}(d)\}$ will coincide with the practical 
$\{T_k = T_{k}^{\delta,N}(d)\}$, for any fixed dimension $d$.

Our own contribution involves looking at the stability of these results as the dimension grows, $d\rightarrow\infty$.

\begin{theorem}\label{theo:stoch_times}
Assume (A\ref{hyp:A}-\ref{hyp:B}) and that $g\in\mathscr{L}_{V^r}$, with $r\in[0,\tfrac{1}{2})$.
Conditionally on almost every realization of the random threshold parameters $\{a_k\}$, there exists an $M=M(m^*(\delta))<\infty$ such that for any
$1\leq N<\infty$, we have
\begin{equation*}
\lim_{d\rightarrow\infty}\mathbb{P}\,[\,\Omega\setminus\Omega_{d}^N\,]
\leq \frac{M}{\sqrt{N}}\ .
\end{equation*}
\end{theorem}
%
%
%Steps 1 and 2 are not difficult to establish them as $d$ grows due to the resampling times lying on a finite grid.
%
\noindent The proof  in Appendix \ref{appendix:stoch_times} focuses on point 3.\@
above of the results in  \cite{delmoral_resampling}. 
Thus,  investigation of the times $\{t_{k}^{\delta}\}$ involving only the asymptotic variance 
function $\sigma^2_{s:t}$ can provide an understanding for the number and location
of resampling times of the practical algorithm that uses the empirical ESS. This is because, with high probability,
that depends on the number of particles (uniformly in $d$), 
the practical resampling times will coincide with $\{t_{k}(d)\}$.

%%%%%%%%%%%%%%%%%%%%%%%%%%%%%%%%%%%%%%%%%%%%%%%%%%%%%%%%%%%%%%%%%%%%%%%%%%%%%%%%%%%%%%%%%%%%%%%%%%%%%%%%%%%%%%%

\section{Example on Symmetric Random Walk}
\label{sec:examples}

%\subsection{Symmetric Random Walk}

We will now verify assumptions (A\ref{hyp:A}-\ref{hyp:B}) when the $\pi_s$-invariant transition kernel
is a Random-Walk Metropolis (RWM) algorithm, with proposed increments $\mathcal{N}(0,s^{-1})$.
That is:
\begin{equation*}
 q_s(x,dy) =\frac{\sqrt{s}}{\sqrt{2\pi}}\, e^{-s\,\frac{(y-x)}{2}^2}\,dy
\end{equation*}
with acceptance probability:
\begin{equation*}
a_s(x,y) = 1 \wedge \frac{\pi_s(y)}{\pi_s(x)} \ .
\end{equation*}
For simplicity we set $q_s(dy)\equiv q_s(0,dy)$.
That is, we will look at the Markov kernel:
\begin{equation}
 k_s(x,dy) = a_s(x,y)\,q_{s}(x,dy) + \delta_{x}(dy)\int_{E}(1-a_s(x,y))q_s(x,dy)\ \label{eq:mh_kernel}.
\end{equation}
Notice that we assume that the variance of the proposal is $1/s$, $s\in[\phi_0,1]$.
One can use $f(s)^{-1}$ for the proposal variance, where $f$ is a bounded positive continuous function that is monotonically increasing with a bounded derivative. This is omitted only for notational clarity and using $f$ in the proofs will only complicate the subsequent notations.

We will assume that for every $s\in[\phi_0,1]$ one has
\begin{itemize}
\item{$\pi_s$ is bounded away from zero on compact sets and is upper-bounded.}
\item{$\pi_s$ is super-exponential with asymptotically regular contours; see \cite{jarner} for details.}
\end{itemize}
We will add the condition
\begin{equation}
C^* := \sup_{x\in\mathbb{R}\,,s\in[\phi_0,1]} \bigg\{\int_{A(x)^c}G(x,z) q_s(z)dz\bigg\} < + \infty\label{eq:assumption_a2}
\end{equation}
with $G(x,z)=g(x)-g(x+z)>0$ on $A(x)^c$ (see \eqref{eq:a(x)_def} for details on $A(x)$). This assumption is used to simplify some calculations in the proof and is verifiable (see Remark \ref{rem:verify_cond}).
The above assumptions will be termed E in the following proposition. The proof can be found in
Appendix~\ref{app:verify}.

\begin{prop}\label{prop:verify}
Assume (E). Then the symmetric random walk kernel \eqref{eq:mh_kernel} satisfies (A1-2).
\end{prop}

\begin{rem}\label{rem:verify_cond}
It is straightforward to verify (A1) using standard results in the literature. However, (A2) is non-standard, due to the difference of invariant measures present in \eqref{eq:mh_kernel}.
Note, for \eqref{eq:assumption_a2}, that if $g(x)=-x^2/2$  then
$
G(x,z) = \frac{1}{2}[z^2+2xz]\ .
$
Hence we have
$$
\int_{A(x)^c} G(x,z)q_s(z)dz \leq \frac{1}{2s}\leq \frac{1}{2\phi_0}\ .
$$
Thus, assumption (\ref{eq:assumption_a2}) will hold in the Gaussian case.
\end{rem}

\section{Discussion and Extensions}\label{sec:summary}

We now discuss the general context of our results, provide some extra results and look at potential generalizations.

\subsection{On the Number of Bridging Steps}

Our analysis has relied on using $\mathcal{O}(d)$  bridging steps. An important question
is what happens when one has more or less time steps.
We restrict our discussion to the case where one does not resample, but one can easily extend the results to the resampling scenario.
Suppose one takes $\lfloor d^{1+\delta}\rfloor$ steps, for some real $\delta>-1$ and annealing sequence:
$$
\phi_n = \phi_0 + \tfrac{n(1-\phi_0)}{\lfloor d^{1+\delta} \rfloor} \ , \quad n \in\{0,\dots, \lfloor d^{1+\delta} \rfloor\}\ .
$$
We are to consider the weak convergence of the centered log-weights, which are now equal to:
$$
\frac{\sqrt{d}}{\lfloor d^{1+\delta} \rfloor^{1/2}}\,\alpha_{i}(d)\,$$
where we have defined
$$
\alpha_i(d) = \frac{1}{\sqrt{d}}\sum_{j=1}^d \overline{W}_j(d)\,;\quad \overline{W}_j(d) = W_j(d)-\Exp\,[\,W_j(d)\,]\ ,
$$
with $i\in\{1,\dots,N\}$ and
$$
W_j(d) = \frac{1-\phi_0}{\lfloor d^{1+\delta} \rfloor^{1/2}}\sum_{n=1}^{\lfloor d^{1+\delta} \rfloor}\big\{\,g(x_{n-1,j}) - \pi_{n-1}(g)\,\big\}\ .
$$
One can follow the arguments of Theorem \ref{theo:main_result} to deduce that, under our conditions:
\begin{equation}
\label{eq:alex}
\alpha_{i}(d) \Rightarrow \mathcal{N}(0,\sigma^2_{\star})\ .
\end{equation}
This observation can the be used to provide the  following result.
\begin{cor}\label{cor:ess=N}
Assume (A\ref{hyp:A}(i)(ii), A\ref{hyp:B}) and that
$g\in\mathscr{L}_{V^r}$ for some $r\in[0,\tfrac{1}{2})$.
Then, for any fixed $N>1$:
\begin{itemize}
\item If $\delta >0$ then  $\textrm{\emph{ESS}}_{(0,\lfloor d^{1+\delta} \rfloor)}(N)\rightarrow_{\mathbb{P}} N$.
\item If $-1<\delta<0$ then $\textrm{\emph{ESS}}_{(0,\lfloor d^{1+\delta} \rfloor)}(N)\rightarrow_{\mathbb{P}} 1$.
\end{itemize}
\end{cor}
\begin{proof}
Following (\ref{eq:alex}), if $\delta>0$ then we have that
$
\frac{\sqrt{d}}{\lfloor d^{1+\delta} \rfloor^{1/2}}\,\alpha_i(d) \rightarrow_{\mathbb{P}} 0.
$
All particles are independent, so the proof of the ESS convergence follows easily.

For the case when $-1<\delta<0$ we work as follows. We consider the maximum $M(d) =\max\{\alpha_{i}(d);1\le i\le d\}$.
Let $\bar{\alpha}_{(1)}(d) \leq \bar{\alpha}_{(2)}(d) \leq \cdots \leq \bar{\alpha}_{(N)}(d)$ denote the ordering of
the variables $\alpha_{1}(d)-M(d),\, \alpha_{2}(d)-M(d),\,\ldots,\,\alpha_{N}(d)-M(d)$. We have that (setting for notational
convenience $f_d:= \frac{\sqrt{d}}{\lfloor d^{1+\delta} \rfloor^{1/2}}$):
\begin{equation}
\label{eq:bla}
\textrm{{ESS}}_{(0,\lfloor d^{1+\delta} \rfloor)}(N) =
\frac{\big(\,\sum_{i=1}^{N} e^{\alpha_{i}(d)f_d}\,\big)^2 }
{\sum_{i=1}^{N} e^{2\,\alpha_i(d)\,f_d}}\equiv
\frac{\big(\,1+\sum_{i=1}^{N-1} e^{\bar{\alpha}_{(i)}(d)\,f_d}\,\big)^2}{1+\sum_{i=1}^{N-1} e^{2\,\bar{\alpha}_{(i)}(d)\,f_d}} \ .
\end{equation}
Due to the continuity of the involved mappings, the fact that $(\alpha_1(d),\ldots, \alpha_{N}(d))\Rightarrow
\mathcal{N}(0,\sigma_{\star}^2\,I_N)$ implies the weak limit
$(\bar{\alpha}_{(1)}(d),\ldots, \bar{\alpha}_{(N-1)}(d))\Rightarrow (\bar{\alpha}_{(1)},\ldots, \bar{\alpha}_{(N-1)})$ as $d\rightarrow \infty$
with the latter variables denoting the ordering $\bar{\alpha}_{(1)}\leq \bar{\alpha}_{(2)}\leq \cdots \leq  \bar{\alpha}_{(N)}\equiv 0$
of $\alpha_1-M,\alpha_2-M, \, \ldots, \alpha_N-M$ where the $\alpha_i$'s are i.i.d.\@ from $\mathcal{N}(0,\sigma^2_{\star})$
and $M$ is their maximum. Since $(\bar{\alpha}_{(1)}(d),\ldots, \bar{\alpha}_{(N-1)}(d))$ and their weak limit take a.s.\@ negative values,
we have that $(\bar{\alpha}_{(1)}(d)f_d,\ldots, \bar{\alpha}_{(N-1)}(d)f_d)\Rightarrow (-\infty, \ldots, -\infty)$ which (continuing from (\ref{eq:bla})) implies the
stated result.
\end{proof}
For the stable scenario, with $\delta>0$, we also have the following.
\begin{cor}
Assume (A\ref{hyp:A}(i)(ii), A\ref{hyp:B}) with $g\in\mathscr{L}_{V^r}$ for some $r\in[0,\tfrac{1}{2})$. Then
for any $1\leq \varrho <\infty$,
%there exists a constant $M=M(\varrho)<\infty$ such that for any
 $N\geq 1$,  $\varphi\in\mathcal{C}_b(\mathbb{R})$, $\delta>0$:
\begin{equation*}
\lim_{d\rightarrow\infty}
\bigg\|\sum_{i=1}^N
\frac{w_d(X_{0:\lfloor d^{1+\delta} \rfloor-1}^{i})}{\sum_{l=1}^{N}w_d(X_{0:\lfloor d^{1+\delta} \rfloor-1}^{l})}
\varphi(X_{\lfloor d^{1+\delta} \rfloor,1}^i)-\pi(\varphi)\bigg\|_{\varrho} =
\bigg\|
\frac{1}{N}\sum_{i=1}^N\varphi(Z_i)
-\pi(\varphi)
\bigg\|_{\varrho}
\end{equation*}
where $Z_i\stackrel{\textrm{i.i.d.}}{\sim}\pi$.
\end{cor}
\begin{proof}
This follows from the proof of Theorem \ref{theo:mc_error} and Corollary \ref{cor:ess=N}.
\end{proof}

\noindent Thus, a number of steps of $\mathcal{O}(d)$ is a critical regime: less than this, will lead to the algorithm
collapsing w.r.t.~the ESS and more steps is `too-much' effort as one obtains very favourable results.

\subsection{Full-Dimensional Kernels}
\label{sec:full}

An important open problem is the investigation of the stability properties of SMC as $d\rightarrow \infty$
when one uses full-dimensional kernels $K_n(x,dx')$, instead of a product of univariate kernels considered in our analysis.
We will state a conjecture for this case here, indicating the increased technical complexity to the scenario of this article and sketching future research in this direction.
We remain in the i.i.d.\@ context for the target density and do not consider resampling for ease of presentation. Consider the Markov kernel $P_n(x,dx')$
  with invariant density $\Pi_n$ corresponding to RWM  with proposal dynamics ($Z\sim \mathcal{N}_d(0,I_d)$, $l>0$):
\begin{align*}
 X_{pr} = x + \sqrt{h}\,Z\ ;\,\quad h = \tfrac{l^2}{d}\,
\end{align*}
so that $X' = x_{pr}$ with
probability $a(x,x_{pr})=1\wedge \{\Pi_n(x_{pr})/\Pi_n(x)\}$; otherwise $X'=x$.
The particular choice of  step-size $h$ shown in the proposal above as an order of $d$ was found in the MCMC literature (\cite{roberts1, roberts2,bedard}) to provide  algorithms that do not degenerate as $d$ increases.

We consider the standard SMC method in Figure \ref{tab:SMC} under the choice of kernels $K_n = (P_{n})^{d}$ for RWM
so that at each instance $n$ we synthesize $d$ steps from $P_n(x,dx')$. We conjecture that this choice for $K_n(x,dx')$ will
provide a stable SMC method as $d\rightarrow\infty$. Some of the fundamental building blocks of our analysis for the asymptotic properties of the ESS when using product kernels in the previous sections are: (i) the independency over the $d$ co-ordinates; (ii) each co-ordinate is making $\mathcal{O}(1)$-steps in it's state space with dynamics
of appropriate ergodicity properties. As analytically explained in the aforementioned MCMC literature, the convolution
of $d$ steps for RWM provides, \emph{asymptotically}, independency between the co-ordinates, with each co-ordinate making
(essentially) $d$ steps of size $1/d$  along the path (over the time period $[0,1]$) of the following limiting scalar SDE:
\begin{equation}
\label{eq:SDE}
dY_n(t) = \tfrac{a_n(l)l}{2}(\log\pi_n)'(Y_n(t))dt + \sqrt{a_n(l)l}\,d\mathcal{W}_t
\end{equation}
with $a_n(l) = \lim_{d\rightarrow \infty}\Exp\,[\,a(X,X_{pr})\,]\in(0,1)$; the expectation is with $x$ in stationarity, $X\sim \Pi_n$.
Thus, we conjecture that, when considering the centered log-weights:
\begin{equation}
\label{eq:newk}
\frac{1}{\sqrt{d}}\sum_{j=1}^{d}\frac{\sum_{n=1}^{d}\,\{\,g(x_{n-1,j})-\pi_n(g)\,\}}{\sqrt{d}}\,,
\end{equation}
their weak limit would remain unchanged if the dynamics of the Markov chain with kernels $K_n=(P_n)^d$ are replaced with those
of a Markov chain with $K^*_n(x,dx')= \prod_{j=1}^{d}k^*_n(x_j, dx_j')$ where
$k^*_n(x_j,dx_j') = \mathbb{P}\,[\,Y_n(1) \in dx'_j\,|\,Y_n(0)=x_j\,]$ is the transition
density of the SDE (\ref{eq:SDE}). Now, under these dynamics, we are within the context
of our main results in Section \ref{sec:main_result} and, under the assumptions stated there, we can prove weak convergence of (\ref{eq:newk}) to $\mathcal{N}(0,\sigma^2_{\star})$ for $\sigma^2_{\star}$ now involving the
continuum $k^*_s(x_j,dx_j')$ of the SDE transition densities.

Thus, the technical challenge left for future research is proving that:
\begin{equation*}
\frac{1}{d}\sum_{n,j=1}^{d}\,\{\,g(x_{n-1,j})-g(y_{n-1,j}(1))\,\}
\Rightarrow 0 \  ,
\end{equation*}
that requires \emph{coupling} the probability measures $\Pi_0\,K_1\,\cdots K_n$
and $\Pi_0\,
K^{*}_1\,\cdots K^{*}_n$ determining the dymamics of the time-inhomogeneous $d$-dimensional Markov chains $\{x_{0},x_{1},\ldots, x_{d}\}$ and $\{y_{0}(1),y_{1}(1),\ldots, y_{d}(1)\}$ respectively. That is to say, a coupling between the $d$-steps of RWM,  $K_n=(P_n)^{d}$, and the sample paths of the limiting diffusions, determining $K_n^{*}$.
This is certainly a non-trivial task that will go beyond
the aforementioned MCMC literature, as the limiting results are based on convergence of generators and do not require strong path-wise convergence.

Under our conjecture, the SMC method based on full-dimensional RWM kernels, with stabilize at a total cost of
$\mathcal{O}(Nd^{3})$. A similar conjecture for MALA (Metropolis-adjusted Langevin algorithm) will involve stability of the SMC
method at a reduced cost of $\mathcal{O}(Nd^{7/3})$ as for MALA one has to synthesize $\mathcal{O}(d^{1/3})$ steps of size $\mathcal{O}(d^{-1/3})$ to obtain the diffusion limit (see \cite{roberts2}).
Finally, we conjecture that an alternative SMC method that uses $\mathcal{O}(d^2)$ bridging
steps ($\phi_{n} = \phi_0 + n(1-\phi_0)/d^2)$ with RWM transition kernels of step-sizes $h=l/d^2$ as before,
instead of convoluting $d$ full-dimensional kernels (for MALA, that would involve using $\mathcal{O}(d^{4/3})$ bridging steps)
will also be stable for fixed $N$ as $d$ increases. This is because of the
the structural similarity of it's dynamics for blocks of $d$ bridging steps with the previous case; however a proof for this case does not seem to be connected
with the work in our paper and will have to follow a different direction. An analytical solution to this latter issue, by consideration of the variances in the CLT, may help to answer whether or not one should
iterate the MCMC kernel or have more annealing steps in high-dimensional scenarios.

\subsection{Beyond I.I.D.~Targets}

In the MCMC literature, the first attempts to move beyond the i.i.d.\@ context involved looking at restricted classes of models,
see e.g.\@ \cite{brey:00, brey:04, bedard}.
The most recent contributions in this still-open research direction have looked at target distributions
in high-dimensions defined as changes of measure from Gaussian laws (\cite{besk, andrew:1, andrew:2}). This probabilistic structure contains a large family of practically relevant statistical models (see e.g.\@ \cite{besko}). We will discuss an extension of our results in this paper in such a direction.
Following \cite{andrew:1, andrew:2}, we consider a target distribution on an infinite-dimensional separable Hilbert space $\mathcal{H}$ determined via the change of measure:
\begin{equation*}
\frac{d\mathsf{\Pi}}
{d\mathsf{\Pi}_0}(x) \propto \exp\{-\Psi(x)\}\ , \quad x\in \mathcal{H} \ ,
\end{equation*}
for some functional $\Psi:\mathcal{H}\mapsto \mathbb{R}$,
with $\mathsf{\Pi}_0=\mathcal{N}(0,\mathcal{C})$ a Gaussian law on $\mathcal{H}$.
Let $\{e_j\}_{j\in\mathbb{N}}$ be the orthonormal base of $\mathcal{H}$ comprised of eigenvectors of $\mathcal{C}$ with
corresponding eigenvalues $\{\lambda^2_j\}_{n\in\mathbb{N}}$. $\mathsf{\Pi}_0$ can be expressed in terms of it's so-called  Karhunen-Lo\`eve expansion:
$$
\mathsf{\Pi}_0 \stackrel{law}{=}\sum_{i=1}^{\infty}\lambda_j\xi_j\,e_j
$$
where $\xi_j \stackrel{\textrm{i.i.d.}}{\sim} \mathcal{N}(0,1)$.
In practice, one will have to project the target to some $d$-dimensional approximation,
and a standard generic approach is to truncate the basis expansion; that is, to work with the $d$-dimensional target:
$$\Pi(x) \propto \exp\{-\Psi_d(x)-
\tfrac{1}{2}\langle x, C_d^{-1} x \rangle\}\ , \quad x\in \mathbb{R}^{d}\ , $$
with $C_d=\mathrm{diag}\{\lambda_1^2,\ldots,\lambda_d^2\}$ and
$\Psi_d(x)=\Psi(\sum_{j=1}^{d}x_j\,e_j)$.

In connection with the SMC method in this paper,
we will look at the algorithm in Figure \ref{tab:SMC} with bridging densities $\Pi_n(x) \propto\{\Pi(x)\}^{\phi_n}$, where $\phi_{n} = \phi_0 + n(1-\phi_0)/d$,
and propagating kernels $K_n = (P_n)^{d}$, with $P_n$ corresponding to Markov transition of a RWM algorithm with target distribution $\Pi_n$
and proposal:
\begin{equation*}
X_{pr} = X + \sqrt{h}\,C_{d}^{1/2}\,Z \,; \quad h = \tfrac{l^2}{d}\ ,
\end{equation*}
with $Z\sim \mathcal{N}_d(0,I_d)$. Again, we do not consider the possibility of resampling, only for notational simplicity.
Our conjecture here is that this SMC method will be stable as $d\rightarrow \infty$, for fixed number of particles $N$,
at a total computational cost $\mathcal{O}(Nd^{3})$.
In a similar context to Section \ref{sec:full}, it is shown in \cite{andrew:1}
that the above choice of step-size $h$ provides a non-degenerate
MCMC algorithm as $d\rightarrow \infty$. More analytically, asymptotically in $d$, the $d$ steps of Markov transitions $P_n$
correspond to making steps of size $1/d$ on the paths of an $\mathcal{H}$-valued SDE. The centered log-weights will now be:
$$
\frac{1-\phi_0}{d}\sum_{n=1}^{d}\Big(-\Psi_d(x_{n-1}) + \mathbb{E}\,[\,\Psi_d(x_{n-1})\,] -
\tfrac{1}{2}\langle x_{n-1}, C_d^{-1}\,x_{n-1} \rangle  + \tfrac{1}{2}\,\Exp\,[\,\langle x_{n-1}, C_d^{-1}\,x_{n-1}\rangle\,]\,  \Big)
$$
with $X_n\mid X_{n-1}=x_{n-1}\sim K_n(x_{n-1},\cdot)$.
We conjecture here, that starting from a $d$-variate version of the Poisson equation (a generalisation of the univariate version for the results proven in this paper) one should aim at showing:
\begin{align*}
\frac{1}{d}\,\sum_{n=1}^{d}\big\{\Psi_d(x_{n-1}) - \mathbb{E}\,[\,\Psi_d(x_{n-1})\,]\,\big\} &\Rightarrow 0 \ ; \\
\frac{1}{d}\,\sum_{n=1}^{d}\big\{-\tfrac{1}{2}\,\langle x_{n-1}, C_d^{-1}\,x_{n-1} \rangle  + \tfrac{1}{2}\,\Exp\,[\,\langle x_{n-1}, C_d^{-1}\,x_{n-1}\rangle\,]\,\} &\Rightarrow \mathcal{N}(0,\sigma^2_{\star}) \ ,
\end{align*}
for some asymptotic variance $\sigma^{2}_{\star}$.
For the first limit, one should consider a Poisson equation
associated to the functional $x\mapsto \Psi_d(x)$, for the Markov chain with dynamics $K_n$.  For the second limit, the $d$-variate Poisson equation should apply upon the functional
$x\mapsto \langle x, C^{-1}_d\,x\rangle/\sqrt{d}$. Both these functionals seem to stabilize as $d\rightarrow \infty$.
The asymptotic variance $\sigma^2_{\star}$ is expected to involve an integral over the transition density of the limiting $\mathcal{H}$-valued SDEs.

%In addition to this non i.i.d.~target, it may also be fruitful to consider other dependence structures that are studied in the MCMC literature (e.g.~\cite{bedard}), using the framework that is described above.

\subsection{Some New Results}

An important application of SMC samplers is in the approximation of the normalizing constant of $\Pi$. This is a non-trivial extension of the work in this article, but we have
obtained the stability in high-dimensions of the relative $\mathbb{L}_2-$error of the SMC estimate; we refer the reader to \cite{beskos1}. This stability is achieved with
a computational cost of $\mathcal{O}(Nd^2)$ with stronger assumptions than in this article.

Recall that we have used the annealing sequence \eqref{eq:tune}. However, one could also
consider a general differentiable, increasing Lipschitz function $\phi(s)$, $s\in[0,1]$ with $\phi(0)=\phi_0\ge 0$,  $\phi(1)=1$, and use the construction
$\phi_{n,d}=\phi(n/d)$; this is also considered in \cite{beskos1}.
The asymptotic results generalized to the choice
of $\phi_{n,d}$ here would involve the variances:
\begin{equation*}
\sigma_{s:t}^{2,\phi}= \int_{s}^{t}
\pi_{\phi(u)}(\widehat{g}_{\phi(u)}^2-k_{\phi(u)}(\widehat{g}_{\phi(u)})^2)
\bigg[\frac{d\phi(u)}{du}\bigg]
d\phi(u)\ , \quad 0\leq s < t \leq 1\label{eq:general_variance}\ ,
\end{equation*}
So for example the bound in Theorem \ref{theo:mc_error} becomes
$$
\frac{M(\varrho)\|\varphi\|_{\infty}}{\sqrt{N}}[e^{\frac{\sigma_{\phi_0:1}^{2,\phi}}{2}\varrho(\varrho-1)} + 1]^{1/\varrho}.
$$
In theory one could use this quantity to choose between SMC algorithms with different annealing schemes; see \cite{beskos1} for some discussion.

An interesting avenue to pursue is the stability of the SMC approximation
of multi-level Feynman-Kac formulae \cite{delmoral}. This is particularly important for problems in rare-events analysis. In this case one introduces a sequence of sets which converge to the rare region of interest.
The question is how to parameterize the sets such that, as one makes the set of interest rarer, the algorithm is stable (e.g.~w.r.t.~logarithmic efficiency). We suggest \cite{cerou1} and \cite{dean} from the splitting literature as useful starting points.
It may also be of interest to investigate more advanced SMC samplers such as \cite{chopin2,chopin_smc}.

\subsubsection*{Acknowledgements}
We thank Pierre Del Moral, Christophe Andrieu, Nicolas Chopin, Adam Johansen and Arnaud Doucet for valuable conversations associated to this work. The work of Dan Crisan was partially supported by the EPSRC Grant No: EP/H0005500/1.
The work of Ajay Jasra was supported by a MOE grant and was employed by Imperial College London during part of this project. We also thank Nikolas Kantas for some useful feedback.
We thank two referees for extremely useful comments that have vastly improved the paper.

\appendix

\section{Technical Results}
\label{app:A}

In this appendix we provide some technical results that will be used in the proofs that follow. The results in Lemma \ref{lem:growth} are fairly standard within the context of the analysis of non-homogeneous Markov chains with drift conditions (e.g.~\cite{doucchains}).
The decomposition in Theorem \ref{th:decompose} will be used repeatedly in the proofs.

For a starting index $n_0=n_0(d)$ we denote here
by $\{X_{n}(d)\,;\,n_0\leq n\leq d\}$ the non-homogeneous scalar Markov chain
evolving via:
\begin{equation*}
\mathbb{P}\,[\,X_n(d)\in dy \mid X_{n-1}(d)=x\,] = k_{n,d}(x,dy )\ ,\quad
n_0 < n \le d \ ,
\end{equation*}
with the kernels $k_{n,d}$ preserving $\pi_{n,d}$.
All variables $X_n(d)$ take values in the homogeneous measurable space
$(E,\mathscr{E})=(\bbR,\mathcal{B}(\bbR))$.
For simplicity, we will often omit indexing the above quantities with $d$.
% and constructed upon the probability space  $(\Omega,\mathscr{F},\mathbb{P})$.
% We will denote by $\{\mathscr{F}_{n,d}\}_{n}$ the filtration generated by $\{X_{n}(d)\,;\,0\leq n\leq d\}$.

Given the Markov kernel $k_s$ with invariant distribution $\pi_s$ (here, $s\in[\phi_0,1]$),
and some function $\varphi$, we consider the Poisson equation
\begin{equation*}
\varphi(x)-\pi_s(\varphi)=f(x)-k_s(f)(x) \ ;
\end{equation*}
under (A\ref{hyp:A}) there is a unique solution $f(\cdot)$ (see e.g.~\cite{meyn}), which
can be expressed via the infinite series
$f(x) = \sum_{l\geq 0}
[k_{s}^l-\pi_{s}](\varphi)(x)$. We use the notation $f=\mathcal{P}(\varphi,k_s,\pi_s)$
to define the solution of such an equation.

We will sometimes use the notation $\Exp_{X_{n_0}}\,[\,\cdot\,]\equiv \Exp\,[\,\cdot\,|\,X_{n_0}\,]$.

%
% For a given collection of functions
% $\{\varphi_{s}\}_{s\in[\phi_0,1]}$, Markov kernels $\{k_s\}_{s\in[\phi_0,1]}$ with invariant measures $\{\pi_s\}_{s\in[\phi_0,1]}$, we consider the Poisson equation $\varphi_s(x)-\pi_s(\varphi_s)=f_s(x)-k_s(f_s)(x)$. Under (A\ref{hyp:A}) there is a unique solution (see e.g.~\cite{meyn}), which
% can be expressed via the infinite series
%
% $
% f_s(x) = \sum_{l\geq 0}
% [k_{s}^l-\pi_{s}](\varphi_s)(x)\ .
% $

% the following notation is used to represent the solution of Poisson equations:
% %
% \begin{equation*}
% \mathcal{P}(\varphi_s,k_s,\pi_s) = \{f_s\in\mathcal{M}(\mathbb{R}\times[\phi_0,1]): \forall x\in\mathbb{R}\ ,\,\,\varphi_s(x)-\pi_s(\varphi_s)=f_s(x)-k_s(f_s)(x)\}\ ,
% \end{equation*}
% %
% where $\mathcal{M}(\cdot)$ denotes measurable functions on the relevant state space.
%
\begin{lem}
\label{lem:growth}
Assume (A\ref{hyp:A}-\ref{hyp:B}). Then, the following results hold.
\begin{itemize}
\item[i)]\label{lemma:b11} Let $\varphi\in \mathscr{L}_{V^{r}}$ for some $r\in[0,1]$ and
set $\widehat{\varphi}=\mathcal{P}(\varphi, k_s, \pi_s)$.
Then, there exists $M=M(r)$ such that
\begin{equation*}
|\widehat{\varphi}(x)| \leq M |\varphi|_{V^r}\,V(x)^r \ .
\end{equation*}
\item[ii)] \label{lemma:b12}
Let $\varphi_s,\varphi_t\in \mathscr{L}_{V^{r}}$ for some $r\in[0,1]$ and consider
 $\widehat{\varphi}_s = \mathcal{P}(\varphi_s, k_s, \pi_s)$ and
$\widehat{\varphi}_t = \mathcal{P}(\varphi_t, k_t, \pi_t)$.
% $\widehat{\varphi}_s-k_s(\widehat{\varphi}_s)=\varphi_s-\pi_s(\varphi_s)$, $\widehat{\varphi}_t-k_t(\widehat{\varphi}_t)=\varphi_t-\pi_t(\varphi_t)$.
Then, there exists $M=M(r)$
such that:
\begin{equation*}
|\widehat{\varphi}_t(x) - \widehat{\varphi}_s(x) | \le M\,(\,
|\varphi_t-\varphi_s|_{V^{r}}+ |\varphi_t|_{V^{r}}\,|||k_s-k_t|||_{V^{r}}\,)\,V(x)^r \ .
\end{equation*}
\item[iii)] \label{lemma:b13}
For any $r\in(0,1]$ and $0\le n_0\le n$:
\begin{equation*}
\mathbb{E}\,[\,V(X_{n})^{r}\,|\,X_{n_0}\,] \le \lambda^{(n-n_0)r}V^r(X_{n_0}) +
\tfrac{1-\lambda^{r(n-n_0)}}{1-\lambda^{r}}\,b^r \le M\,V^r(X_{n_0})\ .
\end{equation*}
\end{itemize}
\end{lem}
\begin{proof}
{\bf i):}
We proceed using the geometric ergodicity of $k_s$:
\begin{align*}
|\widehat{\varphi}(x)| & =  |\,\sum_{l\geq 0}
[k_{s}^l-\pi_{s}](\varphi)(x)\,|
 \leq  |\varphi|_{V^r}\sum_{l\geq 0}
\|[k_{s}^l-\pi_{s}](x)\|_{V^r} \leq  M\,|\varphi|_{V^r}[\,\sum_{l\geq 0}\rho^l]\,V(x)^r
\end{align*}
for some $\rho\in(0,1)$ and $M>0$ not depending on $s$ via (A\ref{hyp:A}); it is now straightforward to conclude.

{\bf ii)} Via the Poisson equation we have
%
%
%\begin{equation*}
$\widehat{\varphi}_{t}(x) - \widehat{\varphi}_{s}(x) =
A(x) + B(x)$
%\end{equation*}
%
where
\begin{align}
A(x) &= \sum_{l\geq 0}[k_{t}^l-\pi_{t}](\varphi_{t})(x) -
\sum_{l\geq 0}[k_{s}^l-\pi_{s}](\varphi_{t})(x)\ ;\nonumber \\
B(x) &= \sum_{l\geq 0}[k_{s}^l-\pi_{s}](\varphi_{t}-\varphi_{s})(x)
\label{eq:e}
\ .
\end{align}
We start with $B(x)$. For each summand we have:
\begin{align*}
|\,[k_{s}^l-\pi_{s}](\varphi_{t}-\varphi_{s})(x)\,| & =
|\varphi_{t}-\varphi_{s}|_{V^r}\,|\,
[k_{s}^l-\pi_{s}]\big(\tfrac{\varphi_{t}-\varphi_{s}}{|\varphi_{t}-\varphi_{s}|_{V^r}}\big)(x)\,|
\\
& \leq |\varphi_{t}-\varphi_{s}|_{V^r}\,\|k_{s}^l-\pi_{s}\|_{V^r}
\leq M\,|\varphi_{t}-\varphi_{s}|_{V^r}\,\rho^l\,V(x)^r \ ,
\end{align*}
where $M>0$ and $\rho\in(0,1)$ depending only on $r$ due to (A\ref{hyp:A}).
Hence, summing over $l$, there exist a $M>0$ such that for any $x\in E$:
\begin{equation*}
B(x) \leq M\,|\varphi_{t}-\varphi_{s}|_{V^r}\,V(x)^r\ .
\end{equation*}
Returning to $A(x)$ in \eqref{eq:e},
one can use Lemma C2 of \cite{andr:10} to show that this is equal to:
\begin{equation*}
\sum_{l\geq 0}\bigg[\sum_{i=0}^{l-1}[k_{t}^i-\pi_{t}]
[k_{t}-k_{s}][k_{s}^{l-i-1}-\pi_{s}](\varphi_{t})(x)
 - [\pi_{t}-\pi_{s}]\big([k_{s}^{l}-\pi_{s}](\varphi_{t})\big)\bigg]\ .
\end{equation*}
Using identical manipulations to \cite{andr:10}, it follows that:
\begin{equation*}
\sum_{l\geq 0}\bigg| \sum_{i=0}^{l-1}[k_{t}^i-\pi_{t}]
[k_{t}-k_{s}][k_{s}^{l-i-1}-\pi_{s}](\varphi_{t})(x)
\bigg| \leq M\,|\varphi_t|_{V^{r}}\, |||k_s-k_t|||_{V^{r}}\, V(x)^r
\end{equation*}
and, for some constant $M=M(r)>0$:
\begin{equation*}
| \sum_{n\geq 0}[\pi_{t}-\pi_{s}]\big([k_{s}^{n}-\pi_{s}](\varphi_{t})\big)\,|
\leq M\,|\varphi_t|_{V^{r}}\, |||k_s-k_t|||_{V^{r}}\, V(x)^r \ .
\end{equation*}
{\bf iii)}
We will use the drift condition in (A\ref{hyp:A}).
Using Jensen's inequality (since $r\le 1$)
we obtain
%
%\begin{equation*}
$k_n(V^{r})(X_{n-1}) \le \lambda^{r} V^{r}(X_{n-1}) + b^{r}$
%\end{equation*}
%
for the constants $b$, $\lambda$ appearing in the drift condition. Using this inequality
and conditional expectations:
\begin{equation*}
\Exp\,[\,V^{r}(X_{n})\,|\,X_{n_0}\,]  =  \Exp\,[\,k_n(V^{r}(X_{n-1}))\,|\,X_{n_0}\,]
\le \lambda^{r}\,\Exp\,[\,V^{r}(X_{n-1})\,|\,X_{n_0}\,] + b^{r}\ .
\end{equation*}
Applying this iteratively gives the required result.
\end{proof}

%%%%%%%%%%%%%%%%%%%%%%%%%%%%%%%%%%%%%%%%%%%%%%%%%%%%%%%%%%%

\begin{theorem}[Decomposition]
\label{th:decompose}
Assume (A\ref{hyp:A}(i)(ii),A\ref{hyp:B}).
Consider the collection of functions $\{\varphi_{s}\}_{s\in[\phi_0,1]}$ with $\varphi_s\in \mathscr{L}_{V^r}$ for
some $r\in[0,1)$ and such that:
\begin{itemize}
\item[i)] $\sup_{s}|\varphi_{s}|_{V^{r}}<\infty$,
\item[ii)] $|\varphi_t-\varphi_s|_{V^r}\leq M\,|t-s|$\ .
\end{itemize}
Set $\varphi_{n}(=\varphi_{n,d}):=\varphi_{\{s=\phi_{n}(d)\}}$ and consider
the solution to the Poisson equation
 $\widehat{\varphi}_{n} = \mathcal{P}(\varphi_n,k_n,\pi_n)$.
%where $\phi_0\le \psi_{1,d}\le \psi_{2,d} \le \cdots\leq\psi_{d,d}\le 1$
%with $\psi_{i,d}-\psi_{i-1,d}\le \tfrac{C}{d}$ for some constant $C>0$.
Then, for $n_0\le n_1\le n_2$ we can write:
\begin{equation*}
\sum_{n=n_1}^{n_2}
\big\{\,\varphi_{n}(X_{n}) - \pi_{n}(\varphi_{n})\} = M_{n_1:n_2} + R_{n_1:n_2}
\end{equation*}
for the martingale term:
\begin{equation*}
M_{n_1:n_2} = \sum_{n=n_1+ 1}^{n_2}\big\{\,
\widehat{\varphi}_{n}(X_n)-k_{n}(\widehat{\varphi}_{n})(X_{n-1})\,\big\}
\end{equation*}
such that for any $p > 1$ with $r\,p\le 1$:
\begin{equation*}
\mathbb{E}\,[\,|M_{n_1:n_2}|^p\,|\,X_{n_0}\,]
\leq M\, d^{\frac{p}{2} \vee 1}\,V^{rp}(X_{n_0}) \ ,
\end{equation*}
and a residual term $R_{n_1:n_2}$ such that for any $p>0$ with $r\,p\le 1$:
\begin{equation*}
\mathbb{E}\,[\,|R_{n_1:n_2}|^p\,|\,X_{n_0}\,] \leq M\, V^{rp}(X_{n_0})\ .
\end{equation*}
%
%
% \begin{equation*}
% \lim_{d\rightarrow\infty}\,\,\sup_{0\le n_1\le n_2\le d}\Exp\,\bigg|\,\frac{1}{d}\sum_{n=n_1}^{n_2}
% \big\{\,\varphi_{n,d}(X_{n}(d)) - \pi_{n,d}(\varphi_{n,d})\,\big\}\,\bigg| = 0\ .
% \end{equation*}
\end{theorem}
\begin{proof}
Using the Poisson equation
$\varphi_n(\cdot)-\pi_n(\varphi_n)=\widehat{\varphi}_n(\cdot)-k_n(\widehat{\varphi}_n)(\cdot)$,
simple addition and subtraction of the appropriate terms gives that:
\begin{equation}
\label{eq:crude}
\sum_{n=n_1}^{n_2}
\big\{\,\varphi_{n}(X_{n}) -\pi_{n}(\varphi_{n})\,\big\} = M_{n_1:n_2} + D_{n_1:n_2} -
E_{n_1:n_2}  + T_{n_1:n_2} \ ;
\end{equation}
\begin{align*}
D_{n_1:n_2} & =\sum_{n=n_1+1}^{n_2}[\widehat{\varphi}_{n}(X_{n-1})-\widehat{\varphi}_{n-1}(X_{n-1})]\ ,\\
E_{n_1:n_2} & =\sum_{n=n_1+1}^{n_2}[\varphi_{n}(X_{n-1})-\varphi_{n-1}(X_{n-1})]\ , \\
T_{n_1:n_2} & = \widehat{\varphi}_{n_1}(X_{n_1}) - \widehat{\varphi}_{n_2}(X_{n_2})
-\pi_{n_1}(\varphi_{n_1}) + {\varphi}_{n_2}(X_{n_2})\ .
\end{align*}
Now, using Lemma \ref{lem:growth}(i),(iii) and the uniform bound in assumption~(i) we get directly
that:
%
% \begin{equation*}
% \tfrac{1}{d}\,\widehat{\varphi}_{1,d}(X_{0}(d))\ ,\,\,
% \tfrac{1}{d}\,\widehat{\varphi}_{d,d}(X_{d}(d))\ ,\,\,
% \tfrac{1}{d}\,\varphi_{1,d}(X_{0}(d))\ , \,\,
% \tfrac{1}{d}\,\varphi_{d,d}(X_{d}(d))\
% \end{equation*}
\begin{equation}
\label{eq:1}
\Exp\,[\,|T_{n_1:n_2}|^{p}\mid X_{n_0}\,] \le M \,\,V^{rp}(X_{n_0})\ .
\end{equation}
Also, Lemma \ref{lem:growth}(i) together with assumption (i) imply that:
\begin{align*}
|\,(\varphi_{n}-\varphi_{n-1})(X_{n-1})\,|  \le
 |\varphi_{n}-\varphi_{n-1}|_{V^{r}}\,V^{r}(X_{n-1}) \le  M \,\tfrac{1}{d}\,V^{r}(X_{n-1})\ ,
\end{align*}
thus, calling again upon Lemma \ref{lem:growth}(iii), one obtains that:
\begin{equation}
\label{eq:2}
\Exp\,[\,|E_{n_1:n_2}|^{p}\mid X_{n_0}\,] \le M\,V^{rp}(X_{n_0})\ .
\end{equation}
Consider now $D_{n_1:n_2}$. Using first Lemma \ref{lem:growth}(ii), then conditions (i)-(ii) and  (A\ref{hyp:B}) one yields:
\begin{equation*}
|\widehat{\varphi}_{n}(X_{n-1})-\widehat{\varphi}_{n-1}(X_{n-1})| \le M\,\tfrac{1}{d}\,
V(X_{n-1})^{r}\ .
\end{equation*}
Thus, using also Lemma \ref{lem:growth}(iii) we obtain directly that:
\begin{equation}
\label{eq:3}
\Exp\,[\,|D_{n_1:n_2}|^{p}\mid X_0\,] \le M\,V(X_{n_0})^{rp} \ .
\end{equation}
The bounds (\ref{eq:1}), (\ref{eq:2}) and (\ref{eq:3}) prove the stated result for the
growth of $\Exp\,[\,|R_{n_1:n_2}|^{p}\,]$.

Now consider the martingale term $M_{n_1:n_2}$. One can use a modification of the Burkholder-Davis-Gundy inequality  (e.g.~\cite[pp.~499-500]{shiryaev})
which states that for any $p>1$:
\begin{equation}
\label{eq:bb}
\mathbb{E}\,[\,|M_{n_1:n_2}|^p\,|\,X_{n_0}\,]
\leq M(p)\,d^{\frac{p}{2} \vee 1 - 1} \sum_{n=n_1+1}^{n_2} \mathbb{E}\,[\,|\,\widehat{\varphi}_{n}(X_n)-k_{n}(\widehat{\varphi}_{n})(X_{n-1})\,|^p\,|\,X_{n_0}\,]\ ,
\end{equation}
see \cite{atch:09} for the proof. Using Lemma \ref{lem:growth}(i)
we obtain that:
\begin{equation*}
|\,\widehat{\varphi}_{n}(X_n)-k_{n}(\widehat{\varphi}_{n})(X_{n-1})\,| \le
M\, |\varphi_{n}|_{V^r}\,(\, V^{r}(X_n) + k_{n}(V^{r})(X_{n-1}) \,)\ .
\end{equation*}
Using this bound, Jensen inequality giving $(k_{n}(V^{r})(X_{n-1}))^p\le k_{n}(V^{rp})(X_{n-1})$, the fact that $r\,p\le 1$ and Lemma \ref{lem:growth}(iii), we continue from (\ref{eq:bb}) to obtain
the stated bound for $M_{n_1:n_2}$.
% \begin{equation*}
% \mathbb{E}\,[\,|M_{n_1:n_2}|^p\,|\,X_{n_0}\,] \le M\,d^{\frac{p}{2}\vee 1}\,\,V^{rp}(X_{n_0})\ .
% \end{equation*}
%
\end{proof}
\begin{prop}
\label{prop:prop}
Let $\varphi\in\mathscr{L}_{V^{r}}$ with $r\in[0,1]$. Consider two sequences of times $\{s(d)\}_d$, $\{t(d)\}_d$
in $[\phi_0,1]$ such that $s(d)<t(d)$ and $s(d)\rightarrow s$, $t(d)\rightarrow t$ with $s<t$.
If we also have that  $\sup_{d}\,\Exp\,[\,V^r(X_{l_d(s(d))})\,] < \infty$, then:
\begin{equation*}
\Exp_{X_{l_d(s(d))}}\,[\,\varphi(X_{l_d(t(d))})\,]\rightarrow \pi_{t}(\varphi)\ , \quad\textrm{in }\,\mathbb{L}_1\ .
\end{equation*}
 \end{prop}
\begin{proof}
Recall that $\pi_u (x) \propto \exp\{u\,g(x)\}$ for $u\in[\phi_0,1]$.
We define, for $c\in(0,\tfrac{1}{2})$:
\begin{equation*}
n_d =  l_d(t(d))-l_d(s(d)) \ ; \quad m_d = \lfloor \{l_d(t_d)-l_d(s_d)\}^{c}\rfloor \ ;\quad  u_d = l_d(s(d)) + n_d - m_d \ .
\end{equation*}
Note that from the definition of $l_d(\cdot)$ we have $n_d = \mathcal{O}(d)$, whereas $m_d =\mathcal{O}(d^c)$.
We have that:
\begin{align}
|\,\Exp_{X_{l_d(s(d))}}\,[\,\varphi(X_{l_d(t(d))})\,]&-\pi_t(\varphi)\,| \leq |\,\Exp_{X_{l_d(s(d))}}\,[\,\varphi(X_{l_d(t(d))})-k_{u_d}^{m_d}(\varphi)(X_{u_d})\,]\,| \nonumber \\
&+ |\,\mathbb{E}_{X_{l_d(s(d))}}\,[\,k_{u_d}^{m_d}(\varphi)(X_{u_d})\,]-\pi_{u_d}(\varphi)\,| + |\,\pi_{u_d}(\varphi)-\pi_t(\varphi)\,|\ . \label{eq:ante1}
\end{align}
Now, the last term on the R.H.S.\@ of (\ref{eq:ante1}) goes to zero as $d\rightarrow\infty$:
this is via dominated convergence after noticing that
\begin{equation*}
\pi_{u_d}(\varphi) = \frac{\int \varphi(x)e^{(\phi_0 + \frac{u_d}{d}(1-\phi_0))g(x)}dx}
{\int e^{(\phi_0 + \frac{u_d}{d}(1-\phi_0))g(x)}dx}
\end{equation*}
with the integrand of the term, for instance, in the numerator converging almost everywhere (w.r.t.\@ Lebesque) to
$\varphi(x)e^{t\,g(x)}$ (simply notice that $\lim u_d/d = \lim \{l_d(t(s))/d\} = (t-\phi_0)/(1-\phi_0)$) and being bounded in absolute value (due to the assumption of $g$ being upper bounded) by the integrable function $M\,V^{r}(x)e^{\phi_0 g(x)}$. Also, the second term on the R.H.S.\@ of (\ref{eq:ante1}) goes to zero in
$\mathbb{L}_1$, due the uniform in drift condition in (A\ref{hyp:A}); to see this, note that (working as in the proof of Lemma \ref{lem:growth}(i)) condition A\ref{hyp:A} gives
$\|k_{s}^l-\pi_{s}\|_{V^r} \leq M\,\rho^l\,V(x)^r \ $ for any $s\in(\phi_0,1]$, so we also have that  $|k_{u_d}^{m_d}(\varphi)(X_{u_d})-\pi_{u_d}(\varphi)| \le M\,\rho^{m_d}\,V(X_{u_d})^r$.
Taking expectations and using Lemma \ref{lem:growth}(iii) we obtain that:
\begin{equation*}
 |\,\mathbb{E}_{X_{l_d(s(d))}}\,[\,k_{u_d}^{m_d}(\varphi)(X_{u_d})\,]-\pi_{u_d}(\varphi)\,|  \le
 M\,\rho^{m_d}\,V(X_{l_d(s(d))})^{r}\ .
\end{equation*}
which vanishes in $\mathbb{L}_1$ as $d\rightarrow \infty$ due to the assumption $\sup_{d}\,\Exp\,[\,V^r(X_{l_d(s(d))})\,] < \infty$.

We now focus on the first term on the R.H.S.\@ of (\ref{eq:ante1}).
The following decomposition holds, as intermediate terms in the sum below cancel out, for $u_d\geq 1$:
\begin{align*}
\mathbb{E}_{X_{l_d(s(d))}}\,&[\,\varphi(X_{l_d(t(d))})-k_{u_d}^{m_d}(\varphi)(X_{u_d})\,] = \\
&\mathbb{E}_{X_{l_d(s(d))}}\big[\,\sum_{j=0}^{m_d-1}\{k_{(u_d+1):(l_d(t(d))-j)}\,k_{u_d}^j(\varphi)(X_{u_d})-
k_{(u_d+1):(l_d(t(d))-(j+1))}k_{u_d}^{j+1}(\varphi)(X_{u_d})
\}\,\big]
\end{align*}
where we use the notation $k_{i:j}(\varphi)(x)= \int k_i(x,dx_1)\times\cdots\times k_{j}(\varphi)(x_{j-i+1})$, $i\leq j$. Each of the  summands is equal to
\begin{equation*}
k_{u_d+1:l_d(t(d))-(j+1)}[k_{l_d(t(d))-j}-k_{u_d}](k_{u_d}^j(\varphi))(X_{u_d})
\end{equation*}
which is  bounded in absolute value by
\begin{equation*}
M\,|\varphi|_{V^{r}}\,k_{u_d+1:l_d(t(d))-(j+1)}(V^r)(X_{u_d})\,|||k_{l_d(t(d))-j}-k_{u_d}|||_{V^r}\ .
\end{equation*}
Now, from Lemma \ref{lem:growth}(iii):
$$
k_{u_d+1:l_d(t(d))-(j+1)}(V^r)(X_{u_d}) \leq M V^r(X_{u_d})\ .
$$
Also, from condition  (A\ref{hyp:B}), there exists an $M>0$ such that
$$
|||k_{l_d(t(d))-j}-k_{u_d}|||_{V^r} \leq M\, \tfrac{(1-\phi_0)}{d}\,(l_d(t(d))-j-u_d) \equiv
 M\,\tfrac{(1-\phi_0)}{d}\,(m_d-j)\ .
$$
Thus, using again Lemma \ref{lem:growth}(iii) we are left with
$$
|\,\mathbb{E}_{X_{l_d(s(d))}}\,[\,\varphi(X_{l_d(t(d))})-k_{u_d}^{m_d}(\varphi)(X_{u_d})\,]\,|
\leq M\,V^r(X_{l_d(s(d))})\,\sum_{j=0}^{m_d-1}\frac{m_d-j}{d}\ .
$$
As $\sup_{d}\,\Exp\,[\,V^r(X_{l_d(s(d))})\,] < \infty$, since $m_d =\mathcal{O}(d^c)$ with $c\in(0,\tfrac{1}{2})$ we can easily conclude.
\end{proof}
\section{Proofs for Section \ref{sec:main_result}}
\label{sec:proof}
%%%%%%%%%%%%%%%%%%%%%%%%%%%%%%%%%%%%%%%%%%%%%%%%%%%%%%%%%%%%%%%%%%%%%%%%%%%%%%%%%%%%%%%%%%%%%%%%%%%%%%%%%%%%%%%%%%%

% We will first prove the convergence in probability of sample path averages of a certain class of functions. This result will be later used in the proof of  Theorem \ref{th:clt}.
There are related results to Theorem \ref{th:clt} (see e.g.\@ \cite{mcle:74, with:81}),
however in our case, the proofs will be based on assumptions commonly made in the MCMC and SMC literature, which will be easily verifiable.
The general framework will involve constructing a Martingale difference array (an approach also followed in the above mentioned
papers).

\begin{prop}
\label{prop:hyp}
Assume (A\ref{hyp:A}(i)(ii), A\ref{hyp:B}) and $g\in \mathscr{L}_{V^{r}}$ with $r\in[0,\tfrac{1}{2})$.
The family of functions $\{{\varphi}_{s}\}_{s\in[\phi_0,1]}$ specified as:
\begin{equation*}
{\varphi}_{s}(x) = k_{s}(\widehat{{g}}_{s}^2)(x) - \{ k_{s}(\widehat{{g}}_{s})(x) \}^2\ ,  \quad
\widehat{{g}}_{s} = \mathcal{P}(g, k_s, \pi_s)\ ,
%\quad\quad \widehat{g}_s - k_s(\widehat{g}_s) = g - \pi_s(g)
\end{equation*}
satisfies conditions (i) and (ii) of Theorem \ref{th:decompose} for $\bar{r}=2r\in[0,1)$.
\end{prop}
\begin{proof}
Lemma \ref{lem:growth}(i) gives that
%
%\begin{equation*}
$|\widehat{g}_s(x)| \le M\,|g|_{V^{r}}\,V^{r}(x)$.
%\end{equation*}
%
Thus, due to the presence of quadratic functions in the definition of $\varphi_s(\cdot)$ we get directly that
%
%\begin{equation*}
$|{\varphi}_{s}(x)|\le M\,V^{\bar{r}}(x)$
%\end{equation*}
%
so condition (i) in Theorem \ref{th:decompose} is satisfied.
We move on to condition (ii) of the theorem. Let us first deal with:
\begin{equation*}
\{k_{t}(\widehat{g}_{t})(x)\}^2 - \{k_{s}(\widehat{g}_{s})(x)\}^2
\end{equation*}
which is equal to
\begin{align*}
\{k_t(\widehat{g}_t)(x)-k_s(\widehat{g}_t)(x)\}\{k_t(\widehat{g}_t)(x)+k_s(\widehat{g}_t)(x)\}
+ \{k_s(\widehat{g}_t-\widehat{g}_s)(x)\}\{k_s(\widehat{g}_t+\widehat{g}_s)(x)\} \ .
\end{align*}
The terms with the additions are bounded in absolute value by $M\,V^{r}(x)$, whereas:
\begin{equation*}
|\,k_t(\widehat{g}_t)(x)-k_s(\widehat{g}_t)(x)\,| \le M\,|t-s|\,V(x)^{r}\ ,\quad
|\,k_s(\widehat{g}_t-\widehat{g}_s)(x)\,| \le  M\,|t-s|\,V(x)^{r} \ ,
\end{equation*}
the first inequality following from assumption (A\ref{hyp:B}) and the second from Lemma \ref{lem:growth}(ii). Thus, we have proved:
\begin{equation*}
|\,\{k_{t}(\widehat{g}_{t})(x)\}^2 - \{k_{s}(\widehat{g}_{s})(x)\}^2\, | \le M\, |t-s|\, V(x)^{\bar{r}}
\end{equation*}
for $\bar{r}=2r\in(0,1)$.
We move on to the second term at the expression for $\varphi_s$ and work as follows:
\begin{equation*}
 k_{t}(\widehat{g}_{t}^2)(x) -
 k_s(\widehat{g}_s^2)(x) =
  k_{t}(\widehat{g}_{t}^2)(x)
  -
 k_s(\widehat{g}_{t}^2)(x)
 +
 k_s(\widehat{g}_{t}^2)(x)
 -
k_s(\widehat{g}_s^2)(x)\ .
\end{equation*}
The first difference is controlled, from assumption (A\ref{hyp:B}), by
$
M\,|t-s|\,V(x)^{\bar{r}}\ ,
$
whereas for the second difference we use Cauchy-Schwarz to obtain:
\begin{align*}
 |k_s(\widehat{g}_{t}^2)(x)-k_s(\widehat{g}_s^2)(x)|& \le
\{ k_s(\widehat{g}_{t}-\widehat{g}_{s})^{2}(x) \}^{1/2}\{ k_s(\widehat{g}_{t}+\widehat{g}_{s})^{2}(x) \}^{1/2}\\
& \le M\,|t-s|\,V(x)^{\bar{r}}
\end{align*}
where, for the second inequality, we have used Lemma \ref{lem:growth}(ii).
The proof is now complete.
\end{proof}

\begin{proof}[Proof of Theorem \ref{th:clt}]

We adopt the decomposition as in Theorem \ref{th:decompose}.
Set $\widehat{{g}}_{s}$
to be a solution to the Poisson equation (with $\pi_s$, $k_s$)
and
$\widehat{{g}}_{n-1,d} = \widehat{{g}}_{\{s=\phi_{n-1}\}}$.
The decomposition is then:
\begin{equation*}
\sum_{n=1}^{l_d(t)} \{ {g}(X_{n-1}(d)) -\pi_{n-1,d}({g})\} = M_{0:l_d(t)-1} + R_{0:l_d(t)-1}
\end{equation*}
where
\begin{equation*}
M_{0:l_d(t)-1} = \sum_{n=1}^{l_d(t)-1}\{\widehat{{g}}_{n,d}(X_{n}(d))-k_{n,d}(\widehat{{g}}_{n,d})(X_{n-1}(d))\}\ .
\end{equation*}
It is clear, via Theorem \ref{th:decompose}, that $R_{0:l_d(t)-1}/\sqrt{d}$ goes to zero in $\mathbb{L}_1$ and hence we need consider the Martingale array term only.

Writing
\begin{equation*}
\xi_{n,d} = \widehat{{g}}_{n,d}(X_{n}(d))-k_{n,d}(\widehat{{g}}_{n,d})(X_{n-1}(d))
\end{equation*}
one observes that $\{\xi_{n,d},\mathscr{F}_{n,d}\}_{n=1}^{d-1}$, with $\mathscr{F}_{n,d}$ denoting the filtration
generated by $\{X_n(d)\}$, is a square-integrable Martingale difference array with zero mean.
In order to prove the fCLT, one can use Theorem 5.1 of \cite{berger} which gives the following sufficient
conditions for proving Theorem~\ref{th:clt}:
\begin{itemize}
%
% \item[a)] $\mathbb{E}\,[\,\xi_{n,d}\mid\mathscr{F}_{n-1,d}\,]=0$,  $\quad\mathbb{P}-$a.s..
\item[a)] For every $\epsilon>0$, $I_{\epsilon,d} := \frac{1}{d}\sum_{n=1}^d
\mathbb{E}\,[\,\xi_{n,d}^2\,\mathbb{I}_{|\xi_{n,d}|\geq \epsilon \sqrt{d}}\mid\mathscr{F}_{n-1,d}\,] \rightarrow 0$ in probability.
% \item[c)] $I_{d}:=\frac{1}{d}\sum_{n=1}^{d-1}\mathbb{E}\,[\,\xi_{n,d}^2\mid\mathscr{F}_{n-1,d}\,]$ converges in probability to $\sigma_{\star}^2$.
\item[b)] For any $t\in[\phi_0,1]$, $ I_d(t):=\frac{1}{d}\sum_{n=1}^{l_d(t)}\mathbb{E}\,[\,\xi_{n,d}^2\mid\mathscr{F}_{n-1,d}\,]$ converges in probability to the quantity
$\sigma^{2}_{\phi_0:t}/(1-\phi_0)^2$.
\end{itemize}
%
% Note we omit the factor $(1-\phi_0)$, but this is easily reintroduced.
We proceed by proving these two statements.

We prove a) first.
Recall that $r\in[0,\tfrac{1}{2})$, so we can choose  $\delta>0$ so that $r(2+\delta)\le 1$. In the first line below, one can use simple calculations and
in the second line Lemma \ref{lem:growth}(i) and the drift condition with $r(2+\delta)\le 1$,
to obtain:
\begin{align*}
|\xi_{n,d}|^{2+\delta} & \le M(\delta)\big(\,|\widehat{{g}}_{n,d}(X_{n}(d))|^{2+\delta}+
|k_{n,d}(\widehat{{g}}_{n,d})(X_{n-1}(d))|^{2+\delta}\,\big) \\
   & \le  M(\delta)\big(\, V(X_{n}(d)) + V(X_{n-1}(d))   \,\big)  \ ,
\end{align*}
Thus,
using Lemma \ref{lem:growth}(iii) we get:
$
\sup_{n,d}\,\mathbb{E}\,[\,|\xi_{n,d}|^{2+\delta}\,]<\infty \ .
$
A straightforward application of H\"older's inequality, then followed by Markov's inequality, now gives that:
\begin{equation*}
\mathbb{E}\,[\,I_{\epsilon,d}\,]\le \frac{1}{d}\sum_{n=1}^d
\big(\,\mathbb{E}\,[\,|\xi_{n,d}|^{2+\delta}\,]\,\big)^{\frac{2}{2+\delta}}\,\,
\big(\,
\mathbb{P}\,[\,|\xi_{n,d}|\geq \epsilon \sqrt{d}\,]
\,\big)^{\frac{\delta}{2+\delta}} \le M\, d^{-\frac{1}{2}\frac{\delta}{2+\delta}} \ .
\end{equation*}
Thus, we have proved a).

For b), we can rewrite:
\begin{equation}
\label{eq:re}
I_d(t) = \frac{1}{d}\,\sum_{n=1}^{l_d(t)} \bigg[k_{n,d}(\widehat{{g}}_{n,d}^2)(X_{n-1}(d)) - \big\{ k_{n,d}(\widehat{{g}}_{n,d})(X_{n-1}(d))\big\}^2 \bigg] \ .
\end{equation}
We will be calling upon Theorem \ref{th:decompose} to prove convergence of the above quantity
to an asymptotic variance. Note that, via Proposition \ref{prop:hyp}, the mappings
\begin{equation*}
\varphi_{s}:=k_{s}(\widehat{{g}}_{s}^2) - \big\{ k_{s}(\widehat{{g}}_{s})\big\}^2
\end{equation*}
satisfy conditions (i)-(ii) of Theorem \ref{th:decompose}.
We define $\varphi_{n,d}=\varphi_{\{s=\phi_{n}(d)\}}$ and rewrite $I_d(t)$ as:
\begin{equation*}
I_d(t) = \frac{1}{d}\,\sum_{n=0}^{l_d(t)-1}\varphi_{n+1,d}(X_{n}(d)) \ .
\end{equation*}
We also define:
\begin{equation*}
J_d(t) = \frac{1}{d}\,\sum_{n=0}^{l_d(t)-1}\varphi_{n,d}(X_{n}(d)) \ .
\end{equation*}
Due to condition (ii) of Theorem \ref{th:decompose}, we have that
$I_d(t) - J_d(t) \rightarrow 0$ in $\mathbb{L}_1$.
Applying Theorem \ref{th:decompose} one can deduce that:
\begin{equation*}
 \lim_{d\rightarrow \infty} \{\,J_d(t) -  \frac{1}{d}\sum_{n=0}^{l_d(t)-1}\pi_{n,d}(\varphi_{n,d})\,\}  = 0\ ,\quad
\textrm{in }\, \mathbb{L}_1 \ .
\end{equation*}
Now, $s\mapsto \pi_s(\varphi_s)$ is continuous as a mapping on $[\phi_0,1]$, so from standard calculus we get that $\frac{1-\phi_0}{d}\sum_{n=0}^{l_d(t)-1}\pi_{n,d}(\varphi_{n,d})\rightarrow \int_{\phi_0}^{t}\pi_s(\varphi_s)ds$.
Combining the results, we have proven that:
\begin{equation*}
I_d(t) \rightarrow  (1-\phi_0)^{-1}\int_{\phi_0}^{t}\pi_s(\varphi_s)ds\equiv \sigma_{\phi_0:t}^2/(1-\phi_0)^2\ ,\quad \textrm{in }\,\mathbb{L}_1 \ .
\end{equation*}
Note that by Corollary 3.1 of Theorem 3.2 of \cite{hall:80}
we also have an CLT for $S_{1}$.

% To complete the proof, consider condition d). Setting $u=\int_{\phi_0}^t\pi_s(\widehat{g}_s^2-k_s(\widehat{g}_s)^2) ds$ and $f(u)$ the value of $t$ for a given $u$, using the time-scale $f(\sigma_{\star}^2u')$, $u=\sigma_{\star}^2u'$,  one yields that (via Theorem \ref{theorem:slln_mc}), for any $x\in\mathbb{R}$ that
% $$
% \frac{x^2}{d}\sum_{n=1}^{\lfloor df(u) \rfloor}\xi_{n,d}^2 \rightarrow
% \sigma_{\star}^2u' x^2
% ,\quad \textrm{in }\,L_1(\mathbb{P}) \
% $$
% from which we can easily conclude.
\end{proof}

\section{Proofs for Section \ref{sec:resampling_stability}}

\subsection{Results for Proposition \ref{prop:limiting_times}}

We will first require a proposition summarising convergence results, with emphasis
on uniform convergence w.r.t.\@ the time index.
%Thus, the notation
% `$f_d(t)\rightarrow_{t} f(t)$' will denote convergence, as $d\rightarrow\infty$,
% uniformly for all $t$ in the domain of $f_d(t)$. The domain of $f_d(t)$ for the cases we will consider
% will actually depend on $d$, so denoting it by $\mathcal{D}_{d}$ the exact statement will be that
% $\sup_{t\in \mathcal{D}_d}|f_{d}(t)-f_t|\rightarrow 0$.
% %

\begin{prop}
\label{pr:unif}
Assume (A\ref{hyp:A}-\ref{hyp:B}).
Let $s(d)$ be a sequence on $[\phi_0,1]$ such that $s(d)\rightarrow s$. Then:
\begin{itemize}
\item[i)] $\sup_{t\in[s(d),1]}\Exp\,[\,|S_{s(d):t,j}\,|\,]/\sqrt{d}\rightarrow 0$.
\item[ii)] $\sup_{t\in[s(d),1]}|\,\Exp\,[\,S^2_{s(d):t,j}\,]-\sigma^2_{s:t}\,|\rightarrow 0$.
\item[iii)] $\sup_{t\in[s(d),1]}|\,\Exp\,[\,S_{s(d):t,j}\,]\,|\rightarrow 0$.
\item[iv)] $\sup_{d\ge 1,s\in[s(d),t]}{\Exp\,[\,S^{2+\epsilon}_{s(d):t}\,]}<\infty$, for some $\epsilon>0$.
\end{itemize}
\end{prop}
\begin{proof}
For simplicity, we will omit reference to the co-ordinate index $j$.
Applying the decomposition of Theorem \ref{th:decompose} for
$\varphi_s\equiv g$ and $n_0=0$ gives that:
\begin{equation*}
S_{s(d):t} = \tfrac{(1-\phi_0)}{\sqrt{d}}\,(M_{l_d(s):(l_d(t)-1)} + R_{l_d(s):(l_d(t)-1)})
\end{equation*}
with (choosing $p=2+\epsilon$ for $\epsilon>0$ so that $r\,p\le 1$):
\begin{equation*}
\Exp[\,|M_{l_d(s):(l_d(t)-1)}|^{2+\epsilon}\,] \le M\,d^{1+\frac{\epsilon}{2}}\,\Exp\,[\,V(X_0)\,]\ ,
\end{equation*}
and (choosing $p=2+\epsilon$ for $\epsilon>0$ so that $r\,p\le 1$):
\begin{equation*}
\Exp[\,|R_{l_d(s):(l_d(t)-1)}|^{2+\epsilon}\,] \le M\,\Exp\,[\,V(X_0)\,] \ .
\end{equation*}
One now needs to notice that these bounds are \emph{uniform} in $s,t,d$, thus statements (i) and (iv) of
the proposition follow directly from the above estimates; statement (iii) also follows directly after
taking under consideration that $\Exp\,[\,M_{l_d(s):(l_d(t)-1)}\,] =0$. It remains to prove (ii).
The residual term $R_{l_d(s):(l_d(t)-1)}/\sqrt{d}$ vanishes in the limit in $\mathbb{L}_{2+\epsilon}$-norm,
thus it will not affect the final result, that is:
\begin{equation*}
\sup_{t\in[s(d),1]} |\,\Exp\,[\,S^{2}_{s(d):t}\,] -
\tfrac{(1-\phi_0)^2}{d}\,\Exp\,[\,M^2_{d,l_d(s(d)):(l_d(t)-1)}\,]\, |
\rightarrow 0\ .
\end{equation*}
Now, straightforward analytical calculations yield:
\begin{align*}
\tfrac{1}{d}\,\Exp\,[\,M^2_{d,l_d(s(d)):(l_d(t)-1)}\,] &=
\frac{1}{d}\sum_{n=l_d(s(d))}^{l_d(t)-1}\Exp\,[\,\{\widehat{g}_n(X_n)-k_n(\widehat{g}_n)(X_{n-1})\}^2\,]\\
&=\Exp\,\big[\, \frac{1}{d}\sum_{n=l_d(s(d))-1}^{l_d(t)-2}\varphi_{n+1}(X_{n})
\big]\ ,
\end{align*}
where we have set:
\begin{equation*}
 \varphi_{s} = k_s(\widehat{g}_s^2)-\{k_s(\widehat{g}_s)\}^2 \ ;\quad \varphi_{n}=\varphi_{\{s=\phi_{n}\}} \ .
\end{equation*}
Since $|\varphi_{n+1}-\varphi_n|_{V^{2r}}\le M\frac{1}{d}$ from Proposition \ref{prop:hyp}, we also have:
\begin{equation*}
\sup_{t\in[s(d),1]}\bigg|\,
\Exp\,\big[\, \frac{1}{d}\sum_{n=l_d(s(d))-1}^{l_d(t)-2}\varphi_{n+1}(X_{n})\big] -
\Exp\,\big[\, \frac{1}{d}\sum_{n=l_d(s(d))-1}^{l_d(t)-2}\varphi_{n}(X_{n})\big]\,\bigg|\rightarrow 0\ .
\end{equation*}
Now, Theorem \ref{th:decompose} and Proposition \ref{prop:hyp} imply that:
\begin{equation*}
\sup_{t\in[s(d),1]}\Exp\,\bigg|\,
\frac{1}{d}\sum_{n=l_d(s(d))-1}^{l_d(t)-2} \{ \varphi_{n}(X_{n}) - \pi_{n}(\varphi_n)\}\,\bigg|\rightarrow 0\ .
\end{equation*}
Finally, due to the continuity of $s\mapsto \pi_s(\varphi_s)$, it is a standard result from Riemann integration
(see e.g.~Theorem 6.8 of \cite{rudin}) that:
\begin{equation*}
\sup_{t\in[s(d),1]}\bigg|\,\tfrac{1-\phi_0}{d}\sum_{n=l_d(s(d))-1}^{l_d(t)-2}\pi_n(\varphi_n)
- \int_{s}^{t}\pi_u(\varphi_u)du\,\bigg| \rightarrow 0\
\end{equation*}
and we conclude.
\end{proof}

\begin{proof}[Proof of Proposition \ref{prop:limiting_times}]
For some sequence $s(d)$ in $[\phi_0,1]$ such that $s(d)\rightarrow s$,
 we will consider the function in $t\in[s(d),1]$:
\begin{equation*}
f_d(s(d),t):= \frac{{\Exp}^2\,\big[\,
\exp\big\{\frac{1}{\sqrt{d}}\sum_{j=1}^d S_{s(d):t,j}\big\}\,
\big]}
{{\Exp}\,\big[\,\exp\big\{\tfrac{2}{\sqrt{d}}
\sum_{j=1}^d S_{s(d):t,j} \big\}\,\big] } \equiv
\Bigg(\frac{ {\Exp}^2\,\big[\,
\exp\big\{\frac{1}{\sqrt{d}}\,S_{s(d):t,1}\big\}\,
\big]}
{{\Exp}\,\big[\,\exp\big\{\tfrac{2}{\sqrt{d}}\,
 S_{s(d):t,1} \big\}\,\big]}\Bigg)^d
\end{equation*}
the second result following due to the independence over $j$.
In the rest of the proof we will omit reference to the co-ordinate index $1$.
Due to the ratio in the definition of $f_d(s(d),t)$, we can clearly re-write:
\begin{equation*}
f_d(s(d),t) = \Bigg(\frac{ {\Exp}^2\,\big[\
\exp\big\{\frac{1}{\sqrt{d}}\,\overline{S}_{s(d):t}\big\}\,
\big]}
{{\Exp}\,\big[\,\exp\big\{\tfrac{2}{\sqrt{d}}\,
 \overline{S}_{s(d):t} \big\}\,\big]}\Bigg)^d
\end{equation*}
for $\overline{S}_{s(d):t} = S_{s(d):t}-\Exp\,[\,S_{s(d):t}\,]$.
We will use the notation
`$h_d(t)\rightarrow_{t} h(t)$' to denote convergence, as $d\rightarrow\infty$,
uniformly for all $t$ in $[s(d),1]$, that is
$\sup_{t\in [s(d),t]}|h_{d}(t)-h_t|\rightarrow 0$.
We will aim at proving, using the results in Proposition \ref{pr:unif}, that:
\begin{equation}
\label{eq:rrr}
f_d(s(d),t) \rightarrow_{t}  e^{-\sigma_{s:t}^2}\ ,
\end{equation}
or, equivalently, that $\sup_{t\in[s(d),1]}|f_d(s(d),t)-e^{-\sigma_{s:t}^2}|\rightarrow 0$,
under the convention that $\sigma_{s:t}^2\equiv 0$ for $t\le s$.
Once we have obtained this, the required result will follow directly by induction. To see that,
note that for proving that $t_{1}(d)\rightarrow t_1$ we will use the established result for $s(d)\equiv \phi_0$:
uniform convergence of $f_d(\phi_0,t)$ to $e^{-\sigma_{\phi_0:t}^2}$ together with the fact that $e^{-\sigma_{\phi_0:t}^2}$ is decreasing in $t$ will give directly that the hitting time of the threshold $a$ for $f_d(\phi_0,t)$ will converge to that
of $e^{-\sigma_{\phi_0:t}^2}$. Now, assuming we have proved that $t_n(d)\rightarrow t_{n}$, we will then use the
established uniform convergence result for $s(d)=t_{n}(d)$ to obtain directly that
$t_{n+1}(d)\rightarrow t_{n+1}$.

We will now establish (\ref{eq:rrr}).
% where
% %
% $$
% \overline{W}_1(d) =
% \frac{1}{\sqrt{d}}
% \sum_{i=0}^{d-1} [g(X_{i,1})-\mathbb{E}[g(X_{i,1})]].
% $$
Note that we have, by construction:
$
\Exp\,[\,\overline{S}_{s(d):t}\,]=0\ .
$
%
% and we have already proven that\footnote{\textcolor{red}{We will have to add this}}:
% %
% \begin{equation}
% \label{eq:tt}
% {\Exp}\,[\,|\tfrac{S_{d,t_0(d),t}}{\sqrt{d}}|\,] \stackrel{\mathcal{U}_t}{\longrightarrow} 0\ ,\quad
% {\Exp}\,[\,S^2_{d,     t_0(d),t}\,]\stackrel{\mathcal{U}_t}{\longrightarrow}
%  \sigma_{t_0:t}^2\,\ , \quad
% \sup_{d\ge 1,\,t\in[t_0(d),1]}\,{\Exp}\,[\,S^{2+\epsilon}_{d,t_0(d),t}\,] < \infty \ .
% \end{equation}
% %
% % This is as
% % $$
% % \overline{W}_1(d) \Rightarrow \mathcal{N}(0,\sigma^2_{\star})
% % $$
% % % and
% % $$
% % \sup_d \overline{\mathbb{E}}[\overline{W}_1(d)^{2+\delta}] < \infty
% % $$
% % (a simple corollary of Lemma 3.1 in the paper).
We use directly Taylor expansions to obtain for any fixed $t\in[s(d),1]$:
\begin{align}
e^{\frac{2}{\sqrt{d}}\overline{S}_{s(d):t}} & =
1 + \tfrac{2}{\sqrt{d}}\,\overline{S}_{s(d):t}  + \tfrac{2}{d}\,\overline{S}_{s(d):t}^2
e^{2\zeta_{d,t}} \ ;\label{eq:t1}\\
e^{\frac{1}{\sqrt{d}}\overline{S}_{s(d):t}} & =  1 +
\tfrac{1}{\sqrt{d}}\,\overline{S}_{s(d):t}  + \tfrac{1}{2d}\,\overline{S}_{s(d):t}^2\,e^{\zeta_{d,t}'}\ ,
\label{eq:t2}
\end{align}
where
$
\zeta_{d,t},\zeta_{d,t}' \in \big[\,\tfrac{1}{\sqrt{d}}\,\overline{S}_{s(d):t}\wedge 0\,
,\,\,\tfrac{1}{\sqrt{d}}\overline{S}_{s(d):t}\vee 0\,\big]\ .
$
Note here that since $g$ is upper bounded and $\sup_{n,d}\Exp\,[\,|g(X_{n,1}(d))|\,]<\infty$,
we have that $\tfrac{1}{\sqrt{d}}\overline{S}_{s(d):t}$ is upper bounded. Thus, we obtain directly that:
\begin{equation*}
\xi_{d,t}\le M\,,\,\,\zeta'_{d,t} \le M\ ; \quad |\zeta_{d,t}| +|\zeta'_{d,t}|\le M\,|\tfrac{1}{\sqrt{d}}\,\overline{S}_{s(d):t}| \ .
\end{equation*}
 Taking expectations in (\ref{eq:t1}):
\begin{equation*}
\Exp\,[\,e^{\frac{2}{\sqrt{d}}\overline{S}_{s(d):t}}\,] =
1 + \tfrac{2}{d}\,{\Exp}\,[\,\overline{S}_{s(d):t}^2\,
e^{2\zeta_{d,t}}\,]\ .
\end{equation*}
%
%Using the three results in Proposition \ref{pr:unif} and the fact that $\Exp\,[\,|e^{2\xi_{d,t}}-1|^{q}\,]\le M(q)
%\,\Exp\,[\,|\xi_{d,t}|\,]$ for any $q\ge 1$ due to $x\mapsto |e^{2x}-1|^q$ being Lipschitz continuous  on %$(-\infty,M\,]$,
%we get that:
%
% \begin{eqnarray*}
% \frac{1}{\sqrt{d}}\overline{W}_1(d) & \rightarrow & 0 \quad \textrm{in}~L_1(\overline{\mathbb{P}})\\
% \overline{W}_1(d)^2 & \rightarrow & \sigma_{\star}^2 \quad \textrm{in}~L_1(\overline{\mathbb{P}})
% \end{eqnarray*}
% and $\xi_d$ is bounded with
% $$
% \sup_d \overline{\mathbb{E}}\bigg[\overline{W}_1(d)^{2+\delta}\bigg]<\infty
% $$
Now consider  the term:
\begin{equation*}
a_d(t):={\Exp}\,[\,\overline{S}_{s(d):t}^2\,
e^{2\zeta_{d,t}}\,]  = {\Exp}\,[\,\overline{S}_{s(d):t}^2\,] +
{\Exp}\,[\,\overline{S}_{s(d):t}^2\,(e^{2\zeta_{d,t}}-1)\,]\ .
       % \rightarrow_t \sigma_{s:t}^2\ .
\end{equation*}
Using Holder's inequality and the fact that $\Exp\,[\,|e^{2\zeta_{d,t}}-1|^{q}\,]\le M(q)
\,\Exp\,[\,|\zeta_{d,t}|\,]$ for any $q\ge 1,$
via the Lipschitz continuity of
 $x\mapsto |e^{2x}-1|^q$ on $(-\infty,M]$, we obtain that for $\epsilon>0$ as in Proposition \ref{pr:unif}(iii):
\begin{align*}
|\,{\Exp}\,[\,\overline{S}_{s(d):t}^2\,(e^{2\zeta_{d,t}}-1)\,]\,| &\le
\Exp^{\frac{2}{2+\epsilon}}\,[\,\overline{S}_{s(d):t}^{2+\epsilon}\,]\,\,
\Exp^{\frac{\epsilon}{2+\epsilon}}\,[\,|e^{2\zeta_{d,t}}-1|^{\frac{2+\epsilon}{\epsilon}}\,] \\
&\le M\,\Exp^{\frac{\epsilon}{2+\epsilon}}\,[\,|\zeta_{d,t}|\,]\rightarrow_t 0
\end{align*}
the last limit following from Proposition \ref{pr:unif}(i). Thus, using also Proposition \ref{pr:unif}(ii)-(iii),
we have proven that
$
a_{d}(t) \rightarrow_t \sigma^{2}_{s:t}\ .
$
Note now that:
\begin{equation*}
|\big( 1 +  \tfrac{2}{d}\,a_{d}(t)\big)^d -\big( 1 +  \tfrac{2\sigma_{s:t}^2}{d}\big)^d |
\le M\,|a_{d}(t)-\sigma_{s:t}^2|\ ;\quad  \big( 1 +  \tfrac{2\sigma_{s:t}^2}{d}\big)^d
\rightarrow_t e^{2\sigma^2_{s:t}} \ ,
\end{equation*}
the first result following from the derivative of $x\mapsto \big(1 +  \tfrac{2x}{d}\big)^d$ being bounded
for $x\in[0,M]$. Thus we have proven that:
$
\big(\,\Exp\,[\,e^{\tfrac{2}{\sqrt{d}}\,
 \overline{S}_{s(d):t}}\,]\,\big)^d\rightarrow_t e^{2\sigma_{s:t}^2}\ .
$
Using similar manipulations and the Taylor expansion (\ref{eq:t2}) we obtain that:
\begin{equation*}
\big(\,\Exp^{2}\,[\,
e^{\frac{1}{\sqrt{d}}\,\overline{S}_{s(d):t}}\,]\,\big)^d \rightarrow_t e^{\sigma_{s:t}^2}\ .
\end{equation*}
Taking the ratio, the uniform convergence result in (\ref{eq:rrr}) is proved.
\end{proof}

\subsection{Results for Theorems \ref{theorem:limit_adaptive} and \ref{theorem:error_adaptive}}
\label{appendix:adaptive}

To prove Theorems \ref{theorem:limit_adaptive} and \ref{theorem:error_adaptive}, we  will first require some technical lemmas. Here the equally weighted
$d$-dimensional resampled (at the deterministic time instances $t_{k}(d)$)
particles are written with a prime notation; so
$X_{l_d(t_{k}(d)),j}^{\prime,i}$
will denote the $j$-th co-ordinate of the $i$-th particle, immediately \emph{after}
the resampling procedure at $t_k(d)$.
% We also use the notation:
% \begin{equation*}
%  X_{l_d(t_{k-1}(d)),1:d}^{\prime,1:N} :=
% \{\,X_{l_d(t_{k-1}(d)),j}^{\prime,i},\,1\le i\le N, 1\le j\le d\,\}\ .
% \end{equation*}
% We consider the following event, for any $\epsilon>0$, $k\geq 1$:
% %
% \begin{align}
% A_{\epsilon,N,k}(d) = \bigcap_{\substack{1\le i\le N  \\ 1\le j\le d} } \{\,
% %       x_{l_d(t_{k-1}(d)),1:d}^{\prime,1:N}
% V(X_{l_d(t_{k-1}(d)),j}^{\prime,i})^{r}/\,d^{1/3}\leq \epsilon\,\}\ .
% \label{eq:good}
% \end{align}
% %
% The dependence structure between events of different dimension will not affect our results (for simplicity,
% one could for instance think that events over different choices of $d$ are independent).
% It should be noted that if one conditions on starting an SMC algorithm (up-to a resampling time), then our fCLT will apply for every collection of points on the set.

\begin{prop}
\label{lemma:technical_lemma2}
Assume (A\ref{hyp:A}(i)(ii)) and let $k\in\{1,\dots,m^*\}$. Then,
there exists  an $M(k)<\infty$  such that for any $N\geq 1$,
$d\geq 1$, $i\in\{1,\dots,N\}$, $j\in\{1,\dots,d\}$:
\begin{equation*}
\mathbb{E}\,[\,V(X_{l_d(t_k(d)),j}^{\prime,i})\,] \leq M(k) N^k.
\end{equation*}
\end{prop}

\begin{proof}
We will use an inductive proof on the resampling times (assumed to be deterministic).
It is first remarked (using Lemma \ref{lem:growth})(iii)) that for every $k\in\{1,\dots,m^*\}$:
\begin{equation}
\mathbb{E}\,[\,V(X_{l_d(t_k(d)),j}^{i})\,|\,\mathscr{F}_{t_{k-1}(d)}^{\prime,N}\,] \leq M\, V(x_{l_d(t_{k-1}(d)),j}^{\prime, i})
\label{eq:main_bound_tech_lemma}
\end{equation}
where $\mathscr{F}_{t_{k-1}(d)}^{\prime,N}$ is the filtration generated by the particle system up-to and including the $(k-1)^{th}$ resampling time and $M<\infty$ does not depend upon $t_{k}(d)$, $t_{k-1}(d)$ or indeed $d$.

At the first resampling time, we have (averaging over the resampling index)
that
$$
\mathbb{E}\,[\,V(X_{l_d(t_1(d)),j}^{\prime,i})\,|\,\mathscr{F}_{t_{1}(d)}^{N}\,] = \sum_{i=1}^N \overline{w}_{l_d(t_{1}(d))}(x^i_{l_d(t_{0}(d)):l_d(t_{1}(d))-1})\,V(x_{l_d(t_{1}(d)),j}^i)
$$
where $\mathscr{F}_{t_{1}(d)}^{N}$ is the filtration generated by the particle system up-to the $1$st resampling time (but excluding resampling)
and $\overline{w}_{l_d(t_{1}(d))}(x^i_{l_d(t_{0}(d)):l_d(t_{1}(d))-1})$ is the normalized importance weight.
Now, clearly (due to normalised weights be bounded by 1):
$$
\mathbb{E}\,[\,V(X_{l_d(t_1(d)),j}^{\prime,i})\,|\,\mathscr{F}_{t_{1}(d)}^{N}\,] \leq
\sum_{i=1}^N  V(x_{l_d(t_{1}(d)),j}^i)
$$
and, via \eqref{eq:main_bound_tech_lemma},
$\mathbb{E}\,[\,V(X_{l_d(t_1(d),j}^{\prime,i})\,] \leq NM$
which gives the result for the first resampling time.

Using induction, if we assume that the result holds at the $(k-1)^{th}$  time we resample ($k\geq 2$), it follows that (for $\mathscr{F}_{t_{k}(d)}^{N}$ being the filtration generated by the particle system up-to the $k$-th resampling time, but excluding resampling):
\begin{align*}
\mathbb{E}\,[\,V(X_{l_d(t_k(d)),j}^{\prime,i})\,|\,\mathscr{F}_{t_{k}(d)}^{N}\,] &= \sum_{i=1}^N
\overline{w}_{l_d(t_{k}(d))}(x^i_{l_d(t_{k-1}(d)):l_d(t_{k}(d))-1})\,V(x_{l_d(t_{k}(d)),j}^i)\\
&\leq \sum_{i=1}^N  V(x_{l_d(t_{k}(d)),j}^i) \ .
\end{align*}
Thus, via \eqref{eq:main_bound_tech_lemma} and the exchangeability of the particle and dimension index,
we obtain that
$$\mathbb{E}\,[\,V(X_{l_d(t_k(d),j}^{\prime,i})\,]
\leq NM\,\mathbb{E}\,[\,V(X_{l_d(t_{k-1}(d),j}^{\prime,i})\,]\ .$$
The proof now follows directly.
\end{proof}

\begin{prop}
\label{lem:extra_result}
Assume (A\ref{hyp:A}(i)(ii), A\ref{hyp:B}). Let $\varphi\in\mathscr{L}_{V^r}$, $r\in[0,\tfrac{1}{2})$. Then for any fixed $N$, any $k\in \{1,\dots,m^*\}$ and any $i\in\{1,\dots,N\}$ we have
\begin{equation*}
\frac{1}{d}\,\sum_{j=1}^d \varphi(X_{l_{d}(t_k(d)),j}^{\prime,i})\rightarrow  \pi_{t_k}(\varphi) \ , \quad \textrm{in }\,\mathbb{L}_1\ .
\end{equation*}
\end{prop}

\begin{proof}
We distinct between two cases:  $k=1$ and $k>1$.
When $k=1$, due to the boundedness of the normalised weights and the exchangeability of the particle indices we have that:
\begin{equation}
\label{eq:three}
\mathbb{E}\,\big|\,\frac{1}{d}\sum_{j=1}^d \varphi(X_{l_{d}(t_1(d)),j}^{\prime,i})-\pi_{t_1}(\varphi)\,\big|
\le N\,\mathbb{E}\,\big|\,\frac{1}{d}\sum_{j=1}^d \varphi(X_{l_{d}(t_1(d)),j}^{i})-\pi_{t_1}(\varphi)\,\big|
\end{equation}
Adding and subtracting the term $\Exp\,[\,\varphi(X_{l_{d}(t_1(d)),j}^{i})\,]$ we obtain that the expectation
on the R.H.S.\@ of the above equation is bounded by:
\begin{equation}
\label{eq:two}
\Exp\,\big|\,\frac{1}{d}\sum_{j=1}^d \varphi(X_{l_{d}(t_1(d)),j}^{i})-\mathbb{E}\,[\,\varphi(X_{l_{d}(t_1(d)),j}^{i})\,]\,\big| +
|\,\mathbb{E}\,[\,\varphi(X_{l_{d}(t_1(d)),j}^{i})]-\pi_{t_1}(\varphi)\,| \ .
\end{equation}
For the first term, due to the independency across dimension, considering second moments we get
the upper bound:
\begin{equation*}
\frac{1}{\sqrt{d}}\,\Exp^{1/2}\,\big[\,\big(\,
\varphi(X_{l_{d}(t_1(d)),j}^{i})-\Exp\,[\,\varphi(X_{l_{d}(t_1(d)),j}^{i})\,]\,\big)^2\,\big]\ .
\end{equation*}
As $\varphi\in\mathscr{L}_{V^r}$ with $r\le 1/2$ the argument of the expectation is upper-bounded by $M V(X_{l_{d}(t_1(d)),j}^{i})$ whose expectation is controlled via Lemma \ref{lem:growth}(iii).
Thus the above quantity is $\mathcal{O}(d^{-1/2})$.
For  the second term in (\ref{eq:two}) we can use directly Proposition \ref{prop:prop}
(for time sequences required there selected as $s(d)\equiv \phi_0$ and $t(d)\equiv t_{1}(d)$)
to show also that this term will vanish in the limit $d\rightarrow \infty$.

The general case with $k>1$ is similar, but requires some additional arguments as resampling
eliminates the i.i.d.\@ property. Again, integrating out the resampling index as in (\ref{eq:three})
we are left with the quantity:
\begin{equation*}
\mathbb{E}\,\big|\,
\frac{1}{d}\sum_{j=1}^d \varphi(X_{l_{d}(t_k(d)),j}^{i})-\pi_{t_k}(\varphi)\,\big| \ .
\end{equation*}
Adding and subtracting
$\frac{1}{d}\sum_{j=1}^d\mathbb{E}_{X_{l_{d}(t_{k-1}(d)),j}^{\prime,i}}[\,\varphi(X_{l_{d}(t_k(d)),j}^{i})\,]$
within the expectation, the above quantity is upper bounded by:
\begin{align}
\mathbb{E}\,\big|\,\frac{1}{d}\sum_{j=1}^d &\varphi(X_{l_{d}(t_k(d)),j}^{i})-\frac{1}{d}\sum_{j=1}^d
\mathbb{E}_{X_{l_{d}(t_{k-1}(d)),j}^{\prime,i}}\,[\,\varphi(X_{l_{d}(t_k(d)),j}^{i})\,]\,\big| \quad + \nonumber\\
&\mathbb{E}\,\big|\,\frac{1}{d}\sum_{j=1}^d\mathbb{E}_{X_{l_{d}(t_{k-1}(d)),j}^{\prime,i}}\,[\,\varphi(X_{l_{d}(t_k(d)),j}^{i})\,]-\pi_{t_k}(\varphi)\,\big| \ . \label{eq:five}
\end{align}
For the first of these two terms, due to \emph{conditional} independency across dimension
and exchangeability in the dimensionality index $j$, looking at the second moment we obtain the upper bound:
\begin{equation*}
\frac{1}{\sqrt{d}}
\mathbb{E}^{1/2}\big[\,\big(\varphi(X_{l_{d}(t_k(d)),j}^{i})-
\mathbb{E}_{X_{l_{d}(t_{k-1}(d)),j}^{\prime,i}}\,[\,\varphi(X_{l_{d}(t_k(d)),j}^{i})\,]\,\big)^2\,\big]\ .
\end{equation*}
Since $|\varphi(x)|\,\le M\, V^{r}(x) $ with $r\le\tfrac{1}{2}$, the variable in the expectation above is upper bounded by
$M(V(X_{l_{d}(t_{k}(d)),j}^{\prime,i})+V(X_{l_{d}(t_{k-1}(d)),j}^{\prime,i}))$ which due to Proposition \ref{lemma:technical_lemma2} is bounded in expectation by some $M(N,k)$.
Thus, the first term in (\ref{eq:five}) is $\mathcal{O}(d^{-1/2})$.
The second term in (\ref{eq:five}) now, due to exchangeability over $j$, is upper bounded by
%
%\begin{equation*}
$\Exp\,
\big|\,\mathbb{E}_{X_{l_{d}(t_{k-1}(d)),j}^{\prime,i}}\,
[\,\varphi(X_{l_{d}(t_k(d)),j}^{i})\,]-\pi_{t_k}(\varphi)\,\big|$,
%\end{equation*}
%
which again due to Proposition \ref{prop:prop} vanishes in the limit $d\rightarrow \infty$.
% $$
% |\mathbb{E}_{X_{l_{d}(t_{k-1}(d)),j}^{',i}}[\varphi(X_{l_{d}(t_k(d)),j}^{i})]-\pi_{t_k}^c(\varphi)|^p
% \leq M \frac{V(X_{l_{d}(t_{k-1}(d)),j}^{',i})^r}{\sqrt{d}}
% $$
% and thus by Lemma C.2.
% $$
% \mathbb{E}[|\frac{1}{d}\sum_{j=1}^d\mathbb{E}_{X_{l_{d}(t_{k-1}(d)),j}^{',i}}[\varphi(X_{l_{d}(t_k(d)),j}^{i})]-\pi_{t_k}^c(\varphi)|^p]^{1/p}
% \leq
% \frac{M(k-1)N^{k-1}}{\sqrt{d}}.
% $$
% This leads to the upper-bound:
% $$
% \mathbb{E}[|\frac{1}{d}\sum_{j=1}^d \varphi(X_{l_{d}(t_k(d)),j}^{i})-\pi_{t_k}^c(\varphi)|^p]^{1/p}
% \leq
% \frac{M(p,k)N^{k}}{\sqrt{d}}
% $$
% and hence we conclude, via the First-Borel Cantelli lemma.
\end{proof}

For the Markov chain $X^{i}_{n,j}$ considered on the instances $n_1\le n\le n_2$ we will henceforth use  the notation $\Exp_{\pi_s}\,[\,g(X^{i}_{n,j})\,]$ to specify that we impose the initial distribution $X^{i}_{n_1,j}\sim \pi_s$.

\begin{prop}
\label{lem:extra_new_result}
Assume (A\ref{hyp:A}-\ref{hyp:B}) and that $g\in\mathscr{L}_{V^r}$ with $r\in [0,\tfrac{1}{2})$.
For  $k\in \{1,\dots,m^*\}$, $i\in\{1,\dots,N\}$ and a
sequence $s_k(d)$ with $s_k(d)> t_{k-1}(d)$ and $s_k(d)\rightarrow s_k > t_{k-1}$ we define:
\begin{equation*}
E_{i,j} = \sum_{n}\big\{ \,\mathbb{E}_{X_{l_d(t_{k-1}(d)),j}^{\prime,i}}\,[\,g(X_{n,j}^i)\,]
- \mathbb{E}_{\pi_{t_{k-1}}}\,\big[\,g(X_{n,j}^i)\,\big]\,\big\}\ ,\quad 1\le j\le d \  ,
\end{equation*}
for subscript $n$ in the range
$l_d(t_{k-1}(d))\le n \le l_d(s_{k}(d))-1$.
Then, we have that:
\begin{equation*}
\frac{1}{d}\,\sum_{j=1}^{d}\,E_{i,j} \rightarrow 0\ , \quad \textrm{in }\,\mathbb{L}_1\ .
\end{equation*}
\end{prop}

\begin{proof}
We will make use of the Poisson equation and employ the decomposition (\ref{eq:crude})
used in the proof of Theorem \ref{th:decompose}. In particular, a straight-forward calculation gives that:
\begin{align}
R_{i,j}=  \sum_{n=n_1+1}^{n_2}& \big\{\,
(\,\mathbb{E}_{X_{n_1,j}}-\mathbb{E}_{\pi_{t_{k-1}}}\,)[\,\widehat{g}_n(X_{n-1,j}^i)-\widehat{g}_{n-1}(X_{n-1,j}^i)\,]
\,\big\} \nonumber \\
&
+ (\,\mathbb{E}_{X_{n_1,j}}-\mathbb{E}_{\pi_{t_{k-1}}}\,)[\,g(X_{n_2,j})-\widehat{g}_{n_2}(X_{n_2,j})\,]
+\widehat{g}_{n_1}(X_{n_1,j}) - \pi_{t_{k-1}}(\widehat{g}_{n_1}) \ ,\label{eq:rr}
\end{align}
where $\widehat{g}_{n}=\mathcal{P}(g,k_{n},\pi_{n})$, and we have set:
\begin{equation*}
n_1 =  l_d(t_{k-1}(d))\ ;\quad n_2 = l_d(s_{k}(d))-1 \ ;\, X_{n_1,j}\equiv X_{l_d(t_{k-1}(d)),j}^{\prime,i} \ .
\end{equation*}
It is remarked that the martingale term in the original expansion (\ref{eq:crude}) has expectation $0$, so is not involved in our
manipulations.
We will first deal with the sum in the first line of (\ref{eq:rr}), that is (when taking into account
the averaging over $j$) with:
% $$
% \frac{1}{d}\sum_{j=1}^d \sum_{n= l_d(t_{k-1}(d))+1}^{l_d(t_k(d))}
% \{\mathbb{E}_{X_{l_d(t_{k-1}(d))}^{',i}}[\widehat{\varphi}_n(X_{n-1,j})-\varphi_{n-1}(X_{n-1,j})] -
% \mathbb{E}_{\pi_{t_{k-1}}^c}[\widehat{\varphi}_n(X_{n-1,j})-\varphi_{n-1}(X_{n-1,j})]\}
% $$
% which can be written in the short-hand (omitting the double-subscripts)
\begin{equation*}
A_d : =\frac{1}{d}\sum_{j=1}^d \sum_{n=n_1+1}^{n_2}[\,\delta_{_{X_{n_1,j}}}-\pi_{t_{k-1}}\,]
\big(\,(k_{n_1+1:n})[\widehat{g}_{n}-\widehat{g}_{n-1}]\,\big) \ .
\end{equation*}
Now each summand in the above double sum is upper bounded by
\begin{equation*}
\frac{M}{d}\,\|\,[\delta_{X_{n_1,j}}-\pi_{t_{k-1}}]
(k_{n_1+1:n})\,\|_{V^r}\ .
\end{equation*}
To bound this $V^r$-norm one can apply Theorem 8 of \cite{doucchains}; here, under (A\ref{hyp:A}-\ref{hyp:B})
we have that either:
\begin{equation}
\label{eq:bbound}
\|\,[\delta_{X_{n_1,j}}-\pi_{t_{k-1}}]
(k_{n_1+1:n})\,\|_{V^r} \le
M\rho^{n-n_1} V(X_{n_1,j})^r + M'\zeta^{n-n_1}
\end{equation}
for some $\rho,\zeta\in(0,1)$, $0<M,M'<\infty$, when $B_{j-1,n}$ (of that paper) is 1. Or, if $B_{j-1,n}>1$, one has the bound
\begin{equation*}
\|\,[\delta_{X_{n_1,j}}-\pi_{t_{k-1}}]
(k_{n_1+1:n})\,\|_{V^r} \le
M\rho^{\lfloor j^*(n-n_1)\rfloor} V(X_{n_1,j})^r + M'\zeta^{\lfloor j^*(n-n_1)\rfloor}
\end{equation*}
with $j^*$ as the final equation of \cite[pp.~1650]{doucchains}.
(Note that this follows from a uniform in time drift condition which follows from Proposition 4 of \cite{doucchains} (via (A\ref{hyp:A}))). By summing up first over $n$ and then over $j$ (and dividing
with $d$), using also Proposition \ref{lem:extra_result} along the way to control $\sum_{j}V(X_{n_1,j})^r/d$, we have that:
\begin{equation*}
A_{d} \rightarrow 0\ , \quad \textrm{in }\,\mathbb{L}_1\ .
\end{equation*}

A similar use of the bound in (\ref{eq:bbound}) and Proposition \ref{lem:extra_result} can give directly that the second term in (\ref{eq:rr}) will
vanish in the limit when summing up over $j$ and dividing with $d$.
Finally, for the last term in (\ref{eq:rr}): Proposition \ref{lem:extra_result} is not directly applicable here as one has to address the fact that the function $\widehat{g}_{n_1}$ depends on $d$.
Using Lemma \ref{lem:growth} (ii), one can replace
$\widehat{g}_{n_1}\equiv \widehat{g}_{l_d(t_{k-1}(d))}$ by $\widehat{g}_{t_{k-1}}$ and then apply
Proposition \ref{lem:extra_result} and the fact that $t_{k-1}(d)\rightarrow t_{k-1}$
to show that the remainder
term goes to zero in $\mathbb{L}_1$ (when averaging over $j$).
The proof is now complete.

\end{proof}

\begin{proof}[Proof of Theorem  \ref{theorem:limit_adaptive}]

Recall the definition of the ESS:
\begin{equation*}
\textrm{ESS}_{(t_{k-1}(d),s_k(d))}(N)  = \frac{\big(\sum_{i=1}^{N}e^{(1-\phi_0){a}^{i}(d)}\big)^2}{\sum_{i=1}^{N}e^{2(1-\phi_0)a^{i}(d)}}\ .
\end{equation*}
where we have defined:
\begin{equation*}
a^{i}(d) = \frac{1}{d}\sum_{j=1}^{d} \{ \overline{G}_{i,j} + E_{i,j}\}
\end{equation*}
with:
\begin{align*}
\overline{G}_{i,j} &=  \sum_{n} \big\{
\,g(X_{n,j}^i)-\Exp_{X_{l_d(t_{k-1}(d)),j}^{\prime,i}}
[\,g(X_{n,j}^i)\,]\,\big\} \ ; \\
E_{i,j} &= \sum_{n}\big\{ \,\mathbb{E}_{X_{l_d(t_{k-1}(d)),j}^{\prime,i}}[\,g(X_{n,j}^i)\,]
- \mathbb{E}_{\pi_{t_{k-1}}}\,[\,g(X_{n,j}^i)\,]\,\big\}\ ,
\end{align*}
for subscript $n$ in the range
$l_d(t_{k-1}(d))\le n \le l_d(s_{k}(d))-1$.
From Proposition \ref{lem:extra_new_result} we get directly that $\sum_{j=1}^{d}E_{i,j}/d\rightarrow 0$ (in $\mathbb{L}_1$). Thus, we are left with  $\overline{G}_{i,j}$ which corresponds to a martingale under the filtration
we define below. In the below proof, we consider the weak convergence for a single particle. However, it possible to prove a multivariate CLT for all the particles using the Cramer-Wold device. This calculation is very similar to that given below and is hence omitted.

%
% and their collection over different $d$:
% %
% \begin{equation*}
% X_{k}^{\prime} := \{ X_{l_d(t_{k-1}(d)),1:d}^{\prime,1:N}\,;\,d\ge 1\}\ .
% \end{equation*}
% %
% We will prove an CLT a.s.\@ in the choice of $X_{k}^{\prime}$.
% Using Lemma \ref{lemma:technical_lemma1}, it suffices to  prove such weak %convergence
% conditionally on any $X_{k}^{\prime}$
% within the set $A:=\lim\inf_d A_{\epsilon,N,k}(d)$.
% By its definition, $A$ occurs if and only if
% \emph{all} events $A_{\epsilon,N,k}(d)$
% occur for large enough $d$, thus allowing us to use the bounds
% $V(X_{l_d(t_{k-1}(d)),j}^{\prime,i})^{r}\le \epsilon\,d^{1/3}$
% when conditioning on $X_{l_d(t_{k-1}(d)),j}^{\prime,i}$ and taking limits %$d\rightarrow\infty$.
%

Consider some chosen particle $i$, with $1\le i \le N$.
For any $d\ge 1$ we define the filtration $\mathcal{G}_{0,d}\subseteq \mathcal{G}_{1,d}\subseteq\cdots \subseteq
\mathcal{G}_{d,d}$ as follows:
\begin{align}
\mathcal{G}_{0,d} &= \sigma(X_{l_d(t_{k-1}(d)),j}^{\prime,l},1\le j\le d,\,1\le l\le N) \ ;\nonumber \\
\mathcal{G}_{j,d} &=  \mathcal{G}_{j-1,d} \bigvee  \sigma(X_{n,j}^{i},\,l_d(t_{k-1}(d))\le n\le l_d(s_{k}(d))-1) \ ,
\quad j \ge 1\ .
\label{eq:filtr}
\end{align}
That is,  $\sigma$-algebra $\mathcal{G}_{0,d}$ contains the information about \emph{all} particles, along
\emph{all} $d$ co-ordinates until (and including) the resampling step; then the rest of the filtration
is build up by adding information for the subsequent trajectory of the various co-ordinates. Critically,
conditionally on $\mathcal{G}_{0,d}$ these trajectories are independent. One can now easily check that
\begin{equation*}
\beta^{i}_{j}(d) = \frac{1}{d} \sum_{k=1}^{j} \overline{G}_{i,k}\ ,\quad 1\le j \le d\ ,
\end{equation*}
is a martingale w.r.t.\@ the filtration in (\ref{eq:filtr}).
Now, to apply the CLT for triangular martingale arrays, we will show that for every $i\in\{1,\dots,N\}$:
\begin{itemize}
\item[a)]{That in $\mathbb{L}_1$:
\begin{equation*}
\lim_{d\rightarrow\infty}
%\frac{1}{d^2}\sum_{j=1}^d \mathbb{E}[\overline{M}_{i,j}(d)^2|X_{l_d(t_{k-1}(d)),j}^{\prime,i}] =
\frac{1}{d^2}\sum_{j=1}^d \mathbb{E}\,[\,\overline{G}_{i,j}^2\,|\,\mathcal{G}_{j-1,d}\,] =
\sigma^2_{t_{k-1}:s_k}
\end{equation*}
}
\item[b)]{For any $\epsilon>0$, that in $\mathbb{L}_1$:
\begin{equation*}
\lim_{d\rightarrow\infty}
\frac{1}{d^2}\sum_{j=1}^d \Exp\,[\,\overline{G}_{i,j}^2\mathbb{I}_{|\overline{G}_{i,j}|\geq \epsilon d}\,|\,\mathcal{G}_{j-1,d}\,] = 0\ .
\end{equation*}
}
\end{itemize}
This will allow us to show that $(1-\phi_0)a^{i}(d)$ will converge weakly to the appropriate normal random variable. Notice, that due to the conditional independency mentioned above and the definition of the filtration in (\ref{eq:filtr}) we in fact have
that:
\begin{align*}
 \mathbb{E}\,[\,\overline{G}_{i,j}^2\,|\,\mathcal{G}_{j-1,d}\,] &\equiv \mathbb{E}_{\Con}\,[\,\overline{G}_{i,j}^2\,] \ ; \\  \Exp\,[\,\overline{G}_{i,j}^2\mathbb{I}_{|\overline{G}_{i,j}|\geq \epsilon d}\,|\,\mathcal{G}_{j-1,d}\,] &\equiv
\Exp_{\Con}[\,\overline{G}_{i,j}^2\mathbb{I}_{|\overline{G}_{i,j}|\geq \epsilon d}\,]\ .
\end{align*}
We make the following definition:
\begin{equation*}
% \overline{W}_{i,j}(d) & = & \frac{(1-\phi_0)}{\sqrt{d}}M_{i,j}\\
% W_{i,j}(d) & = &\frac{(1-\phi_0)}{\sqrt{d}} \sum_{n} \big\{ g(X_{n,j}^i) -\pi_{n}(g) \big\}
G_{i,j}  = \sum_{n} \big\{ g(X_{n,j}^i) -\pi_{n}(g) \big\} \equiv M_{n_1:n_2,i,j} + R_{n_1:n_2,i,j} \ ,
\end{equation*}
(for convenience we have set $n_1 = l_d(t_{k-1}(d))$ and $n_2=l_d(s_{k}(d))-1$)
with the terms $M_{n_1:n_2,i,j}$ and $R_{n_1:n_2,i,j}$ defined as in Theorem \ref{th:decompose} with the extra
subscripts indicating the number of particle and the co-ordinate.
Notice that $\overline{G}_{i,j}=G_{i,j}-\Exp_{\Con}[\,G_{i,j}\,]$.

We start with a). We first use the fact that:
\begin{equation*}
 \frac{1}{d^2}\sum_{j=1}^d\Exp_{\Con}\,[\,\overline{G}_{i,j}^2\,] -
\frac{1}{d^2}\sum_{j=1}^d \Exp_{\Con}\,[\,G_{i,j}^2\,] \rightarrow 0 \ ,\quad \textrm{in }\,\mathbb{L}_1\ .
\end{equation*}
To see that, simply note that the above difference is equal to:
\begin{equation*}
  \frac{1}{d^2}\sum_{j=1}^d\Exp^2_{\Con}\,[\,G_{i,j}\,] \equiv
\frac{1}{d^2}\sum_{j=1}^d\Exp^2_{\Con}\,[\,R_{i,j}\,] \le \frac{1}{d^2}\sum_{j=1}^{d}V(\Con)^{2r}
\end{equation*}
where we first used the fact that $M_{n_1:n_2,i,j}$ is a martingale (thus, of zero expectation) and then
Theorem \ref{th:decompose} to obtain the bound;
the bounding term vanishes due to Proposition \ref{lemma:technical_lemma2}.
We then have that:
\begin{align}
\frac{1}{d^2}&\sum_{j=1}^d \Exp_{\Con}\,[\,G_{i,j}^2\,]  =
\frac{1}{d^2}\sum_{j=1}^d \Exp_{\Con}\,[\,M_{i,j}^2+R_{i,j}^2 + 2\,M_{i,j}R_{i,j}\,]\nonumber  \\
& =  \frac{1}{d^2}\sum_{j=1}^d \Exp_{\Con}\,[\,M_{i,j}^2\,] + \mathcal{O}(d^{-1/2})\ .
\label{eq:aaa}
\end{align}
To yield the $\mathcal{O}(d^{-1/2})$ one can use the bound
$$\Exp_{\Con}\,[\,R_{i,j}^2\,]\le M\,V(\Con)^{2r}$$ from Theorem \ref{th:decompose},
and then (using Cauchy-Schw\"artz and Theorem \ref{th:decompose}):
\begin{align*}
|\,\Exp_{\Con}\,[\,M_{i,j}R_{i,j}\,]\,|&\le
\Exp^{1/2}_{\Con}\,[\,M_{i,j}^2\,]\cdot\Exp^{1/2}_{\Con}\,[\,R^2_{i,j}\,]\\
&\le M\,\sqrt{d}\,V(\Con)^{2r}\ .
\end{align*}
One then only needs to make use of Proposition \ref{lemma:technical_lemma2} to get (\ref{eq:aaa}).
%
% as the latter term will converge to $\sigma^2_{t_{k-1}:s_k}$. Now, we can decompose the latter convergence into the consideration of the term
% \begin{equation}
% |\frac{1}{d}\sum_{j=1}^d \mathbb{E}[W_{i,j}(d)^2|X_{l_d(t_{k-1}(d)),j}^{\prime,i}]-\frac{(1-\phi_0)^2}{d^2}\sum_{j=1}^d \sum_n \pi_{n}(\varphi_{n,d})|\label{eq:sig_conv}
% \end{equation}
% if $|\frac{1}{d}\sum_{j=1}^d \mathbb{E}[W_{i,j}(d)|X_{l_d(t_{k-1}(d))}^{\prime,i}]|$ goes to zero. This holds as
% \begin{eqnarray*}
% |\frac{1}{d}\sum_{j=1}^d \mathbb{E}[W_{i,j}(d)|X_{l_d(t_{k-1}(d)),j}^{\prime,i}]|
% & = & \frac{(1-\phi_0)}{d^{3/2}}|\sum_{j=1}^d\mathbb{E}[\sum_n \{g(X_{n,j}^i) -\pi_{n}(g)\}|X_{l_d(t_{k-1}(d)),j}^{\prime,i}]|\\
% & \leq & \frac{(1-\phi_0)}{d^{3/2}}\sum_{j=1}^d|
% \mathbb{E}[\overline{M}_{i,j} + \overline{R}_{i,j} |X_{l_d(t_{k-1}(d)),j}^{\prime,i}]|\\
% & = & \frac{(1-\phi_0)}{d^{3/2}}\sum_{j=1}^d|\mathbb{E}[\overline{R}_{i,j} |X_{l_d(t_{k-1}(d)),j}^{\prime,i}]|
% \end{eqnarray*}
%
% where we have applied the decomposition Theorem \ref{th:decompose} and $\overline{M}_{i,j}$, $\overline{R}_{i,j}$ are the Martingale and remainder terms of that Theorem. Then, using the bound in Theorem \ref{th:decompose}, it follows that
% $$
% |\frac{1}{d}\sum_{j=1}^d \mathbb{E}[W_{i,j}(d)|X_{l_d(t_{k-1}(d)),j}^{\prime,i}]|
% \leq \frac{M}{d^{3/2}}\sum_{j=1}^d V(X_{l_d(t_{k-1}(d)),j}^i)^r.
% $$
% Application of Lemma \ref{lem:extra_result} shows that the latter term goes to zero almost surely.
%
Now, using the analytical definition of $M_{i,j}$ from Theorem \ref{th:decompose} we have:
\begin{align}
 \frac{1}{d^2}\sum_{j=1}^d \Exp_{\Con}\,[\,M_{i,j}^2\,] &=
\frac{1}{d^2}\sum_{j=1}^d \sum_{n=n_1+1}^{n_2}\big\{\,
\Exp_{\Con}\,[\,\widehat{g}_{n}^2(X^i_{n,j})-k_n^2(\widehat{g}_{n})(X^i_{n-1,j})\,]\,\big\}
\nonumber\\
& =
\frac{1}{d^2}\sum_{j=1}^d \sum_{n=n_1}^{n_2-1}
\Exp_{\Con}\,[\,\varphi_{n+1}(X^i_{n,j})\,] =: A_d \label{eq:eee}
\end{align}
where:
\begin{equation*}
 \varphi_{n} = k_{n}(\widehat{g}_{n}^2)-[k_{n}(\widehat{g}_{n})]^2 \ ;\quad
\widehat{g}_n = \mathcal{P}(g, k_n, \pi_n) \ .
\end{equation*}
Using again the decomposition in Theorem \ref{th:decompose}, but now for $\varphi_n$ as above (which
due to Proposition \ref{prop:hyp} satisfies the requirements of Theorem \ref{th:decompose}),
we get that:
\begin{align*}
\big|\,\Exp_{\Con}\,\big[ \sum_{n=n_1}^{n_2-1}
\varphi_{n+1}(X^i_{n,j}) &- \pi_{n}(\varphi_{n+1})\,\big]\,\big| =
\big|\,\Exp_{\Con}\,[\,R^{'}_{n_1:(n_2-1),i,j}\,]\,\big| \\
&\le M\,V^{2r}(\Con)\ .
\end{align*}
Thus, continuing from (\ref{eq:eee}), and using the above bound and Proposition \ref{lemma:technical_lemma2}, we have:
\begin{equation}
\big|\,A_d - \frac{1}{d} \sum_{n=n_1}^{n_2-1} \pi_{n}(\varphi_{n+1})\,\big| = \mathcal{O}(d^{-1})\ .
\end{equation}
The proof for a) is completed using to the deterministic limit:
\begin{equation*}
\frac{1-\phi_0}{d}\sum_{n=n_1}^{n_2-1} \pi_{n}(\varphi_{n+1}) \rightarrow
\int_{t_{k-1}}^{s_k} \pi_{u}(\widehat{g}_{u}^2-k_{u}(\widehat{g}_{u})^2)du \ .
\end{equation*}

For b), we choose some $\delta$ so that $r(2+\delta)\le 1$, and obtain the following bound:
\begin{align*}
\Exp\,_{\Con}[\,\overline{G}_{i,j}^{2+\delta}\,] &\le
M\,\Exp\,_{\Con}[\,G_{i,j}^{2+\delta}\,] \\ &\le
M\,\Exp\,_{\Con}[\,M_{i,j}^{2+\delta}+R_{i,j}^{2+\delta}\,] \\
&\le M V(\Con)^{r(2+\delta)}\,d^{1+\frac{\delta}{2}} \ ,
\end{align*}
where for the last inequality we used the growth bounds in Theorem \ref{th:decompose}.
Also using, first, Holder inequality, then, Markov inequality and, finally, the above bound we find that:
\begin{align*}
\Exp_{\Con}[\,\overline{G}_{i,j}^2\mathbb{I}_{|\overline{G}_{i,j}|\geq \epsilon d}\,] &\le
\big(\,\Exp_{\Con}[\,\overline{G}^{2+\delta}_{i,j}\,]\,\big)^{\frac{2}{2+\delta}}
\cdot \big(\,\mathbb{P}_{\Con}\,[\,|\overline{G}_{i,j}|^{2+\delta}\geq (\epsilon d)^{2+\delta}\,]\,\big)^{\frac{\delta}{2+\delta}}\\
&\le M\, V(\Con)^{2r}\,d\,\cdot \,
\frac{V(\Con)^{r\delta}d^{\delta/2}}{(\epsilon\,d)^{\delta}} \ .
\end{align*}
Thus, we also have:
\begin{align*}
\frac{1}{d^2}\sum_{j=1}^d \Exp_{\Con}\,[\,\overline{G}_{i,j}^2\mathbb{I}_{|\overline{G}_{i,j}|\geq \epsilon d}\,]
\le M\, d^{-\delta/2}\,\frac{1}{d}\,\sum_{j=1}^{d}V(\Con)^{r(2+\delta)}\ .
\end{align*}
Due to Proposition \ref{lemma:technical_lemma2}, this bound proves part b).
%
%
%

% Due to the above arguments, we can conclude that
% \begin{equation*}
% \lim_{d\rightarrow\infty} \mathbb{E}\,\big[\,
% \textrm{ESS}_{(t_{k-1}(d),s_k(d))}(N)\,\big]
% =
% \mathbb{E}\bigg[\frac{[\sum_{i=1}^N e^{X_i^k}]^2}{\sum_{i=1}^N e^{2X_i^k}}\bigg].
% \end{equation*}

\end{proof}

\begin{proof}[Proof of Theorem  \ref{theorem:error_adaptive}]
The proof is similar to that of Theorem \ref{theorem:limit_adaptive} (as the final resampling time is strictly less than 1) and Theorem \ref{theo:mc_error}; it is omitted for brevity.
%and we only sketch the proof. The only real difference is the consideration of Proposition \ref{prop:prop}.
%When $t_{m^*}<1$, the latter result follows by considering $d$ large enough so that $u_d\geq l_d(t_{m^*}(d))$ and by applying Lemma \ref{lemma:technical_lemma2}.
%If $t_{m^*}=1$, then one can simply integrate out the resampling-time:
%$$
%\mathbb{E}\,[\,\varphi(X_{d,1})-\pi(\varphi)\,] =
%\mathbb{E}\,
%[\,\sum_{i=1}^N \overline{w}_{d}^i\,k_1(\varphi)(X_{d,1}^i)-\pi(\varphi)\,]
%$$
%where $\overline{w}_{d}^{i}\equiv
%\overline{w}_{d}(x^i_{l_d(t_{m^*-1}(d)):d-1})$ are used as a
%short-hand for the normalized weights. Then one follows the above argument.
\end{proof}

\subsection{Stochastic Times}\label{appendix:stoch_times}

\begin{proof}[Proof of Theorem \ref{theo:stoch_times}]
Our proof will keep $d$ fixed until the point at which we can apply
Theorem~\ref{theorem:limit_adaptive}. Conditionally on the chosen $\{a_k\}$ we have:
\begin{align*}
\mathbb{P}\,&[\,
\Omega\setminus\Omega_{d}^N\,] \leq  
\sum_{k=1}^{m^*(\delta)}\sum_{s \in G_{\delta}\cap [\,t_{k-1}^\delta(d),t_{k}^\delta(d)\,]} \mathbb{P}\,\big[\,\big|\tfrac{1}{N}\,\textrm{ESS}_{(t_{k-1}^{\delta}(d),s)}(N)-\textrm{ESS}_{(t_{k-1}^{\delta}(d),s)}\big| \geq \upsilon \big|\textrm{ESS}_{(t_{k-1}^{\delta}(d),s)}-
a_{k}\big|\,\big]\ .
\end{align*}
Define
$$
\epsilon(d) := \inf_{n}\inf_{s}|\,\textrm{ESS}_{(t_{k-1}^{\delta}(d),s)}-
a_{k}\,| \ ;
$$
 we remark $\lim_{d\rightarrow\infty}\epsilon(d)=\epsilon>0$ (with probability one).
Hence we have:
\begin{align*}
\mathbb{P}\,[\,\Omega\setminus\Omega_{d}^{N}\,] \leq 
\sum_{k=1}^{m^*(\delta)}\sum_{s\in G_{\delta} [\,t_{k-1}^\delta(d),t_{k}^\delta(d)\,]} \mathbb{P}\,\big[\,\big|\tfrac{1}{N}\,\textrm{ESS}_{(t_{k-1}^{\delta}(d),s)}(N)-\textrm{ESS}_{(t_{k-1}^{\delta}(d),s)}\big| \geq \upsilon \epsilon(d)\big|\,\big]\ .
\end{align*}
Application of the Markov inequality yields that:
$$
\mathbb{P}\,[\,
\Omega\setminus\Omega_{d}^N\,] \leq
\frac{m^*(\delta)\delta}{\upsilon \epsilon(d)}
\max_{k,s} \mathbb{E}\,\big[\,\big|\tfrac{1}{N}\,\textrm{ESS}_{(t_{k-1}^{\delta}(d),s)}(N)
-\textrm{ESS}_{(t_{k-1}^{\delta}(d),s)}\big|\,\big]\ .
$$
Since $k,s$ lie in a finite set and $\epsilon>0$, we need only deal with the expectation as $d$ grows.  Note, in the expectation, the case $s=t_k^\delta(d)$ is not of interest; ESS is constant and hence lower-bounded all other cases.

Application of Theorem \ref{theorem:limit_adaptive} now yields:
\begin{equation*}
\lim_{d\rightarrow\infty} \mathbb{E}\,\big[\,\big|\tfrac{1}{N}\textrm{ESS}_{(t_{k-1}^{\delta}(d)),s)}(N)-
\textrm{ESS}_{(t_{k-1}^{\delta}(d),s)}\, \big]
=
\mathbb{E}\,\big[\,\big|\tfrac{1}{N}\,\textrm{ESS}_{(t_{k-1}^{\delta},s)}(N)-
\textrm{ESS}_{(t_{k-1}^{\delta},s)}\big|\,\big]
\end{equation*}
where
\begin{equation*}
\textrm{ESS}_{(t_{k-1}^{\delta},s)}(N) =   \frac{(\sum_{j=1}^N \exp\{X_j^{k}\})^2}{
\sum_{j=1}^N \exp\{2X_j^{k}\}}\ ;\quad
\textrm{ESS}_{(t_{k-1}^{\delta},s)} =  \exp\big\{-\sigma^2_{t_{k-1}^{\delta}:s}\big\}\ ,
\end{equation*}
with $X_j^{k}\stackrel{\textrm{i.i.d.}}{\sim}\mathcal{N}(0,\sigma^2_{t_{k-1}^{\delta}:s})$.
We set:
$$
\alpha_j^k=\exp\{X_j^{k}\}\ ;\quad
\beta_j^k=\exp\{2X_j^{k}\}\ ;\quad
\alpha^k=\exp\{\tfrac{1}{2}\,\sigma^2_{t_{k-1}^{\delta}:s}\}\ ;\quad
\beta^k=\exp\{2\sigma^2_{t_{k-1}^{\delta}:s}\}\ .
$$
Then, we are to bound:
$$
\mathbb{E}\,\bigg[\,\bigg|\frac{(\frac{1}{N}\sum_{j=1}^N\alpha_j^k)^2}{\frac{1}{N}\sum_{j=1}^N\beta_j^k}
-\frac{(\alpha^k)^2}{\beta_k}\bigg|\,\bigg]\ .
$$
We have the decomposition
$$
\frac{(\frac{1}{N}\sum_{j=1}^N\alpha_j^k)^2}{\frac{1}{N}\sum_{j=1}^N\beta_j^k}
-\frac{(\alpha^k)^2}{\beta_k} =
\bigg(\frac{(\frac{1}{N}\sum_{j=1}^N\alpha_j^k)^2}{\beta^k\frac{1}{N}\sum_{j=1}^N\beta_j^k}\bigg)
\bigg[\beta^k-\frac{1}{N}\sum_{j=1}^N\beta_j^k\bigg]
+ \frac{1}{\beta^k}\bigg[(\frac{1}{N}\sum_{j=1}^N\alpha_j^k)^2-(\alpha^k)^2\bigg].
$$
For the first term of the R.H.S. in the above equation, as ESS divided by $N$ is upper-bounded by 1, we can use Jensen and the Marcinkiewicz-Zygmund inequality. For the second term, via the  relation $x^2-y^2=(x+y)(x-y)$ and Cauchy-Schw\"artz, one can use the same inequality to conclude that for some finite $M(k,\delta,s)$:
$$
\mathbb{E}\,\big[\,\big|\tfrac{1}{N}\,\textrm{ESS}_{(t_{k-1}^{\delta},s)}(N)-
\textrm{ESS}_{(t_{k-1}^{\delta},s)}\big|\,\big] \leq \frac{M(k,\delta,s)}{\sqrt{N}}\ .
$$
Thus, we have proven that:
$
\lim_{d\rightarrow\infty}\mathbb{P}\,[\,\Omega\setminus\Omega_{d}^N\,]
\leq \frac{M(m^*(\delta))}{\sqrt{N}}
$
as required.
\end{proof}

\section{Verifying the Assumptions}\label{app:verify}

\begin{proof}[Proof of Proposition \ref{prop:verify}]

We start with (A\ref{hyp:A})(i)-(ii); to establish uniform (in $s$) drift and minorization conditions for the kernel $k_s$. The proof is standard and included for completeness.

It is first noted that, for any $\delta_q>0$, if $|x-y|<\delta_q$:
\begin{equation}
q_s(x,y) \geq \frac{\phi_0^{1/2}}{\sqrt{2\pi}}\exp\bigg\{-\frac{s}{2}\delta_q^2\bigg\} \geq \frac{\phi_0^{1/2}}{\sqrt{2\pi}}\exp\bigg\{-\frac{1}{2}\delta_q^2\bigg\}\ .
\label{eq:q_lower_bound}
\end{equation}
This property will be used below.
To establish the minorization, one can follow the proof of Theorem 2.2 of \cite{roberts3} to show that for any $x$,  with $y\in B(x,\delta_q/2)$ (the open ball, centered $x$ and of radius $\delta_{q/2}$), $A\in\mathscr{B}(\mathbb{R})$,
$A\subseteq B(x,\delta_q/2)$
$$
k_s(y,A)  \geq  \eta(x,\delta_q/2) \int_A (q_s(z,y) \wedge q_s(y,z)) dz
 \geq  \eta(x,\delta_q/2) \epsilon_q \int_A dz
$$
where
$
\eta(x,\delta_q/2) = \inf_{x\in B(x,\delta_q/2)} \pi_1(x)/\|\pi_{\phi_0}\|_{\infty}
$
and $\delta_q$ is as \eqref{eq:q_lower_bound}, $\epsilon_q$ as the RHS of the inequality in \eqref{eq:q_lower_bound}.
Hence, we have the uniform minorization condition.

To prove the drift, we do not require it hold for $s=\phi_0$ as, in the algorithm,
we sample exactly from $\pi_{\phi_0}$. None-the-less, by our assumptions there exist a drift condition for $k_{\phi_0}$ (a symmetric normal random walk Metropolis-kernel of invariant $\pi_{\phi_0}$); write the parameters $\lambda$, $b$. Now, for any
$s\in(\phi_0,1]$, via Lemma 5 of \cite{andrieu2}
and using that for any
$x,y$,
$
\frac{q_s(x,y)}{q_{\phi_0}(x,y)} \leq \frac{1}{\sqrt{\phi_0}}
$
one has
$$
k_s(V)(x) \leq \frac{1}{\sqrt{\phi_0}} (k_{\phi_0}(V)(x)-V(x)) + V(x)
$$
where
\begin{equation}
V(x)=\|e^{\phi_0g}\|_{\infty}^{1/2}/e^{\frac{\phi_0}{2}g(x)}
\label{eq:lyapunov_function}.
\end{equation}
Now one can easily find a $\bar{c}\in[(1-\phi_0^{-1/2})\wedge(-\lambda/\sqrt{\phi_0}),1-\lambda\phi_0^{-1/2}]$ such that
$
k_s(V)(x) \leq \widetilde{\lambda} V(x) + \widetilde{b}\,\mathbb{I}_C(x)
$
with $\widetilde{\lambda}\in(0,1)$, $\widetilde{b}<\infty$. Hence, the uniform drift condition is verified.
(A\ref{hyp:A}) (iii) can be verified in a similar manner to e.g.~\cite{doucchains}  and is omitted.

Now to (A\ref{hyp:B}), which is a little more complex. Recall, we want to establish that there exist an $M<\infty$ such that
for any $s,t\in(\phi_0,1]$,
$
|||k_s-k_t|||_V  \leq M |s-t|.
$
For simplicity, we will consider only the increment of proposal (via change of variables), so $q_s$ is a zero mean normal density, with variance $1/s$. For any fixed $x\in \mathbb{R}$ $q_s$ is a bounded-continuous function of $s\in[\phi_0,1]$ and further, the first derivative w.r.t.~$s$ is upper-bounded
by
$
\frac{1}{2\sqrt{2\pi\phi_0}}e^{-\phi_0 x^2/2}
$
hence it follows that for any $x\in\mathbb{R}$, $s,t\in[\phi_0,1]$:
\begin{equation}
|q_s(x)-q_t(x)| \leq \bigg(\frac{1}{2\sqrt{2\pi\phi_0}}e^{-\phi_0 x^2/2}\bigg) |s-t| \ .
\label{eq:q_is_lipschitz}
\end{equation}
Now central to our proof is the consideration of the acceptance probability, which is
$
\alpha_s(x,z) = 1\wedge \exp\{s(g(x+z)-g(x))\} \ .
$
Let
\begin{equation}
A(x) = \{z:g(x+z)-g(x)>0\}\label{eq:a(x)_def}
\end{equation}
then if $z\in A(x)$, $\alpha_s(x,z)=1$.
We begin by considering the acceptance part of the kernel. The difficult issue is when $z\in A(x)^c$ which is dealt with now:
\begin{equation}
\int_{A(x)^c} \varphi(x+z) \exp\{-sG(x,z)\}q_s(z)dz-
\int_{A(x)^c} \varphi(x+z) \exp\{-tG(x,z)\}q_t(z)dz
\label{eq:cont_kernel}
\end{equation}
where $\varphi\in\mathscr{L}_V$.
Now for any fixed $x$, $z\in A(x)^c$ one has that
\begin{equation}
|\exp\{-sG(x,z)\}-\exp\{-tG(x,z)\}|
\leq \bigg(G(x,z)e^{-\frac{\phi_0}{2}G(x,z)}\bigg)|s-t|\label{eq:acc_prob_cont}
\end{equation}
for every $s,t\in[\phi_0,1]$.
Then, returning to \eqref{eq:cont_kernel}, it can be decomposed into the sum of
\begin{equation}
\int_{A(x)^c} \varphi(x+z)[\exp\{-sG(x,z)\}-\exp\{-tG(x,z)\}]q_s(z)dz
\label{eq:cont_kernel1}
\end{equation}
and
\begin{equation}
\int_{A(x)^c} \varphi(x+z)\exp\{-tG(x,z)\}[q_s(z)-q_t(z)]dz\ .
\label{eq:cont_kernel2}
\end{equation}
First consider \eqref{eq:cont_kernel1}. 
Applying  \eqref{eq:acc_prob_cont}, it follows that \eqref{eq:cont_kernel1} is upper-bounded by
$$
C_{\phi_0}|\varphi|_{V}|s-t|\int_{A(x)^c}
e^{-\frac{\phi_0}{2}g(x+z)} G(x,z)e^{-\frac{\phi_0}{2}G(x,z)}q_s(z)dz
$$
where $C_{\phi_0}$ is associated to the Lyapunov function \eqref{eq:lyapunov_function}.
Now as
$
e^{-\frac{\phi_0}{2}g(x+z)}e^{-\frac{\phi_0}{2}G(x,z)} = e^{-\frac{\phi_0}{2}g(x)}
$
which is controlled by $V(x)$ and by assumption \eqref{eq:assumption_a2}, $\int_{A(x)^c} G(x,z) q_s(z)dz$ is dealt with; hence
\eqref{eq:cont_kernel1} divided by $V(x)$ is upper-bounded by
$
C_{\phi_0}|\varphi|_{V}|s-t|C^*.
$
Our next task is \eqref{eq:cont_kernel2}. Applying \eqref{eq:q_is_lipschitz}, it is upper-bounded by
\begin{align*}
C_{\phi_0}|\varphi|_{V}&\int_{A(x)^c}e^{-\frac{\phi_0}{2}g(x+z)}
e^{-tG(x,z)}|s-t|
\frac{1}{2\sqrt{2\pi\phi_0}}e^{-\phi_0 z^2/2}
dz =\\
&C_{\phi_0}|\varphi|_{V}e^{-tg(x)}\int_{A(x)^c}e^{(t-\frac{\phi_0}{2})g(x+z)}
|s-t|
\frac{1}{2\sqrt{2\pi\phi_0}}e^{-\phi_0 z^2/2}
dz\ ;
\end{align*}
on dividing by the Lyapunov function, we are to deal with the expression
$
\exp\{-(t-\frac{\phi_0}{2})G(x,z)\}.
$
Now, $t>\phi_0/2$ and for any $x$, $z\in A(x)^c$, one has that $G(x,z)>0$ hence this latter expression is upper-bounded by 1. This leaves the term
$
\int_{A(x)^c}\frac{1}{2\sqrt{2\pi\phi_0}}e^{-\phi_0 z^2/2}dz
$
which is finite. Hence, putting together the above arguments, we have shown that there exists an $M<\infty$ such that for any $s,t\in[\phi_0,1]$, $x\in\mathbb{R}$ one has
$$
\big|\int_{A(x)^c} \varphi(x+z) e^{-sG(x,z)}q_s(z)dz-
\int_{A(x)^c} \varphi(x+z) e^{-tG(x,z)}q_t(z)dz\big|\,/\, V(x)
\leq M |s-t|\ ,
$$
where we have applied \eqref{eq:q_is_lipschitz}.

Turning to the acceptance part of the kernel on $A(x)$, we have
$$
\int_{A(x)}\varphi(x+z)[q_s(z)-q_t(z)]dz \leq C_{\phi_0}|\varphi|_V \int_{A(x)}
V(x+z) |s-t|\frac{1}{2\sqrt{2\pi\phi_0}}e^{-\phi_0 z^2/2}dz\ .
$$
As $V(x+z)\leq V(x)$ on $A(x)$, it follows that the term of interest is upper-bounded by
$
M |s-t| V(x)
$
for some $M<\infty$.
Hence the acceptance part of the kernel, divided by $V$, is upper bounded by $M |s-t|$.
In the rejection part of the kernel, we have to control:
$$
\varphi(x)\bigg[
\int_{A(x)^c}[\alpha_t(x,z)-\alpha_t(x,z)]q_t(z)dz
+
\int_{A(x)^c}[q_t(z)-q_s(z)]\alpha_s(x,z)dz
+
\int_{A(x)}[q_t(z)-q_s(z)]dz
\bigg]\ .
$$
Now, as $\varphi$ is controlled by $V$, we need to consider the continuity of the terms in the bracket. The latter two terms, via \eqref{eq:q_is_lipschitz}, are upper-bounded by $M|s-t|$. The first term is upper-bounded by
$
|s-t|\int_{A(x)^c} G(x,z) e^{-\phi_0/2 G(x,z)}q_t(z)dz
$
using \eqref{eq:acc_prob_cont}. As $e^{-\phi_0/2 G(x,z)}\leq 1$, we can use
\eqref{eq:assumption_a2} to complete the argument.
\end{proof}

%\bibliographystyle{plain}
%\bibliography{references}

\begin{thebibliography}{99}

\bibitem{andrieu1}
{\sc Andrieu,} C. \& {\sc Moulines} \'E.~(2006). On the ergodicity
properties of some adaptive MCMC algorithms. {\it Ann. Appl.
Prob.}, \textbf{16},
1462--1505.

\bibitem{andrieu_adap_rev}
{\sc Andrieu,} C. \& {\sc Thoms} J.~(2008).
A tutorial on adaptive MCMC. {\it Statist. Comp.}, \textbf{18}, 343--373.

\bibitem{andrieu2}
{\sc Andrieu,} C., {\sc Breyer}, L. \& {\sc Doucet}, A.~(2001). Convergence of simulated annealing using Foster-Lyapunov criteria. {\it J. Appl.
Prob.}, \textbf{4},
975--994.

\bibitem{andr:10}
{\sc Andrieu,} C., {\sc Jasra}, A., {\sc Doucet}, A. \& {\sc Del Moral} P.~(2011). On non-linear Markov chain Monte Carlo. {\it Bernoulli},
{\bf 17},  987--1014.


\bibitem{atch:09}
{\sc Atchad\'e}, Y.~(2009). A strong law of large numbers for martingale arrays. Technical Report, University of Michigan.

\bibitem{baer}
{\sc Baehr}, C, \& {\sc Pannekoucke}, O.~(2010).
Some issues and results on the EnKF and particle filters for meteorological models. In \emph{Chaotic Systems: Theory and Applications}, C. H. Skiadas \& I. Dimotikalis Eds, 27--24, World Scientific: London.

\bibitem{bedard}
{\sc B\'edard}, M.~(2007).
Weak convergence of Metropolis algorithms for non-i.i.d.~target distributions. {\it Ann. Appl. Prob.}, \textbf{17}, 1222--1244.

\bibitem{berger}
{\sc Berger}, E.~(1986).
Asymptotic behaviour of a class of stochastic approximation procedures. {\it Probab. Theory. Rel. Fields}, {\bf 71}, 517--552.

\bibitem{bengtsson}
{\sc Bengtsson}, T., {\sc Bickel}, P., \& {\sc Li}, B.~(2008).
Curse-of-dimensionality revisited: Collapse of the particle filter in very large scale systems.
In \emph{Essays in Honor of David A. Freeman},
D. Nolan \& T. Speed, Eds, 316--334, IMS.



\bibitem{beskos1}
{\sc Beskos}, A., {\sc Crisan}, D., {\sc Jasra}, A. \& {\sc Whiteley}, N.~(2011).
Error bounds and normalizing constants for sequential Monte carlo. Technical Report, Imperial College London.

\bibitem{beskos}
{\sc Beskos}, A., {\sc Pillai}, N., {\sc Roberts}, G. O., {\sc Sanz-Serna}, J. M.~\& {\sc Stuart} A.~(2012). Optimal tuning of the hybrid Monte Carlo algorithm. \emph{Bernoulli}, (to appear).

\bibitem{besk}
{\sc Beskos}, A., {\sc Roberts}, G. \& {\sc Stuart}, A.~(2009).
Optimal scalings for local Metropolis--Hastings chains on nonproduct targets in high dimensions.
{\it Ann. Appl. Probab.}, \textbf{19}, 863--898.

\bibitem{besko}
{\sc Beskos}, A. \& {\sc Stuart}, A.~(2009).
MCMC methods for sampling function space.
In \emph{Proceedings of the International Congress of Industrial and Applied Mathematics}, Zurich, 2007.

\bibitem{bickel}
{\sc Bickel}, P., {\sc Li}, B. \& {\sc Bengtsson}, T.~(2008).
Sharp failure rates for the bootstrap particle filter in high dimensions.
In \emph{Pushing the Limits of Contemporary Statistics},
B. Clarke \& S. Ghosal, Eds, 318--329, IMS.

\bibitem{bill:99}
{\sc Billingsley}, P.~(1999).\emph{Convergence of Probability Measures}. 2nd edition. Wiley:New York.

\bibitem{brey:04}
{\sc Breyer}, L., {\sc Piccioni}, M. \& {\sc Scarlatti}, S.~(2004).
Optimal scaling of MALA for non-linear regression. {\it Ann. Appl. Probab.}, \textbf{14}, 1479--1505.

\bibitem{brey:00}
{\sc Breyer}, L. \& {\sc Roberts}, G.~(2000).
From Metropolis to diffusions: Gibbs states and optimal scaling. {\it Stochastic Process. Appl.}, \textbf{90}, 181--206.

\bibitem{cappe}
{\sc Capp\'e}, O, {\sc Gullin}, A., {\sc Marin}, J. M. \& {\sc Robert}, C. P.~(2004). Population Monte Carlo. \emph{J. Comp. Graph. Stat.},  {\bf 13}, 907--925.

\bibitem{cerou1}
{\sc C\'erou}, F., {\sc Del Moral}, P. \& {\sc Guyader}, A.~(2011).
A non-asymptotic variance theorem for un-normalized Feynman-Kac particle models. \emph{Ann. Inst. Henri Poincare}, {\bf 47}.

\bibitem{chopin1}
{\sc Chopin}, N. (2002). A
sequential particle filter for static models. \emph{Biometrika}, \textbf{89}, 539--552.


\bibitem{chopin}
{\sc Chopin}, N. (2004). Central limit theorem and its
application to Bayesian inference. \emph{Ann
Statist.}, \textbf{32}, 2385--2411.

%\bibitem{chopin3}
%{\sc Chopin}, N.~(2007). Inference and model choice for sequentially ordered hidden Markov models. \emph{J.~R.~Statist.} \emph{Soc. B}, {\bf 69}, 269--284.

\bibitem{chopin2}
{\sc Chopin}, N.~\& {\sc Jacob}, P.~(2011).
Free energy Sequential Monte Carlo, application to mixture modelling.
\emph{Bayesian Statistics 9}, 91-118.

\bibitem{chopin_smc}
{\sc Chopin}, N., {\sc Jacob}, P. \& {\sc Papaspiliopoulos}, O.~(2011).
SM$\textrm{C}^2$: A sequential Monte Carlo algorithm with particle Markov chain Monte Carlo updates. Technical Report, ENSAE.

%\bibitem{cotter}
%{\sc Cotter}, S. L., {\sc Dashti}, M. \& {\sc Stuart}, A. M.~(2010).
%Approximation of Bayesian inverse problems for PDEs. \emph{SIAM J. Numer. Anal.},
%{\bf 48}, 322--345.

\bibitem{crisroz}
{\sc Crisan}, D.~\& {\sc Rozovsky}, B.~(2011).
\emph{The Oxford Handbook of Nonlinear Filtering}
OUP:Oxford.

\bibitem{dean}
{\sc Dean}, T. \& {\sc Dupuis}, P.~(2009).
Splitting for rare event simulation: A large deviations approach to design and analysis. {\it Stoch. Proc. Appl.}, {\bf 119}, 562--587.

\bibitem{delmoral}
{\sc Del Moral}, P.~(2004). \textit{Feynman-Kac Formulae: Genealogical and
Interacting Particle Systems with Applications}. Springer: New York.

\bibitem{delm:06}
{\sc Del Moral,} P., {\sc Doucet}, A. \& {\sc Jasra}, A.~(2006).
Sequential Monte Carlo samplers.
\emph{J.~R.~Statist.} \emph{Soc. B}, {\bf 68}, 411--436.

\bibitem{delmoralabc}
{\sc Del Moral,} P., {\sc Doucet}, A. \& {\sc Jasra}, A.~(2012).
An adaptive sequential Monte Carlo method for approximate Bayesian computation.
\emph{Statist. Comp.} (to appear).


\bibitem{delmoral_resampling}
{\sc Del Moral}, P., {\sc Doucet}, A. \& {\sc Jasra}, A.~(2012).
On adaptive resampling procedures for sequential Monte Carlo methods. {\it Bernoulli},
{\bf 18}, 252-278.

\bibitem{delmoral_coal}
{\sc Del Moral,} P., {\sc Patras}, F. \& {\sc Rubenthaler}, S.~(2009).
Coalescent tree based functional representations for some
Feynman-Kac particle models.
\emph{Ann. Appl. Probab.}, {\bf 19}, 778--825.


\bibitem{douc}
{\sc Douc}, R. \& {\sc Moulines}, \'E.~(2008).
Limit theorems for weighted samples with applications to sequential Monte Carlo methods. {\it Ann. Statist.}, \textbf{36}, 2344--2376.

\bibitem{doucchains}
{\sc Douc}, R., {\sc Moulines}, \'E. \& {\sc Rosenthal}, J. S.~(2004).
Quantitative bounds on convergence of time-inhomogeneous Markov Chains
\emph{Ann. Appl. Probab}., {\bf 14}, 1643--1665.

\bibitem{doucet}
{\sc Doucet,} A., {\sc De Freitas}, J.~F.~G. \& {\sc Gordon},
N. J.~(2001). \emph{Sequential Monte Carlo Methods in Practice}. Springer:
New York.

\bibitem{gelman}
{\sc Gelman}, A. \& {\sc Meng}, X.~L..~(1998).
Simulating normalizing constants: From importance sampling to bridge sampling to path sampling. {\it Statist. Sci.}, \textbf{13}, 163--185.

\bibitem{hall:80}
{\sc Hall}, P. \& {\sc Heyde}, C.~(1980). \emph{Martingale Limit Theory and its Application}. Academic Press: New York.

%\bibitem{hairer}
%{\sc Hairer}, M., {\sc Stuart}, A. M.~\& {\sc Vollmer}, S.J.~(2011).
%Spectral gaps for Metropolis-Hastings algorithms in infinite dimensions. Technical Report, University of Warwick

\bibitem{heine}
{\sc Heine, K. \& Crisan, D.} (2008) \emph{Uniform approximations of discrete-time filters}, \emph{Adv. Appl. Probab.} {\bf 40}, 979--1001.

\bibitem{jarner}
{\sc Jarner}, S. F. \& {\sc Hansen} E.~(2000).
Geometric ergodicity of Metropolis algorithms. {\it Stoch. Proc. Appl.},
{\bf 85}, 341--361.

\bibitem{jarzynski}
{\sc Jarzynski}, C., 1997. Nonequilibrium equality for free energy differences. \emph{Phys. Rev. Lett.}, {\bf 78}, 2690--2693.

\bibitem{jasra_adaptive}
{\sc Jasra}, A., {\sc Stephens,} D. A., {\sc Doucet}, A. \& {\sc Tsagaris}, T.~(2011).
Inference for L\'evy driven stochastic volatility models via adaptive sequential Monte Carlo. {\it Scand. J. Statist.}, {\bf 38}, 1--22.

\bibitem{jasra_static}
{\sc Jasra}, A., {\sc Stephens,} D. A., \& {\sc Holmes}, C. C.~(2007).
On population-based simulation for static inference. {\it Statist. Comp.}, {\bf 17}, 263--269 .


\bibitem{kong}
{\sc Kong}, A., {\sc Liu}, J. S. \& {\sc Wong}, W. H.~(1994). Sequential imputations and Bayesian missing data problems. \emph{J. Amer. Statist. Assoc.}, \textbf{89}, 278--288.

\bibitem{kunsch}
{\sc K\"unsch}, H.~(2005). Recursive Monte Carlo filters: Algorithms and theoretical analysis. \emph{Ann. Statist.}, \textbf{33}, 1983--2021.


\bibitem{lee_gpu}
{\sc Lee}, A., {\sc Yau}, C., {\sc Giles}, M., {\sc Doucet}, A. \& {\sc Holmes}, C. C.~(2010). On the utility of graphics cards to perform massively parallel implementation of advanced Monte Carlo methods. \emph{J. Comp. Graph. Statist.}, {\bf 19}, 769--789.

\bibitem{liu}
{\sc Liu}, J. S. (2001). \emph{Monte Carlo Strategies in Scientific Computing}.
Springer: New York.

\bibitem{andrew:1}
{\sc Mattingly,} J., {\sc Pillai}, N. \& {\sc Stuart}, A.~(2011).
Diffusion limits of the Random walk Metropolis algorithm in High Dimensions.
\emph{Ann. Appl. Probab.} (to appear).

\bibitem{mcle:74}
{\sc McLeish}, D. L.~(1974). Dependent central limit theorems and invariance principles. \emph{Ann. Probab}, {\bf 2}, 620--628.

\bibitem{meyn}
{\sc Meyn}, S. \& {\sc Tweedie}, R. L.~(2009). \emph{Markov Chains and Stochastic Stability}. 2nd edition, Cambridge: CUP.

\bibitem{neal:01}
{\sc Neal}, R. M. (2001). Annealed importance sampling. \emph{Statist. Comp.}, \textbf{11}, 125--139.

\bibitem{andrew:2}
{\sc Pillai,} N., {\sc Stuart}, A. \& {\sc Thiery}, A.~(2011).
Optimal scaling and diffusion limits for the Langevin algorithm in high dimensions.
\emph{Ann. Appl. Probab.} (to appear).

\bibitem{roberts3}
{\sc Roberts}, G. O. \& {\sc Tweedie} R. L.~(1996).
Geometric convergence and central limit theorems for multidimensional
Hastings and Metropolis algorithms. {\it Biometrika},
{\bf 83}, 95--110.

\bibitem{roberts1}
{\sc Roberts}, G. O., {\sc Gelman}, A.  \& {\sc Gilks} W. R.~(1997).
Weak convergence and
optimal scaling of random walk Metropolis algorithms. {\it Ann. Appl. Probab.},
{\bf 7}, 110--120.

\bibitem{roberts2}
{\sc Roberts}, G.O. \& {\sc Rosenthal}, J.S.~(1998). Optimal scaling of discrete approximations to Langevin diffusions.
\emph{J. R. Statist. Soc. Ser. B}, {\bf 60}, 255-268.

\bibitem{rudin}
{\sc Rudin}, W. (1976). \emph{Principles of Mathematical Analysis}.
International Series in Pure and Applied Mathematics, McGraw-Hill Book Co.: New York.

\bibitem{shaefer}
{\sc Sch\"afer}, C. \& {\sc Chopin}, N.~(2012).
Adaptive Monte Carlo on binary sampling spaces. \emph{Statist. Comp.}, (to appear).

\bibitem{shiryaev}
{\sc Shiryaev}, A. (1996). {\it
Probability}, Springer: New York.

\bibitem{snyder}
{\sc Snyder}, C., {\sc Bengtsson}, T., {\sc Bickel}, P., \& {\sc Anderson}, J.~(2008). Obstacles to high-dimensional particle filtering. \emph{Month. Weather Rev.}, {\bf 136}, 4629--4640.

\bibitem{whiteley}
{\sc Whiteley}, N.~(2012). Sequential Monte Carlo samplers: Error bounds and insensitivity to initial conditions. \emph{Stoch. Anal. Appl.} (to appear).

\bibitem{whiteley2}
{\sc Whiteley}, N.~(2011). Stability properties of some particle filters. Technical Report, University of Bristol.


\bibitem{whiteley1}
{\sc Whiteley}, N., {\sc Kantas}, N, \& {\sc Jasra}, A.~(2012). Linear variance bounds for particle approximations of time homogeneous Feynman-Kac formulae.
\emph{Stoch. Proc. Appl.}, {\bf 122}, 1840--1865.

\bibitem{with:81}
{\sc Withers}, C.~(1981). Central limit theorems for dependent variables I. \emph{Z. Wahrsch. Verw. Gebiete}, {\bf 57}, 509--534.

\end{thebibliography}

\begin{small}

\end{small}

\end{document}